\documentclass[11pt]{article}%
\usepackage{amssymb}
\usepackage{amsfonts}
\usepackage{amsmath}
\usepackage{amssymb}
\usepackage{color}
\usepackage{graphicx}
\usepackage{caption}%
\usepackage{hyperref}
\usepackage{float}
\usepackage{subcaption}

\usepackage{tabularx}
\usepackage{algorithm}
\usepackage{algpseudocode}
\providecommand{\U}[1]{\protect\rule{.1in}{.1in}}
\newtheorem{theorem}{Result}

\newtheorem{remark}[theorem]{Remark}

\newenvironment{proof}[1][Proof]{\noindent \textbf{#1.} }{\  \rule{0.5em}{0.5em}}
\begin{document}
	
	\title{\textbf{Robust inference for intermittently-monitored step-stress tests under Weibull lifetime distributions}}
	\author{Narayanaswamy Balakrishnan$^1$,  María Jaenada$^2$ and Leandro Pardo$^2$}
	\date{$^1$McMaster University, $^2$Complutense University of Madrid}
	%\date{ }
	\maketitle

\begin{abstract}
	%One-shot devices are products that can be used only once, so they get destroyed when tested. However, the destructiveness assumption may not be necessary in many practical applications such as assessing the effect of temperature on some electronic components. 
	Many modern products exhibit high reliability under normal operating conditions. Conducting life tests under these conditions may result in very few observed failures, insufficient for accurate inferences. Instead, accelerated life tests (ALTs) must be performed. 
	%Further, one-shot devices generally have large mean lifetime to failure, and so 
	%
	One of the most popular ALT designs is the step-stress test, which shortens the product's lifetime by progressively increasing the stress level at which units are subjected to at some pre-specified times.
	%the non-destructive devices are tested at certain inspection times and  surviving units can continue within the experiment providing extra information. 
	Classical estimation methods based on the maximum likelihood estimator (MLE) enjoy suitable asymptotic properties but they lack robustness. That is, data contaminationcan significantly impact the statistical analysis.
	In this paper, we develop robust inferential methods for highly reliable devices based on the density power divergence (DPD) for estimating and testing under the step-stress model with intermittent monitoring and Weibull lifetime distributions. We theoretically and empirically examine asymptotic and robustness properties of the minimum DPD estimators and associated Wald-type test statistics. Moreover, we develop robust estimators and confidence intervals for some important lifetime characteristics.
	%, namely reliability at certain mission times, distribution quantiles and mean lifetime of a device. 
	The effect of temperature 
	%in three electronic components, 
	in solar lights, medium power silicon bipolar transistors and LED lights using real data arising from an step-stress ALT is analyzed applying the robust methods proposed.
\end{abstract}

%	Many modern products exhibit high reliability under normal operating conditions. Conducting life tests under these conditions may result in very few observed failures, insufficient for accurate inferences. Instead, accelerated life tests (ALTs) must be performed. One of the most popular ALT designs is the step-stress test, which shortens the product's lifetime by progressively increasing the stress level at which units are subjected to at some pre-specified times. Classical estimation methods based on the maximum likelihood estimator (MLE) enjoy suitable asymptotic properties but they lack robustness. That is, data contaminationcan significantly impact the statistical analysis. In this paper, we develop robust inferential methods for \textcolor{blue}{highly reliable non-destructive devices} based on the density power divergence (DPD) for estimating and testing under the step-stress model with \textcolor{blue}{ intermittent monitoring} and Weibull lifetime distributions. We theoretically and empirically examine asymptotic and robustness properties of the minimum DPD estimators and associated Wald-type test statistics. Moreover, we develop robust estimators and confidence intervals for some important lifetime characteristics. The effect of temperature in solar lights, medium power silicon bipolar transistors and LED lights using real data arising from an step-stress ALT is analyzed applying the robust methods proposed.

\section{Introduction}

Nowadays products often have large mean times to failure under normal operating conditions. Therefore,
accelerated life tests (ALTs)
%,  which shorten the life of the products by increasing the stress level at which the devices are subjected,  
need to be conducted to infer their lifetime distribution. The objective of these accelerated tests is to induce the devices to fail in a shorter time by increasing the stress level at which the devices are subjected, so as to efficiently collect useful and sufficient information for analyzing the lifetimes of the products depending on the stress level.
After suitable inference is developed, results can be extrapolated to normal operating conditions; see Meeter and Mecker (1994) and Meeker et al. (1998).
Once the lifetime distribution is estimated, some lifetime characteristics of interest, such as reliability of the device at certain mission times, distribution quantiles or mean lifetime of the product,  may be determined.

Accelerated life-testing may be performed either at constant stress or linearly increasing stress levels %, step-by-step 
over time.  
Multiple step-stress ALTs increase the stress level at which devices are tested at pre-specified times, referred in the following as times of stress change. The step-stress design requires less samples for inference and is more efficient than constant-stress ALTs under the optimal situations
(see Han and Ng (2013)). This advantage in terms of experimental sample typically reduces the experiment cost.
 
For step-stress designs a model relating the effect of increased stress levels on the lifetime distribution of a device is needed.
Three main proposals have been made in the literature: tampered random variable model (DeGroot and  Goel (1979)), tampered failure-rate model (Bhattacharyya and Soejoeti (1989)), and cumulative exposure model (Nelson (1980)). 
We will adopt the cumulative exposure approach here, which relates the lifetime distribution of experimental devices at one stress level to the distributions at preceding stress levels by assuming that the residual life of the experimental devices depends only  on the current cumulative exposure experienced and current stress, with no memory of how the exposure was accumulated.
%Therefore, the lifetime of surviving devices is modelled by to the cumulative distribution for the current stress, 
The cumulative exposure model is widely used in the literature due to its natural assumption regarding changes in the distribution and mathematical simplicity; the distribution at the new stress level is adjusted (shifted) by the cumulative wear and tear experienced so far, but behaves as if it had always been exposed to this new level.
Moreover, there is substantial empirical evidence supporting the effectiveness of the model, and practical studies have demonstrated its accuracy and reliability in predicting product lifetimes under accelerated testing conditions.

%27. Miller, R.W. and Nelson, W.B. (1983). Optimum simple step-stress plans for accelerated life testing. IEEE Transactions on Reliability 32: 59–65.
%30. Nelson, W.B. (1980). Accelerated life testing—step-stress models and data analysis. IEEE Transactions on Reliability 29: 103–108.
%31. Nelson, W.B. (1990). Accelerated testing—statistical models, test plans, and data analyses. New York: Wiley
%Miller and Nelson (1983) studied the optimal simple step-stress ALT plans with the cumulative exposure model under exponential distributions and  Bai et al. (1989) extended the results  to censoring schemes. 
Balakrishnan (2009) reviewed likelihood-based inferential techniques for step-stress models and associated optimal accelerated life-tests, based on complete and censored samples under exponential distributions. 
%	More recently, Balakrishnan et al. (2023, 2024a, 2024b) have studied robust inference for interval-censored step-stress ALTs under exponential, gamma and log-normal distributions.
In the present work,
%Bai, D.S., Kim, M.S., & Lee, S.H. (1989). Optimum simple step-stress accelerated life tests with censoring. IEEE Transactions on Reliability 38: 528–532.
we will assume that lifetimes of  devices follow a Weibull distribution with common shape parameter (not depending on the stress level), and scale parameter related to the current stress level through a log-linear relationship.  
By assuming common shape for the Weibull distribution, it can be parameterized as a proportional hazards model, meaning that the hazard rates of any two products stay  constant  over time. In an accelerated life testing applications,  a common assumption is that the ALT does not change the failure mode of the units or materials, and so the  homogeneity condition on the shape is valid in most applications; therefore, the stress level would change only the scale parameter of the life distribution.
This assumption simplifies the application of CE and is quite common in ALT literature (See, for example, Wang and Fei (2003), Mondal and Kundu (2019), Samanta et al. (2019) and Pal et al. (2021) ) as it makes the estimation problem mathematically more tractable.
 However, the homogeneity condition on the shape parameter may be violated in some practical applications, when the stress level may affect the shape parameter as well.

%However,  there are some other applications where this condition is fulfilled. See, for example, [17], [18], and [2]. (!!!)
%[18] W. Q. Meeker, L. A. Escobar, and C. J. Lu, “Accelerated degradation tests: Modeling and analysis,” Technometrics, vol. 40, pp. 89–99, 1998.

On the other hand, due to product or experimental constraints, reliability experiments often deal with censored data. For example, 
in destructive one-shot device testing %involve an extreme case of interval censoring, wherein 
we only know if a device has failed or not when it is tested and so
the data collected from such devices are either left- or right-censored. %, with the lifetime being less than the inspection time if the device fails when tested (resulting in left-censoring), and the lifetime being more than the inspection time for surviving devices (resulting in right-censoring).
% as the only usable information  is  if they had failed or not when tested.
Destructive one-shot devices, which get destroyed when tested, have been widely studied in the literature of ALTs.
However, the destructiveness assumption may not be necessary in some engineering applications (see for instance Cheng and Elsayed (2018)), and so the surviving devices would continue in the experiment, providing extra-information.
Then, the observed data from the experiment will instead be interval-censored.
Other examples of interval-censored data may appear in 
%The proposed model would  also be useful for 
intermittently-monitored experiments wherein the exact failure times of  devices can not be recorded, but the state of the devices can only be checked at certain readout times. For example, Gouno (2001) used this experimental set up to infer the lifetime distribution of two electronic products related to medium-power silicon bipolar transistors. Hassan et al. (2014) determined optimum inspection times.
More recently, Han and Bai (2019, 2020, 2022) studied EM estimation for constant and step-stress ALTs under interval-monitoring and compared the results with the continuous monitoring set-up.
We will assume here that the life status of devices under test can only be intermittently checked, leading to an interval-censoring setup.
Constant-stress ALT for one-shot devices have been studied extensively, including classical and robust methods.
One may refer to the recent book by Balakrishnan et al. (2021) for a detailed state-of-the-art review on inferential techniques for ALTs with one-shot devices. %In contrast, little attention has been paid to step-stress ALT for one-shot devices. 

On the other hand, recent works on intermittently-monitored step-stress ALTs have shown the advantage of using divergence-based methods in terms of robustness, with an unavoidable (but not significant) loss of efficiency. Balakrishnan et al. (2023, 2024a, 2024b) developed robust estimation methods based on density power divergence (DPD) for the step-stress model for different lifetime distributions, including exponential, gamma and lognormal.

Although the Weibull distribution is commonly used as a lifetime model in engineering, physical, and biomedical sciences, robust inferences have not yet been studied under this distribution.
In this paper, we extend this work to the case of Weibull distributions. % which has not been addressed in the literature despite the importance of the Weibull distribution in reliability. 
Additionally, we develop a Wald-type statistic to robustly test whether Weibull or exponential distributions should be considered. 
The properties of the minimum DPD estimators for intermittent monitored step-stress ALTs are theoretically derived, including their asymptotic distribution, which allows for the construction of asymptotic confidence intervals. From the robust estimators, point estimates and asymptotic confidence intervals for the Weibull lifetime characteristics are derived.
Additionally, the performance of the proposed robust estimators is empirically compared to classical likelihood-based techniques through Monte Carlo simulations. The advantages for practical applications are illustrated using two different datasets analyzing the effect of temperature in light emitting devices.

The rest of this paper is organised as follows. 
Section \ref{sec:modelformulation} presents the multiple step-stress model for Weibull lifetime distributions. 
In Section \ref{sec:MDPDE},   the classical maximum likelihood estimator (MLE) of the model and the minimum DPD estimators (MDPDE) are introduced, and their estimating equations and asymptotic properties are discussed.
Section \ref{sec:interval} investigates several lifetime characteristics, namely, reliability, distribution quantiles and mean lifetime, as well as their point estimation and confidence intervals (CIs) based on the MDPDEs. Further, transformed CIs are proposed to avoid some practical drawback of direct CIs.
In Section \ref{sec:robusttest}, Wald-type test statistics, based on the MDPDE, for testing general composite null hypothesis are developed. Asymptotic distribution of the proposed test statistic under the null hypothesis and asymptotic behaviour of the power of the tests are established.
Section \ref{sec:IF} theoretically study the robustness properties of the estimators and test statistics based on the MDPDE through the analysis of their influence function (IF). 
In Section \ref{sec:simstudy},  an extensive simulation study is carried out for examining the robustness properties of the estimators and the test statistics is carried out.
Section \ref{sec:realdata} illustrates the applicability of the proposed model and the inferential methods developed here with two real datasets from the reliability literature, regarding the effect of temperature in different electronic components.
Finally, some concluding remarks are made in Section \ref{sec:conclusions}.

\section{Motivating Data Sets \label{sec:Motivating}}

We aim to analyze the influence of temperature  in some electronic components such as solar lights, medium power silicon bipolar transistors and LED lights.
Allowing systems to run for prolonged periods of time in high temperatures can severely decrease the longevity and reliability of some electronic devices. 
Therefore, in order to get the most out of an electronic component in different environments or otherwise choose the best material for a certain function, it will be interest to study its reliability based on the temperature.
%Lifetime distributions of devices are estimated using the proposed MDPDEs and the lifetime characteristics, such as reliability at certain mission times, mean life and  $95\%$ distribution quantiles, are then estimated. 
Recall that the step-stress model enables to obtain accurate estimates with fewer sample sizes than constant-stress ALTs. Therefore, the model is especially adequate for the presented industrial experiments, which traditionally have small sample sizes.

\subsection{Effect of temperature on solar lighting devices \label{Dataset1}}

Solar lights are known to be resilient and environmental-friendly by design and so they offer an attractive lighting alternative for outdoor and indoor places. Conversely, there is an unwanted effect of  extreme weather on the reliability of such devices.
%Extreme heat can actually more of a negative effect on solar panels than extreme cold.
In particular, the efficiency of solar lights generally decreases as the temperature rises, due to a voltage decrease.
% due to the fact that the decrease in voltage will be greater than the increase voltage that comes with warmer weather.
Under normal operating conditions, solar lighting devices are expected to last for hundreds of hours, and therefore an ALT must be carried out to infer the effect of temperature on solar lighting devices.
For this purpose, Han and Kundu (2014) conducted a simple step-stress accelerated life testing experiment.
A set of 35 solar light prototypes were placed on equal environmental conditions, and  the temperature was increased at a pre-specified time $\tau_1=5$ (in hundred hours) from normal operating temperature ($293K$) to $353K.$
The experiment was terminated at $\tau_2=6$ (in hundred hours), when only $4$ devices were still functioning. The election of the time of stress-change were made in advance by the experimenter.
Failure times observed during the experiment  are as follows:
\begin{center}
	0.140, 0.783, 1.324, 1.582, 1.716, 1.794, 1.883, 2.293, 2.660, 2.674, 2.725, 3.085, 3.924, 4.396, 4.612, 4.892, 5.002, 5.022, 5.082, 5.112, 5.147, 5.238, 5.244, 5.247, 5.305, 5.337, 5.407, 5.408, 5.445, 5.483, 5.717.
\end{center}

The collected data can be easily transformed into intermittent inspecting devices with inspection times $t=1.5,3,5,5.2,5.4,6$ (in hundred ours),  by counting the number of failures within each inspection interval as if the exact failure times were not available. The resulting observed data are given by the counts of failures $(3, 8, 5, 5, 5, 5, 4).$
To ensure that the Weibull distribution is appropriate for the observed grouped data, we perform a non-parametric test using the Pearson chi-squared statistic, resulting in a p-value of 0.4. Thus, we can assume the Weibull distribution is appropriate.

%####################### medium power silicon bipolar transistors.
%## AN INFERENCE METHOD FOR TEMPERATURE STEP-STRESS ACCELERATED LIFE TESTING. EVANS GOUNO
%\subsection{Effect of temperature on medium power silicon bipolar transistors \label{Dataset2}}
%
%Gouno (2001) carried out a temperature multiple step-stress ALT with 31 medium power silicon bipolar transistors.  
%He considered the Arrhenius model with censored data to describe the temperature stress.
%The devices were successively exposed to 10 temperature levels of 120, 140, 160, 180, 190, 200, 210, 220, 230, 240 degrees during equal time intervals of length 168 h. 
%At each time of stress change, the number of failures were recorded, yielding the observed vector of failures to be $n = (0,0,0,2,5,5,3,3,0,9).$ Here, all right-censored data have been removed from the dataset.
%We will assume that the lifetime of the medium power silicon bipolar transistors follows a Weibull distribution with common shape parameter, and that the Weibull scale parameter is affected by the temperature.
%% To fit the multiple step-stress model, the temperature levels where normalized and the censored observations were removed.

\subsection{Effect of temperature on Light Emitting Diodes (LEDs) \label{Dataset3}}

Temperature has a tangible effect on the material and output of an LED light. 
Cold environments favour good functioning of LEDs, and light output diminishes with temperature increment.
Zhao and Elsayed (2005) examined the effect of temperature in LEDs through a multiple step-stress ALT experiment with four levels. 
The LEDs are considered to be working properly as long as the change in the light intensity
does not exceed a pre-fixed threshold level. 
Then, they placed $27$ LED units in a temperature and humidity chamber and the humidity and current in the circuit were held constant at 70\% RH and 28 mA, respectively.
Four temperature levels were considered during the test, 
According to expert engineering judgment, temperatures less than $80^\circ C$ are not appropriate since it is difficult to observe the LED failure in a reasonable test time. Conversely, temperatures above and above $200^\circ C$ are not recommended either as extrapolation can not be justified. Therefore, four temperature levels ranging from $80$ to $200^\circ $ were considered for the ALT as $T=363K, 413K, 433K$ and $448K.$ 
The normal operating temperature for LEDs is $50^\circ C.$ 
The times of stress change were fixed at $t = 300, 500$ and $600$ hours, and the test got terminated at $t=720$ hours. Recorded failure times were as follows:
\begin{center}
	347, 397, 432, 491, 512, 567, 574, 588, 597, 603, 605, 615, 633, 634, 637, 644, 653, 675, 684, 699, 706, 718, 720.
\end{center}

We will transformed the data to intermittent tested data in which the inspection times coincide with the times of stress change.
To ensure that the Weibull distribution is appropriate for the observed grouped data, we performed a non-parametric goodness of fit test using the Pearson chi-squared statistic, resulting in a p-value of 0.44. Therefore, the Weibull distribution can be assumed for  lifetime distributions for the LEDs.

\section{The multiple step-stress model} \label{sec:modelformulation}

We consider the multiple step-stress ALT for intermittently tested devices with $k$ ordered stress levels $ x_1 < x_2 < \dots < x_k$ and $N$ devices under test. 
At the start of the experiment, all devices are subjected to the same stress level, $x_1,$ which gets increased to $x_2$ at a certain pre-specified time of stress change, $\tau_1.$ All devices are subjected to the new stress level until the next time of stress change, $\tau_2,$ and the process is repeated until
all remaining surviving devices are subjected $k$-th stress, $x_k.$ The termination time of the experiment is fixed at $\tau_k.$ 
In addition, we consider $L$ pre-fixed inspection times, 
$$0 < t_1 < \dots < t_{l_1} = \tau_1 < t_{l_1+1} < \dots < t_{l_k} = \tau_k, $$
where $l_i$ denotes the number of inspection times before the $i$-th time of stress change and $l_k= L,$ when all the surviving devices are tested and the number of failures is recorded. 
We assume that all times of stress change are inspection times.  
Observed failures from intermittently-monitored devices are grouped as count of failures within each inspected interval.
%The non-destructive assumption implies that all devices that did not fail until a certain time can continue on the experiment, and so observed failures are grouped as count of failures within each inspected interval.

The cumulative exposure (CE) model forms a composite failure distribution function by assuming that the product lifetime is shifted at the time of stress change such that the survival functions of the change time are the same under two different stress levels. We will assume that the lifetime of devices at a constant stress level $x_i$ follow a Weibull distribution with common shape parameter $\eta>0$  (homogeneity condition) and scale parameter $\alpha_i > 0$ depending on the stress level through the log-linear relationship
\begin{equation} \label{eq:loglinear}
	\alpha_i =  \exp(a_0 + a_1 x_i), \hspace{0.3cm} i = 1,..,k,
\end{equation} 
where $(a_0, a_1)^T \in \mathbb{R} \times \mathbb{R}^{-} .$
Note that the CE assumes that the  probability of failure gets increased when the stress level increases, so that the scale $\alpha_i$ should decrease with the stress, or equivalently the log-linear parameter $a_1$ need to be negative. 
The shape parameter $\eta$ is positive, so the parameter $\boldsymbol{\theta} = (a_0,a_1,\eta)^T $ is $\Theta = \mathbb{R} \times \mathbb{R}^{-} \times \mathbb{R}^{+}.$
The log-linear relationship is frequently adopted in accelerated life modeling as it fits some important physical models and have proved to perform well in practice. In particular, when stress is temperature as is the case of our illustrative examples, it is common to consider an Arrhenius model as stress function and an exponential or Weibull life distribution with scale parameter
	$$ \lambda_i = \lambda_0\exp\left( + \frac{E_0}{K}x_i\right)$$
	where $E_0$ and $K$ are unknown parameters that need to be estimated,
	 $x_i = \frac{1}{T_0} - \frac{1}{T_i}$ with $T_i$ the temperature at level $i$ and $T_0$ the normal operating temperature, and $\lambda_0$  the scale parameter under normal operating conditions.

Then, the cumulative distribution function (cdf) of failure time is given by
\begin{equation}\label{eq:distributionT}
	G_T(t) = 
	\begin{cases}
		 1-\exp\left(- \left(\frac{t}{\alpha_1}\right)^{\eta}\right), & 0<t< \tau_1\\
		 1-\exp\left(-\left(\frac{t+h_1}{\alpha_2}\right)^{\eta}
		\right), & \tau_1 \leq t < \tau_2, \\
		\vdots & \vdots \\
		 1-\exp\left(-\left(\frac{t+h_{k-1}}{\alpha_k}\right)^{\eta}
		\right), & \tau_{k-1} \leq t < \infty \\
	\end{cases},
\end{equation}
and correspondingly the reliability function is given by 
\begin{equation*}
	\begin{aligned}
	R_T(t) &= 1- G_T(t)\\ &= \begin{cases} \exp\left(- \left(\frac{t}{\alpha_1}\right)^{\eta}\right),  & 0<t< \tau_1\\
		\exp\left(-\left(\frac{t+h_1}{\alpha_2}\right)^{\eta}
		\right), & \tau_1 \leq t < \tau_2 \\
		\vdots & \vdots \\
		 \exp\left(-\left(\frac{t+h_{k-1}}{\alpha_k}\right)^{\eta}
		\right), & \tau_{k-1} \leq t < \infty \\
	\end{cases},
\end{aligned}
\end{equation*}
and the probability distribution function (pdf) is given by
\begin{equation} \label{eq:densityT}
	g_T(t) = 
	\begin{cases}
	 \frac{\eta}{\alpha_1}\left(\frac{t}{\alpha_1}\right)^{\eta-1} \exp\left(-\left(\frac{t}{\alpha_1}\right)^{\eta}\right), & 0<t <\tau_1\\
	 \frac{\eta}{\alpha_2}\left(\frac{t + h_1}{\alpha_2}\right)^{\eta-1} \exp\left(-\left(\frac{t + h_1}{\alpha_2}\right)^{\eta}\right), & \tau_1 \leq t < \tau_2 \\
		\vdots & \vdots\\
	 \frac{\eta}{\alpha_{k}}\left(\frac{t + h_{k-1}}{\alpha_k}\right)^{\eta-1} \exp\left(-\left(\frac{t + h_{k-1}}{\alpha_k}\right)^{\eta}\right),   & \tau_{k-1} \leq t < \infty \\
	\end{cases},
\end{equation}
with 
\begin{equation}\label{eq:ai}
	h_{i} =  \alpha_{i+1} \sum_{k=0}^{i-1}\left(-\frac{1}{\alpha_{i+1-k}} + \frac{1}{\alpha_{i-k}} \right)\tau_{i-k},
\end{equation}
for $i=1,...,k-1.$ 
Although the distribution function is continuous in $(0,\infty),$ the density function has $k$  points of discontinuity at each time of stress change. 
Equation (\ref{eq:ai}) can be obtained thanks to the homogeneity condition on the shape parameter, as the shifting time assuring continuity of the cdf has an explicit expression only when the change distribution depends only on a scale parameter.
Further, the probability of failure at the $j$-th interval  is
\begin{equation} \label{eq:th.prob}
	\pi_j(\boldsymbol{\theta}) =  G_T(t_j) - G_T(t_{j-1}), \hspace{0.3cm} j = 1,..,L,
\end{equation}
and finally the probability of survival at the end of the experiment is given by $\pi_{L+1}(\boldsymbol{\theta}) = 1 - G_T(t_L).$

\section{Minimum density power divergence estimator}\label{sec:MDPDE}

In this section, we develop robust estimators for the SSALT model under Weibull lifetimes based on the DPD introduced by Basu et al. (1998), and then study their asymptotic and robustness properties.  

Let $n_j$  denote the number of failures within the $j$-th interval and $n_{L+1}$  denote the number of surviving devices at the end of the experiment.
We introduce probability vector $\boldsymbol{\pi} (\boldsymbol{\theta}) = (\pi_1(\boldsymbol{\theta})),...,\pi_{L+1}(\boldsymbol{\theta}))^T$ quantifying the probability of failure within each interval
and its corresponding empirical probability vector $\widehat{\boldsymbol{p}} =  \left(n_1/N, ..., n_{L+1}/N\right)^T.$
Then, the likelihood function of the multinomial model with probability $\boldsymbol{\pi} (\boldsymbol{\theta})$
and $N$ trials (total number devices tested)  is given by
$$\mathcal{L}(\boldsymbol{\theta}; n_1,..,n_{L+1}) = \frac{N!}{n_1!\cdots n_{L+1}!} \prod_{j=1}^{L+1} \pi_j(\boldsymbol{\theta})^{n_j}$$
and then the MLE of the model parameter $\boldsymbol{\theta}$ for the SSALT model under Weibull lifetimes is given by
\begin{align*}
	\boldsymbol{\widehat{\theta}}^{MLE} &= \left(\widehat{a}_0^{MLE}, \widehat{a}_1^{MLE}, \widehat{\eta}^{MLE}\right)^T\\
	& = \operatorname{arg} \operatorname{max}_{\boldsymbol{\theta}  \in \Theta} \mathcal{L}(\boldsymbol{\theta}; n_1,..,n_{L+1}).
\end{align*}

It is useful to note that the MLE can be equivalently derived from an information theory approach by using the Kullback-Leibler (KL) divergence between the empirical vector $\widehat{\boldsymbol{p}}$ and the theoretical probability  $\boldsymbol{\pi} (\boldsymbol{\theta})$ of the multinomial model, since the expression of the KL divergence is related to the likelihood function as
\begin{align*}
	d_{KL} (\boldsymbol{\widehat{p}}, \boldsymbol{\pi}(\boldsymbol{\theta})) &= \sum_{j=1}^{L+1}\widehat{p}_j \log\left(\frac{\widehat{p}_j}{\pi_j(\boldsymbol{\theta})}\right)\\
	& = \sum_{i=1}^{L+1} \widehat{p}_i\log(\widehat{p}_i) - \frac{1}{N}\log \mathcal{L}(\boldsymbol{\theta}; n_1,..,n_{L+1}).
\end{align*}
%where first term does not depend on the model parameter, $\boldsymbol{\theta}.$ 
Then, the minimizer of the KL divergence coincides with the MLE.
This classical estimator is known to be asymptotically efficient in the absence of contamination, but it lacks robustness as contamination in the data could influence the parameter estimate. 
To overcome this, we introduce a family of estimators based on the DPD, indexed by a tuning parameter $\beta \geq 0$ controlling the trade-off between efficiency and robustness. 

The DPD between the empirical and theoretical probability vectors, $\widehat{\boldsymbol{p}}$ and $\boldsymbol{\pi},$ is given by
\begin{equation}\label{eq:DPDloss}
	d_{\beta}\left( \widehat{\boldsymbol{p}},\boldsymbol{\pi}\left(\boldsymbol{\theta}\right)\right)   = \sum_{j=1}^{L+1} \left(\pi_j(\boldsymbol{\theta})^{1+\beta} -\left( 1+\frac{1}{\beta}\right) \widehat{p}_j\pi_j(\boldsymbol{\theta})^{\beta}  +\frac{1}{\beta} \widehat{p}_j^{\beta+1} \right)
\end{equation}
and correspondingly the minimum DPD estimator (MDPPE), for $\beta >0, $ is defined as
\begin{equation} 
	\boldsymbol{\widehat{\theta}}^{\beta} = \left(\widehat{a}_0^{\beta},\widehat{a}_{1}^{\beta}, \widehat{\eta}^{\beta}\right)^T = \operatorname{arg} \operatorname{min}_{\boldsymbol{\theta}  \in \Theta} d_{\beta} \left( \widehat{\boldsymbol{p}},\boldsymbol{\pi}\left(\boldsymbol{\theta}\right)\right).
\end{equation}
For $\beta=0,$ the DPD can be defined by taking continuos limits  $\beta \rightarrow 0,$ and the resulting expression is indeed the KL divergence. Thus,
the DPD family generalizes the likelihood procedure to a broader  class of estimators, including the  classical MLE
for the case when $\beta = 0.$ 

The following result establishes the estimating equations of the MDPDE for any $\beta \geq0.$

%\begin{remark}
%	We could have derived the likelihood function of the model using binomial distributions at each interval, by using conditional probabilities of failure, given the device did not fail at earlier times. Nonetheless, both models are equivalent and yield to the same likelihood function.
%\end{remark}

%Given two density or mass functions, $f_{\boldsymbol{\theta}}$ and $g,$ the DPD between them is defined, for $\beta >0$ as
%\begin{equation*}\label{eq:DPD}
%	d_{\beta}(g, f_{\boldsymbol{\theta}}) =\int \left\{ f_{\boldsymbol{\theta}}^{1+\beta}(y)-\frac{\beta +1}{\beta }f_{\boldsymbol{\theta}}^{\beta}(y)g(y)+\frac{1}{\beta}g^{1+\beta}(y)\right\} dy.
%\end{equation*}
%The parameter $\beta$ indexing the DPD divergence 
%In fact, the DPD can be defined at $\beta = 0$ by taking continuous limits in (\ref{eq:DPD}), yielding to the Kullback-Leibler divergence. 

%Further, since the last term in (\ref{eq:DPDloss}) does not depend on $\boldsymbol{\theta}$ it can be avoid for the minimization.

\begin{theorem}\label{thm:esitmatingequations}
	The estimating equations of the MDPDE for the SSALT model under Weibull lifetime distributions, satisfying the log-linear relation in (\ref{eq:loglinear}), are given by
	$$\boldsymbol{W}^T\boldsymbol{D}_{\boldsymbol{\pi}(\boldsymbol{\theta})}^{\beta-1}\left( \widehat{\boldsymbol{p}}- \boldsymbol{\pi}(\boldsymbol{\theta})\right) = \boldsymbol{0}_3,$$
	where $\boldsymbol{0}_3$ is the 3-dimensional null vector, $\boldsymbol{D}_{\boldsymbol{\pi}(\boldsymbol{\theta})}$ denotes a $(L+1)\times(L+1)$ diagonal matrix with diagonal entries $\pi_j(\boldsymbol{\theta}),$ $j=1,...,L+1,$ and $\boldsymbol{W}$ is a $(L+1) \times 3$ matrix with rows
	$
	\boldsymbol{w}_j =  \boldsymbol{z}_j-\boldsymbol{z}_{j-1},
	$
	where
	\begin{align}
		\label{eq:zj} \boldsymbol{z}_j  &= g_T(t_j)\begin{pmatrix}
			-(t_j+h_{i-1})\\
			-(t_j+h_{i-1})x_i + h_{i-1}^\ast \\
			\log\left(\frac{t_j+h_{i-1}}{\alpha_i}\right)\frac{t_j+h_{i-1}}{\eta}
		\end{pmatrix}, \hspace{0.1cm} \\ 
		h_{i}^\ast &= h_{i}x_{i+1}+ \alpha_{i+1}\sum_{k=0}^{i-1}\left(\frac{x_{i+1-k}}{\alpha_{i+1-k}} - \frac{x_{i-k}}{\alpha_{i-k}} \right)\tau_{i-k} \label{aast},\hspace{0.1cm} 
	\end{align}
	for any $j = 1,...,L$ and $i=1,...,k-1$ with $ \boldsymbol{z}_{-1} = \boldsymbol{z}_{L+1} = \boldsymbol{0}$ and $i$ being the stress level at which the units are tested before the $j-$th inspection time.
\end{theorem}
\begin{proof}
	See Appendix.
\end{proof}

Note that estimating equations of the MLE are  obtained ready by evaluating the previous expression at $\beta=0,$ resulting in
$$\boldsymbol{W}^T\boldsymbol{D}_{\boldsymbol{\pi}(\boldsymbol{\theta})}^{-1}\left( \widehat{\boldsymbol{p}}- \boldsymbol{\pi}(\boldsymbol{\theta})\right) = \boldsymbol{0}_3.$$

Next, we present the asymptotic distribution of the proposed estimator for any $\beta \geq 0$.
\begin{theorem} \label{thm:asymptoticestimator}
	Let $\boldsymbol{\theta}_0$ be the true value of the parameter $\boldsymbol{\theta}$. Then, the asymptotic distribution of the MDPDE, $\boldsymbol{\widehat{\theta}}^{\beta},$ for the SSALT model, under Weibull lifetime distribution, is given by
	$$ \sqrt{N}\left(\boldsymbol{\widehat{\theta}}^{\beta} - \boldsymbol{\theta}_0\right)\xrightarrow[N\rightarrow \infty]{L}\mathcal{N}\left(\boldsymbol{0}, \boldsymbol{J}_\beta^{-1}(\boldsymbol{\theta}_0)\boldsymbol{K}_\beta(\boldsymbol{\theta}_0)\boldsymbol{J}_\beta^{-1}(\boldsymbol{\theta}_0)\right),$$
	where \begin{equation} \label{eq:JK}
		\begin{aligned}
		\boldsymbol{J}_\beta(\boldsymbol{\theta}_0) &= \boldsymbol{W}^T \boldsymbol{D}_{\boldsymbol{\pi}(\boldsymbol{\theta_0})}^{\beta-1} \boldsymbol{W},\\
		\boldsymbol{K}_\beta(\boldsymbol{\theta}_0)& = \boldsymbol{W}^T \left(\boldsymbol{D}_{\boldsymbol{\pi}(\boldsymbol{\theta_0})}^{2\beta-1}-\boldsymbol{\pi}(\boldsymbol{\theta}_0)^{\beta}\boldsymbol{\pi}(\boldsymbol{\theta}_0)^{\beta T}\right) \boldsymbol{W},
	\end{aligned}
	\end{equation}
	$\boldsymbol{D}_{\boldsymbol{\pi}(\boldsymbol{\theta_0})}$ denotes the diagonal matrix with entries $\pi_j(\boldsymbol{\theta_0}),$ $j=1,...,L+1,$ and $\boldsymbol{\pi}(\boldsymbol{\theta}_0)^{\beta}$ denotes the vector with components $\pi_j(\boldsymbol{\theta}_0)^{\beta}.$
\end{theorem}

\begin{proof}
	The proof is similar to that of Theorem 2 in Balakrishnan et al. (2022).
\end{proof}
The Fisher information matrix associated with the SSALT model under Weibull lifetime distribution is given by $\boldsymbol{I}_F(\boldsymbol{\theta}_0) =  \boldsymbol{W}^T \boldsymbol{D}_{\boldsymbol{\pi}(\boldsymbol{\theta_0})}^{-1} \boldsymbol{W},$ and so the convergence of the MLE is obtained as a particular case at $\beta=0,$ to be
$$ \sqrt{N}\left(\boldsymbol{\widehat{\theta}}^{0} - \boldsymbol{\theta}_0\right)\xrightarrow[N\rightarrow \infty]{L}\mathcal{N}\left(\boldsymbol{0}, \boldsymbol{I}_F^{-1}(\boldsymbol{\theta}_0)\right).$$

\section{Point estimation and confidence intervals of reliability, distribution quantiles and mean lifetime of a device } \label{sec:interval}

Engineering applications often demand estimated values of some important lifetime characteristics, such as the reliability  function, distribution quantiles and mean lifetime of a device under normal operating conditions.
Hence,  point estimation and confidence intervals (CIs) of such lifetime characteristics will be of great interest for reliability engineers.
In this section, we develop point estimation and CIs for these three main lifetime characteristics based on the MDPDEs. Further, we establish asymptotic distributions of each lifetime characteristic and derive the corresponding approximate direct and transformed CIs.

Lifetime distribution characteristics are functions of the model parameter $\boldsymbol{\theta},$ and so their asymptotic distribution can be derived by using the Delta-method. In particular, the reliability of the device at a constant stress level, $x_0,$ and  at a fixed time $t,$ is given by
\begin{equation}\label{eq:survival}
	R_t(\boldsymbol{\theta}) = \exp \left(- \left( \frac{t}{\exp(a_0 + a_1x_0)} \right)^\eta\right),
\end{equation}
the $1-\alpha$ distribution quantile  at  stress level $x_0$ is given by the inverse distribution (or reliability) function, as
\begin{equation}\label{eq:quantiles}
	\begin{aligned}
		Q_{1-\alpha}(\boldsymbol{\theta})  &=	R_0^{-1}(1-\alpha) = G_0^{-1}(\alpha)\\
		& = \exp(a_0+a_1x_0)\left(-\log(1-\alpha)\right)^{\frac{1}{\eta}},
	\end{aligned}
\end{equation}
%$$
%	\begin{cases}
%	\frac{-\log(1-\alpha)}{\lambda_k} - a_{k-1} +\tau_{k-1} & 0 \leq 1-\alpha < e^{-\lambda_k(\tau_{k-1}+a_{k-1}-\tau_{k-1})}  \\
%	\vdots \\ 
%	\frac{-\log(1-\alpha)}{\lambda_2}-a_1 +\tau_1 & e^{-\lambda_2(\tau_2+a_1-\tau_1)} \leq 1-\alpha <  e^{-\lambda_1\tau_1}  \\
%	\frac{-\log(1-\alpha)}{\lambda_1}  & e^{-\lambda_1\tau_1}< 1-\alpha<1 \\	
%\end{cases}
%$$
and the mean lifetime of the device is given by
\begin{equation}\label{eq:loglinearmean}
	\operatorname{E}_T(\boldsymbol{\theta}) = \mathbb{E}_{\boldsymbol{\theta}}[T]= \exp(a_0+a_1x_0)\Gamma\left(1+\frac{1}{\eta}\right)
\end{equation}
with $\Gamma(\cdot)$ denoting the gamma function.

Given an MDPDE, $\widehat{\boldsymbol{\theta}}^\beta$, it is straightforward to obtain a point estimation of the reliability at a certain time, quantile and mean lifetime of  devices  under normal operating conditions by substituting the estimated parameters in (\ref{eq:survival}), (\ref{eq:quantiles}) and (\ref{eq:loglinearmean}), respectively.
The next result presents the asymptotic distributions of these three estimated lifetime characteristics based on the MDPDE.

\begin{theorem} \label{thm:asymptoticreliability}
	Let $\boldsymbol{\theta}_0$ be the true value of the parameter $\boldsymbol{\theta},$ and then consider $\widehat{\boldsymbol{\theta}}^\beta$ as the MDPDE with tuning parameter $\beta.$
	Then, the asymptotic distribution of the estimated reliability of devices
	at a certain time $t$ under normal operating conditions,   based on the MDPDE,  $R_t(\widehat{\boldsymbol{\theta}}^\beta),$ is given by
	$$\sqrt{N}(R_t(\widehat{\boldsymbol{\theta}}^\beta)- R_t(\boldsymbol{\theta}_0) )\xrightarrow[N \rightarrow \infty]{L} \mathcal{N}\left(\boldsymbol{0}, \sigma(R_t(\boldsymbol{\theta}_0)) \right)$$
	with
	$$\sigma(R_0^\beta(t))^2= \nabla R_t\left(\boldsymbol{\theta}_0\right)^T \boldsymbol{J}_\beta^{-1}(\boldsymbol{\theta}_0)\boldsymbol{K}_\beta(\boldsymbol{\theta}_0)\boldsymbol{J}_\beta^{-1}(\boldsymbol{\theta}_0)\nabla R_t\left(\boldsymbol{\theta}_0\right),$$
	where the matrices 
	$\boldsymbol{J}_\beta(\boldsymbol{\theta}_0)$ and $\boldsymbol{K}_\beta(\boldsymbol{\theta}_0)$ are as defined in (\ref{eq:JK}) and 
	$\nabla R_t\left(\boldsymbol{\theta}\right)^T = tg_0(t) \left(1, x_0, -\log\left(\frac{t}{\alpha_0}\right)\frac{1}{\eta}\right)$ is the gradient of the function $R_t(\boldsymbol{\theta})$ defined in (\ref{eq:survival}).
\end{theorem}

%\begin{proof}
%	Since the MDPDE $\widehat{\boldsymbol{\theta}}^\beta$ has an asymptotic normal distribution,
%	$$ \sqrt{N}\left(\boldsymbol{\widehat{\theta}}^{\beta} - \boldsymbol{\theta}_0\right) \rightarrow \mathcal{N}\left(\boldsymbol{0}, \boldsymbol{J}_\beta^{-1}(\boldsymbol{\theta}_0)\boldsymbol{K}_\beta(\boldsymbol{\theta}_0)\boldsymbol{J}_\beta^{-1}(\boldsymbol{\theta}_0)\right),$$
%	%with matrices $\boldsymbol{J}_\beta(\boldsymbol{\theta}_0)$ and $\boldsymbol{K}_\beta(\boldsymbol{\theta}_0)$ defined in (\ref{eq:JK}).
%	 we can apply the Delta-method to obtain the asymptotic distribution of $h(\widehat{\theta_0},\widehat{\theta_1}) = \exp(-\widehat{\theta}_0\exp(\widehat{\theta}_1x_0)t) = \widehat{R}_0^\beta(t).$
%\end{proof}

\begin{theorem}\label{thm:asymptoticquantiles}
	Under the same assumptions as in Result \ref{thm:asymptoticreliability},
	the asymptotic distribution of the estimated $1-\alpha$ quantile of the lifetime distribution of devices
	under normal operating conditions based on the MDPDE,  $Q_{1-\alpha}(\widehat{\boldsymbol{\theta}}^\beta),$ is given by
	$$\sqrt{N}(\widehat{Q}_{1-\alpha}(\widehat{\boldsymbol{\theta}}^\beta)- Q_{1-\alpha}(\boldsymbol{\theta}_0)) \xrightarrow[N \rightarrow \infty]{L} \mathcal{N}\left(\boldsymbol{0}, \sigma(Q_{1-\alpha}(\boldsymbol{\theta}_0)) \right) $$
	with
	
	\begin{align*}
		\sigma(Q_{1-\alpha}&(\boldsymbol{\theta}_0) )^2 =\\ &\nabla Q_{1-\alpha}\left(\boldsymbol{\theta}_0\right)^T \boldsymbol{J}_\beta^{-1}(\boldsymbol{\theta}_0)\boldsymbol{K}_\beta(\boldsymbol{\theta}_0)\boldsymbol{J}_\beta^{-1}(\boldsymbol{\theta}_0)\nabla Q_{1-\alpha}\left(\boldsymbol{\theta}_0\right) 
	\end{align*}
	where the matrices 
	$\boldsymbol{J}_\beta(\boldsymbol{\theta}_0)$ and $\boldsymbol{K}_\beta(\boldsymbol{\theta}_0)$ are as defined in (\ref{eq:JK}) and 
	$\nabla Q_{1-\alpha}(\boldsymbol{\theta})^T = Q_{1-\alpha}(\boldsymbol{\theta}) \left(1, x_0, \frac{-\log(-\log(1-\alpha))}{\eta^2}\right)$ is the gradient of the function $Q_{1-\alpha}(\boldsymbol{\theta})$ defined in (\ref{eq:quantiles}).
\end{theorem}

\begin{theorem}\label{thm:asymptoticmean}
	Under the same assumptions as in Result \ref{thm:asymptoticreliability}, the asymptotic distribution of the estimated mean lifetime of devices
	under normal operating conditions based on the MDPDE $\widehat{\boldsymbol{\theta}}^\beta$, $\operatorname{E}_T(\widehat{\boldsymbol{\theta}}^\beta),$ is given by
	$$\sqrt{N}(\operatorname{E}_T(\widehat{\boldsymbol{\theta}}^\beta) - \operatorname{E}_T(\boldsymbol{\theta}_0)) \xrightarrow[N \rightarrow \infty]{L} \mathcal{N}\left(\boldsymbol{0}, \sigma(\operatorname{E}_T(\boldsymbol{\theta}_0)) \right) $$
	with
	$$\sigma(\operatorname{E}_T(\boldsymbol{\theta}_0))^2 = \nabla \operatorname{E}_T(\boldsymbol{\theta}_0)^T \boldsymbol{J}_\beta^{-1}(\boldsymbol{\theta}_0)\boldsymbol{K}_\beta(\boldsymbol{\theta}_0)\boldsymbol{J}_\beta^{-1}(\boldsymbol{\theta}_0)\nabla \operatorname{E}_T(\boldsymbol{\theta}_0),$$
	where the matrices 
	$\boldsymbol{J}_\beta(\boldsymbol{\theta}_0)$ and $\boldsymbol{K}_\beta(\boldsymbol{\theta}_0)$ are as defined in (\ref{eq:JK}) and 
	$\nabla\operatorname{E}_T(\boldsymbol{\theta})^T = \operatorname{E}_T(\boldsymbol{\theta}) \left(1, x_0, \psi_0\left(1+\frac{1}{\eta}\right)\frac{-1}{\eta^2}\right)$ is the gradient of the function $\operatorname{E}_T(\boldsymbol{\theta})$ defined in (\ref{eq:loglinearmean}), and $\psi_0\left(\cdot\right)$ denotes the digamma function.
\end{theorem}

As $\widehat{\boldsymbol{\theta}}^\beta$ is a consistent estimator of the true parameter value, $\boldsymbol{\theta}_0$, 
the diagonal entries of the matrix $\boldsymbol{J}_\beta^{-1}(\widehat{\boldsymbol{\theta}}^\beta)\boldsymbol{K}_\beta(\widehat{\boldsymbol{\theta}}^\beta)\boldsymbol{J}_\beta^{-1}(\widehat{\boldsymbol{\theta}}^\beta)$
are consistent estimators of the asymptotic variances of the model parameters, $(\widehat{a}_0^\beta, \widehat{a}_1^\beta, \widehat{\eta}^\beta)$.
%for $\beta >0,$ in the following denoted by $\sigma^2(\theta_i^\beta).$

Therefore, from Result \ref{thm:asymptoticestimator}, we can obtain asymptotic CIs for $a_0,$ $a_1$ and $\eta$  with confidence level $(1-\alpha)\%,$ as
\begin{equation}\label{eq:CI}
	\widehat{a}_i^\beta \pm z_{\alpha/2}\frac{\widehat{\sigma}(a_i^\beta)}{\sqrt{N}}, \hspace{0.3cm} i = 0,1, \hspace{0.3cm}\text{ and }\hspace{0.3cm} \widehat{\eta}^\beta \pm z_{\alpha/2}\frac{\widehat{\sigma}(\eta^\beta)}{\sqrt{N}},
\end{equation} with $z_{\alpha/2}$ being the upper $\alpha/2$ percentage point of the standard normal distribution and $\sigma^2(a_i^\beta)$ and $\sigma^2(\eta^\beta)$  are the estimated variances of $\widehat{a}_i^\beta$ and $\widehat{\eta}^\beta,$ respectively.
Furthermore, approximate two-sided $100(1 - \alpha)\%$ CI of the reliability, $(1-\alpha)$-quantile and mean lifetime under normal condition are given, respectively, by
$$R_t(\widehat{\boldsymbol{\theta}}^\beta) \pm  z_{\alpha/2}\frac{\sigma(R_t(\widehat{\boldsymbol{\theta}}^\beta))}{\sqrt{N}}, \hspace{0.3cm}
Q_{1-\alpha}(\widehat{\boldsymbol{\theta}}^\beta) \pm  z_{\alpha/2}\frac{\sigma(Q_{1-\alpha}(\widehat{\boldsymbol{\theta}}^\beta))}{\sqrt{N}}
\hspace{0.2cm} $$
and
$$\hspace{0.2cm} \operatorname{E}_{T}(\widehat{\boldsymbol{\theta}}^\beta) \pm  z_{\alpha/2}\frac{\sigma(\operatorname{E}_T(\widehat{\boldsymbol{\theta}}^\beta))}{\sqrt{N}},$$
where $\sigma(R_t(\widehat{\boldsymbol{\theta}}^\beta)),$ $\sigma(Q_{1-\alpha}(\widehat{\boldsymbol{\theta}}^\beta))$ and $\sigma(\operatorname{E}_{T}(\widehat{\boldsymbol{\theta}}^\beta))$ are as defined in Results \ref{thm:asymptoticreliability}, \ref{thm:asymptoticquantiles} and \ref{thm:asymptoticmean}, respectively.

The above CIs work well for large samples. However, in the case of small samples, the interval limits may need to be truncated to satisfy some constraints, namely, positivity of the quantiles and mean lifetime and reliability lying between $(0,1)$. To avoid such truncation, Balakrishnan et al. (2022) proposed transformed CIs based on the logit function (for the reliability) and logarithmic function (for the quantile and mean lifetime).
The resulting asymptotic CIs for the reliability, quantile and mean lifetime are, respectively, given by
$$  \left[ \frac{R_t(\widehat{\boldsymbol{\theta}}^\beta)}{R_t(\widehat{\boldsymbol{\theta}}^\beta)+(1-R_t(\widehat{\boldsymbol{\theta}}^\beta))S}, \frac{R_t(\widehat{\boldsymbol{\theta}}^\beta)}{R_t(\widehat{\boldsymbol{\theta}}^\beta)+(1-R_t(\widehat{\boldsymbol{\theta}}^\beta))/S}\right],
$$
\begin{align*}
	\Bigg[&Q_{1-\alpha}(\widehat{\boldsymbol{\theta}}^\beta) \exp\left( -\frac{z_{\alpha/2}}{\sqrt{N}}\frac{\sigma(Q_{1-\alpha}(\widehat{\boldsymbol{\theta}}^\beta))}{Q_{1-\alpha}(\widehat{\boldsymbol{\theta}}^\beta)}\right), \\ 
	& \hspace{1.7cm} \widehat{Q}_{1-\alpha}^\beta \exp\left( \frac{z_{\alpha/2}}{\sqrt{N}}\frac{\sigma(Q_{1-\alpha}(\widehat{\boldsymbol{\theta}}^\beta))}{Q_{1-\alpha}(\widehat{\boldsymbol{\theta}}^\beta)}\right) \Bigg],
\end{align*}
%with $S_1 = \exp\left( \frac{z_{\alpha/2}}{N}\frac{\sigma(\widehat{\operatorname{E}}_T^\beta)}{\widehat{\operatorname{E}}_T^\beta}\right)$
and 
\begin{align*}
	\Bigg[&\operatorname{E}_T(\widehat{\boldsymbol{\theta}}^\beta)\exp\left( -\frac{z_{\alpha/2}}{\sqrt{N}}\frac{\sigma(\operatorname{E}_T(\widehat{\boldsymbol{\theta}}^\beta))}{\operatorname{E}_T(\widehat{\boldsymbol{\theta}}^\beta)}\right),\\ & \hspace{1.7cm} \operatorname{E}_T(\widehat{\boldsymbol{\theta}}^\beta) \exp\left( \frac{z_{\alpha/2}}{\sqrt{N}}\frac{\sigma(\operatorname{E}_T(\widehat{\boldsymbol{\theta}}^\beta))}{\operatorname{E}_T(\widehat{\boldsymbol{\theta}}^\beta)}\right) \Bigg],
\end{align*}
with $ S= \exp\left( \frac{z_{\alpha/2}}{\sqrt{N}}\frac{\sigma(R_t(\widehat{\boldsymbol{\theta}}^\beta))}{\widehat{R}_0^\beta(t)(1-\widehat{R}_0^\beta(t))},\right)$ and $\sigma(R_t(\boldsymbol{\theta})),$ $\sigma(Q_{1-\alpha}(\boldsymbol{\theta}))$ and $\sigma(\operatorname{E}_T(\boldsymbol{\theta}))$ are as defined in Results \ref{thm:asymptoticreliability}-\ref{thm:asymptoticmean}.

Balakrishnan and Ling (2013) empirically showed that  CIs for the reliability based on the MLE  constructed by applying the logit-transformation approach are more accurate than direct CIs, but the transformation approach  does not work well in the case of small samples as the MLE of the mean lifetime for Weibull distribution  in (\ref{eq:loglinearmean}) does not possess a near normal distribution.

\section{Wald-type tests} \label{sec:robusttest}

In this section, we will consider composite null hypothesis on the model parameter $\boldsymbol{\theta}$ of the form
\begin{equation}\label{eq:null}\operatorname{H}_0 : \boldsymbol{m}(\boldsymbol{\theta}) = 0,
\end{equation}
where $\boldsymbol{m} : \mathbb{R}^{3} \rightarrow \mathbb{R}^{r} $, with $r \in \{1,2\}.$
Choosing $\boldsymbol{m}(\boldsymbol{\theta}) = a_1,$ the null hypothesis tests if the stress level does affect the lifetime of the one-shot devices or not, and with the choice of $\boldsymbol{m}(\boldsymbol{\theta}) = (m_0, m_1,m_2)\boldsymbol{\theta} - d,$ with $d \in \mathbb{R},$  linear null hypothesis can be tested.

	The Wald-type statistic based on the MDPDE, $\widehat{\boldsymbol{\theta}}^\beta,$ for  testing the null hypothesis in (\ref{eq:null}) is given by
	\begin{equation}\label{eq:Wtypest}
		\begin{aligned}
		W_{N}(\widehat{\boldsymbol{\theta}}^\beta) = &N \boldsymbol{m}^T(\widehat{\boldsymbol{\theta}}^\beta) \bigg(\boldsymbol{M}^T(\widehat{\boldsymbol{\theta}}^\beta)\boldsymbol{J}_\beta^{-1}(\widehat{\boldsymbol{\theta}}^\beta)	\\
		 & \times \boldsymbol{K}_\beta(\widehat{\boldsymbol{\theta}}^\beta)\boldsymbol{J}_\beta^{-1}(\widehat{\boldsymbol{\theta}}^\beta)\boldsymbol{M}(\widehat{\boldsymbol{\theta}}^\beta)\bigg)^{-1}
		\boldsymbol{m}(\widehat{\boldsymbol{\theta}}^\beta),
		%\left(\boldsymbol{m}^T\widehat{\boldsymbol{\theta}}^\beta - d\right).
	\end{aligned}
	\end{equation}
	where $\boldsymbol{M}(\widehat{\boldsymbol{\theta}}^\beta) = \frac{\partial \boldsymbol{m}^T(\boldsymbol{\theta})}{\partial \boldsymbol{\theta}}$ is a matrix of rank $r$ and $\boldsymbol{J}_\beta(\boldsymbol{\theta})$ and $\boldsymbol{K}_\beta(\boldsymbol{\theta})$ are defined as in (\ref{eq:JK}).

\begin{remark}
		The exponential distribution is a simple yet useful model for fitting lifetime data. However, it assumes that the lifetime hazard function is constant over time, indicating that the probability of failure at any given moment is the same, regardless of how long the device has been operating. This assumption may not be adequate for products that experience degradation over time.
	Because the exponential distribution is a particular case of the Weibull family with $\eta = 1,$ its validity can be tested in practice using the linear null hypothesis $\operatorname{H}_0: \eta = 1.$% and the presented Wald-type test statistics.
\end{remark}
The following result presents the asymptotic distribution of this Wald-type statistic.

\begin{theorem}\label{thm:asymptoticstat}
	The asymptotic distribution of the Wald-type statistic defined in (\ref{eq:Wtypest}), under the composite null hypothesis (\ref{eq:null}), is a chi-squared $(\chi^2)$ distribution with $r$ degrees of freedom.
\end{theorem}
\begin{proof}
	See Appendix
\end{proof}

Based on Result \ref{thm:asymptoticstat}, for any $\beta \geq 0,$ the critical region with significance level $\alpha$ for the hypothesis test with linear null hypothesis in (\ref{eq:null}) is given by
\begin{equation}\label{eq:criticalregionWtest}
	\mathcal{R}_{\alpha} = \{(n_1,...,n_{L+1})  \text{ s.t. }  W_{N}(\widehat{\boldsymbol{\theta}}^\beta) > \chi^2_{r,\alpha}\},
\end{equation}
where $\chi^2_{r,\alpha}$ denotes the upper $\alpha$ percentage point of a chi square distribution with $r$ degrees of freedom.
The following result provides an asymptotic approximation to the power of the test.

\begin{theorem} \label{thmasymptoticpower}
	Let $\boldsymbol{\theta}^\ast \in \Theta$ be the true value of the parameter $\boldsymbol{\theta}$ with $\boldsymbol{m}(\boldsymbol{\theta}^\ast) \neq 0.$ Then, the approximate power function of the test
	in (\ref{eq:criticalregionWtest}) is given by
	$$\beta_N\left(\boldsymbol{\theta}^\ast\right) \approx 1- \Phi\left(\frac{\sqrt{N}}{\sigma_{W_N(\boldsymbol{\theta}^\ast)} } \left(\frac{\chi^2_{r,\alpha}}{N}- \ell^\ast(\boldsymbol{\theta}^\ast, \boldsymbol{\theta}^\ast) \right)\right),$$
	where 
	\begin{align*}
	\ell^\ast(\boldsymbol{\theta}_1,\boldsymbol{\theta}_2) =& N \boldsymbol{m}^T(\boldsymbol{\theta}_1) \bigg(\boldsymbol{M}^T(\boldsymbol{\theta}_2)\boldsymbol{J}_\beta^{-1}(\boldsymbol{\theta}_2)\\
	& \times \boldsymbol{K}_\beta(\boldsymbol{\theta}_2)\boldsymbol{J}_\beta^{-1}(\boldsymbol{\theta}_2)\boldsymbol{M}(\boldsymbol{\theta}_2)\bigg)^{-1}
	\boldsymbol{m}(\boldsymbol{\theta}_1),
\end{align*}
	
	\begin{align*}
		\sigma_{W_N(\boldsymbol{\theta}^\ast)} =& \frac{\partial \ell^\ast(\boldsymbol{\theta},\boldsymbol{\theta}^\ast)}{\partial \boldsymbol{\theta}^T} \bigg|_{\boldsymbol{\theta} = \boldsymbol{\theta}^\ast} \boldsymbol{J}_\beta^{-1}(\boldsymbol{\theta}^\ast)\boldsymbol{K}_\beta(\boldsymbol{\theta}^\ast)\boldsymbol{J}_\beta^{-1}(\boldsymbol{\theta}^\ast)\\
		&\times 
		\frac{\partial \ell^\ast(\boldsymbol{\theta},\boldsymbol{\theta}^\ast)}{\partial \boldsymbol{\theta}} \bigg|_{\boldsymbol{\theta} = \boldsymbol{\theta}^\ast}
	\end{align*} 
	and $\Phi(\cdot)$ denotes the  distribution function of a standard normal distribution.
\end{theorem}
\begin{proof}
	See Appendix.
\end{proof}

It is clear that $\lim_{N\rightarrow \infty} \beta_N(\boldsymbol{\theta}^\ast)=1,$ and consequently the Wald-type statistic is consistent in the  sense of Fraser (Fraser (1957)).

\subsection{Contiguous alternative hypothesis}

We may find a better approximation to the power function of the Wald-type test statistic, $W_{N}(\widehat{\boldsymbol{\theta}}^\beta),$ under contiguous alternative hypothesis of the form
\begin{equation} \label{eq:cont1}
	\operatorname{H}_{1,n}: \boldsymbol{\theta} = \boldsymbol{\theta}_n  = \boldsymbol{\theta}_0 + \frac{\boldsymbol{d}}{\sqrt{N}}\in \Theta \setminus \Theta_0,
\end{equation}
where $\boldsymbol{\theta}_0 \in \Theta_0$  is the closest element to $\boldsymbol{\theta}_n $ in terms of Euclidean distance and $\boldsymbol{d} \in \mathbb{R}^3$ indicates the direction of the difference between the true value of the vector parameter and the closest element in the null space.
Further, alternative contiguous hypothesis could also be defined by relaxing the null hypothesis constraint, yielding
\begin{equation}\label{eq:cont2}
	\begin{aligned}
		\operatorname{H}_0&: \boldsymbol{m}(\boldsymbol{\theta}) = \boldsymbol{0}_r,\\
		\operatorname{H}_{1,n}^\ast&: \boldsymbol{m}(\boldsymbol{\theta}) = \frac{\boldsymbol{\delta}}{\sqrt{N}},
	\end{aligned}
\end{equation}
for some $\boldsymbol{\delta} \in \mathbb{R}^r.$ Note that  a Taylor series expansion of $\boldsymbol{m}(\boldsymbol{\theta}_n)$ around $\boldsymbol{\theta}_0$ yields
\begin{align*}
	\boldsymbol{m}(\boldsymbol{\theta}_n) &= \boldsymbol{M}^T(\boldsymbol{\theta}_0)\left(\boldsymbol{\theta}_n - \boldsymbol{\theta}_0\right) + o\left(\parallel \boldsymbol{\theta}_n - \boldsymbol{\theta}_0 \parallel \right)\\
	& =  \frac{\boldsymbol{M}^T(\boldsymbol{\theta}_0)\boldsymbol{d}}{\sqrt{N}} + o\left(\parallel \boldsymbol{\theta}_n - \boldsymbol{\theta}_0 \parallel \right)
\end{align*}
and so $\operatorname{H}_{1,n}$ and $\operatorname{H}_{1,n}^\ast$ are asymptotically equivalent by choosing $\boldsymbol{\delta} = \boldsymbol{M}^T(\boldsymbol{\theta}_0)\boldsymbol{d}.$ The following result states the asymptotic distribution of the Wald-type test statistics under both contiguous alternative hypotheses.

\begin{theorem}
	The asymptotic distribution of the Wald-type test statistic based on the MDPDE, with tuning parameter $\beta$, $W_{N}(\widehat{\boldsymbol{\theta}}^\beta),$ is given by
	\begin{itemize}
		\item $W_{N}(\widehat{\boldsymbol{\theta}}^\beta) \xrightarrow[n\rightarrow \infty]{L} \chi^2_r\left( \nu_1 \right)$ 
		with
		\begin{equation}\label{eq:v1}
					\begin{aligned}
				\nu_1 =  \boldsymbol{d}^T\boldsymbol{M}(\boldsymbol{\theta}_0)&\bigg[ \boldsymbol{M}^T(\boldsymbol{\theta}_0) \boldsymbol{J}_\beta^{-1}(\boldsymbol{\theta}_0) \boldsymbol{K}_\beta(\boldsymbol{\theta}_0)
				\\ &\times \boldsymbol{J}_\beta^{-1}(\boldsymbol{\theta}_0) \boldsymbol{M}(\boldsymbol{\theta}_0) \bigg]^{-1}\boldsymbol{M}^T(\boldsymbol{\theta}_0)\boldsymbol{d}
			\end{aligned}
		\end{equation}

		under the alternative hypothesis $\operatorname{H}_{1,n}$ given in (\ref{eq:cont1});
		
		\item $W_{N}(\widehat{\boldsymbol{\theta}}^\beta) \xrightarrow[n\rightarrow \infty]{L} \chi^2_r\left(\nu_2\right)$
		with
		\begin{equation}\label{eq:v2}
			\begin{aligned}
			\hspace{-0.3cm}  \nu_2 = \boldsymbol{\delta}^T [ \boldsymbol{M}^T(\boldsymbol{\theta}_0) \boldsymbol{J}_\beta^{-1}(\boldsymbol{\theta}_0) \boldsymbol{K}_\beta(\boldsymbol{\theta}_0) \boldsymbol{J}_\beta^{-1}(\boldsymbol{\theta}_0) \boldsymbol{M}(\boldsymbol{\theta}_0) ]^{-1} \boldsymbol{\delta} 
		\end{aligned}
	\end{equation}
	 under the alternative hypothesis $\operatorname{H}_{1,n}$ given in (\ref{eq:cont2}).
		
	\end{itemize}
\end{theorem}

\begin{proof}
	See Theorem 3 of Basu et al. (2016)
\end{proof}

From the above result, we could approximate the power function of the test under contiguous alternative hypothesis as
\begin{equation*}
	\beta(\boldsymbol{\theta}_n) = 1- F_{\chi_{r}^2(\nu)}(\chi_{r,\alpha}^2),
\end{equation*}
where $F_{\chi_{r}^2}(\cdot)$ is the distribution function of a non-central chi-square with $r$ degrees of freedom and non-centrality parameter $\nu_1$ and $\nu_2$ defined as in Equations (\ref{eq:v1}) and (\ref{eq:v2}), respectively.

\section{Influence function analysis \label{sec:IF}}

In this section ,we analyze the robustness properties of the MDPDE and their associated Wald-type test statistics through their Influence Function (IF). The IF of an estimator or statistic represents the rate of change in its associated statistical functional with respect to a small amount of contamination by another distribution, i.e., it quantifies the impact of an infinitesimal perturbation in the true distribution underlying the data on the asymptotic value of the resulting estimator.
%Hampel
Thus, Hampel’s IF (Hampel, 1986) is an important measure of robustness for examining the local stability along with the global reliability of an estimator.
Therefore, robust estimators and test statistics must have bounded IF.

\subsection{Influence function of the MDPDE}
Let $\boldsymbol{G}$ be the true density underlying the data with mass function $\boldsymbol{g}.$
The IF of the MDPDE, $\widehat{\boldsymbol{\theta}}^\beta,$ at a point perturbation $\boldsymbol{n}$ is computed in terms of its corresponding statistical functional,  denoted by $\boldsymbol{T}_\beta(\boldsymbol{G}),$ as 
\begin{equation} \label{eq:IFmath}
	\text{IF}\left(\boldsymbol{t}, \boldsymbol{T}, \boldsymbol{G}\right) = \lim_{\varepsilon \rightarrow 0}\frac{\boldsymbol{T}(\boldsymbol{G}_\varepsilon)- \boldsymbol{T}(\boldsymbol{G})}{\varepsilon} = \frac{\partial \boldsymbol{T}(G_\varepsilon)}{\partial \varepsilon}\bigg|_{\varepsilon = 0},
\end{equation} 
where $\boldsymbol{G}_\varepsilon = (1-\varepsilon)\boldsymbol{G} + \varepsilon\Delta_{\boldsymbol{n}}$ is the contaminated version of $\boldsymbol{G},$ with $\varepsilon$ being the contamination proportion and $\Delta_{\boldsymbol{n}}$ the degenerate distribution at the contamination point $\boldsymbol{n}.$ For the SSALT model, we could consider being one cell contamination, and so the contamination point $ \boldsymbol{n}$ would have all elements equal to zero except for one component. 

Let us denote $F_{\boldsymbol{\theta}}$ for the assumed distribution of the multinomial model with mass function $\boldsymbol{\pi}(\boldsymbol{\theta})$ given by the SSALT model with Weibull lifetime distribution. Then, the statistical functional $\boldsymbol{T}_\beta(\boldsymbol{G})$ associated with the MDPDE is computed as  the minimizer of the DPD between the two mass functions, $\boldsymbol{\pi}_{\boldsymbol{\theta}}$ and $\boldsymbol{g}.$ The next result presents an explicit expression of the IF function of the MDPDE for the SSALT model with Weibull lifetime distribution.

\begin{theorem}\label{thm:IF}
	The IF of the MDPDE of the SSALT model, $\widehat{\boldsymbol{\theta}}^\beta,$ at a point contamination $\boldsymbol{n}$ and the assumed model distribution $F_{\boldsymbol{\theta}_0}$ is given by
	\begin{equation}\label{eq:IF}
		\text{IF}\left(\boldsymbol{n}, \boldsymbol{T}_\beta, F_{\boldsymbol{\theta}_0} \right) = \boldsymbol{J}_\beta^{-1}(\boldsymbol{\theta}_0) \boldsymbol{W}^T \boldsymbol{D}_{\boldsymbol{\pi}(\boldsymbol{\theta}_0)}^{\beta-1}\left(-\boldsymbol{\pi}(\boldsymbol{\theta}_0)+\Delta_{\boldsymbol{n}} \right).
	\end{equation}
\end{theorem}
\begin{proof}
	See Appendix.
\end{proof}

\begin{remark}
	To examine the robustness of the estimators, we should analyze the boundedness of their IFs.
	The matrix $\boldsymbol{J}_\beta(\boldsymbol{\theta}_0)$ is usually assumed to be bounded, and so the robustness of the estimators depends on the the second factor of the IF, given by
	\begin{equation}\label{eq:secondfactorIF}
		\begin{aligned}
		&\boldsymbol{W}^T \boldsymbol{D}_{\boldsymbol{\pi}(\boldsymbol{\theta}_0)}^{\beta-1}\left(-\boldsymbol{\pi}(\boldsymbol{\theta}_0)+\Delta_{\boldsymbol{n}} \right) = \\
		&\sum_{j=1}^{L+1} (\boldsymbol{z}_j -\boldsymbol{z}_{j-1} )\pi_j(\boldsymbol{\theta}_0)^{\beta-1} \left(-\pi_j(\boldsymbol{\theta}_0)+
		\Delta_{\boldsymbol{n}j}\right),
			\end{aligned}
	\end{equation}
	where $\boldsymbol{z}_j$ is as defined in (\ref{eq:zj}). 
	
	As discussed in Balakrishnan et al. (2022), all terms in the expression in (\ref{eq:secondfactorIF}) are bounded for fixed stress levels and inspection times at any contamination point $\boldsymbol{n},$ and so all MDPDEs for $\beta \geq 0,$ including the MLE, are robust against vertical outliers.
	
	Then, we should study the limiting behaviour of the IF for large inspection times or stress levels (leverage points).
	% whereas the MDPDEs for positive values of $\beta$ remain robust. 
	We first consider the situation where an inspection time, $t_j ,$ tends to infinity, for fixed $j$. We set $i$ for the fixed stress level corresponding to the $j-$th inspection time, i.e., $\tau_{i-1} < t_j \leq \tau_{i}.$
	As the inspection times are ordered, all inspection times from $j$ onwards tend to infinity.
	Further, note that the values of $\alpha_i, h_{i-1}$ and $h_{i-1}^\ast$  are positive constants for any $i = 1,...,k$.
	%there will be no more terms in the summation after the $j$-th term.
	%We must establish the boundedness of all terms of the summation from the $j-$th onwards. Since the inspection times are ordered, we have that $t_l \rightarrow \infty$ for all $l \geq j.$ 
	%For the $j$-th term, 
	For the $j$-th term, we can write
\resizebox{0.45\textwidth}{!}{%
\begin{minipage}{0.45\textwidth}
	\begin{equation}
	\begin{aligned}
		&(\boldsymbol{z}_j-\boldsymbol{z}_{j-1})\pi_j(\boldsymbol{\theta}_0)^{\beta-1} =(G_i(T_j)-G_i(T_{j-1}))^{\beta-1} \\
		&
	\hspace{-0.3cm} \times	\left(  g_i(T_j)\begin{pmatrix}
			-T_j\\
			-T_jx_i + h_{i-1}^\ast \\
			\log\left(\frac{T_j}{\alpha_i}\right)\frac{T_j}{\eta}
		\end{pmatrix}-  g_T(T_{j-1})\begin{pmatrix}
			-T_{j-1}\\
			-T_{j-1}x_i + h_{i-1}^\ast\\
			\log\left(\frac{T_{j-1}}{\alpha_i}\right)\frac{T_{j-1}}{\eta}
		\end{pmatrix}
		\right)\\
		=& 
		\Bigg[ 
		\frac{\eta}{\alpha_{i}}\left(\frac{T_j}{\alpha_i}\right)^{\eta} \exp\left(-\left(\frac{T_j}{\alpha_i}\right)^{\eta}\right)
		\begin{pmatrix}
			-T_{j}\\
			-T_{j}x_i + h_{i-1}^\ast\\
			\log\left(\frac{T_j}{\alpha_i}\right)\frac{T_j}{\eta}
		\end{pmatrix}\\
		&-  \frac{\eta}{\alpha_{i}}\left(\frac{T_{j-1}}{\alpha_i}\right)^{\eta} \exp\left(-\left(\frac{T_{j-1}}{\alpha_i}\right)^{\eta}\right)
		\begin{pmatrix}
			-T_{j-1}\\
			-T_{j-1}x_i + h_{i-1}^\ast \\
			\log\left(\frac{T_{j-1}}{\alpha_i}\right)\frac{T_{j-1}}{\eta}
		\end{pmatrix}
		\Bigg]\\
		& \times \left( - \exp\left(-\left(\frac{T_j}{\alpha_i}\right)^{\eta}
		\right)  +   \exp\left(-\left(\frac{T_{j-1}}{\alpha_i}\right)^{\eta} \right) \right)^{\beta-1}
	\end{aligned}
	\end{equation}
	\end{minipage}
}
	%\begin{align*}
	%	 (\boldsymbol{z}_j-\boldsymbol{z}_{j-1})\pi_j(\boldsymbol{\theta})^{\beta-1} =& 
	%	\left( g_T(T_j )\begin{pmatrix}
	%		\frac{T_j}{\theta_0}\\
	%		T_jx_i - a_{i-1}^\ast
	%	\end{pmatrix} - g_T(T_{j-1} )\begin{pmatrix}
	%	\frac{T_{j-1}}{\theta_0}\\
	%	T_{j-1}x_i - a_{i-1}^\ast
	%\end{pmatrix} 
	%\right)(G_i(T_j)-G_i(T_{j-1}))^{\beta-1}\\
	%=& 
	%\bigg[ \theta_0 \exp(\theta_1x_i - \theta_0 \exp(\theta_1x_i)T_j)\begin{pmatrix}
	%	\frac{T_j}{\theta_0}\\
	%	T_jx_i - a_{i-1}^\ast
	%\end{pmatrix} \\ 
	% &-  \theta_0 \exp(\theta_1x_i - \theta_0 \exp(\theta_1x_i)T_{j-1})\begin{pmatrix}
	%	\frac{T_{j-1}}{\theta_0}\\
	%	T_{j-1}x_i - a_{i-1}^\ast
	%\end{pmatrix} 
	%\bigg]\\
	%& \left( - \exp( - \theta_0 \exp(\theta_1x_i)T_j) +   \exp( - \theta_0 \exp(\theta_1x_i)T_{j-1})\right)^{\beta-1}
	%\end{align*}
	with $T_j = t_j+h_{i-1}.$ All terms depending on times before $t_j$ are bounded and so taking limits on $T_j\rightarrow \infty,$ we get 
	$$
	\lim_{t_j \rightarrow \infty}(\boldsymbol{z}_j-\boldsymbol{z}_{j-1})\pi_j(\boldsymbol{\theta}_0)^{\beta-1} = \begin{cases}
		+ \infty & \text{if } \beta = 0,\\
		< \infty & \text{if } \beta > 0.\\
	\end{cases}
	$$
	%A similar reasoning gives the convergence of the consecutive terms, $l=j+1,..L+1,$ yielding to the same result. 
	Hence, the IF of the MDPDEs for positives values of $\beta$ is bounded when increasing any inspection time, whereas the IF of the MLE is unbounded for this class of leverage points.
	
	Similarly, let us consider a stress level $x_i$ and let $x_i \rightarrow \infty.$ We choose $t_j$ such that $t_j = \tau_{i-1},$ the time of stress change for the $i$-th stress level. 
	Again, since the stress levels are ordered, all subsequent  stress levels tend to infinity, and we then need to establish the boundedness of all terms from $j$ onwards.
	Here, the quantities depending on the stress level, such as $\alpha_l, h_{l-1}$ and $h_{l-1}^\ast$ for $l=i,...,k,$ are not constant. In particular, taking limits on the relations in (\ref{eq:loglinear}), (\ref{eq:ai}) and (\ref{aast}), we have 
	$$\lim_{ x_i \rightarrow \infty} \alpha_l =  \lim_{ x_i \rightarrow \infty} h_{l-1}  = \lim_{ x_i \rightarrow \infty} h^\ast_{l-1} = \begin{cases}
		0 & \text{if } a_1 \leq 0\\
		\infty & \text{if } a_1 > 0,\\
	\end{cases}$$
	%and taking limits in (\ref{eq:ai}) and (\ref{aast}) we get
	%	$$\lim_{ x_i \rightarrow \infty} h_{l-1}  = \lim_{ x_i \rightarrow \infty} h^\ast_{l-1} =\begin{cases}
	%		0 & \text{if } a_1 \leq 0\\
	%		\infty & \text{if } a_1 > 0.\\
	%	\end{cases}$$
	for $l=i,...,k.$
	As the parameter $a_1$ is assumed to be negative, all these quantities tend to zero.
	%Moreover, the limit behaviour of the IF of the parameters $\theta_0$ and $\theta_1$ may be different, 
	%so we discuss the boundedness of the IF of the MDPDE in the different possible scenarios.
	Note that times of stress change are points of discontinuity of the lifetime density function.
	We discuss the boundedness of each term in the summation (\ref{eq:secondfactorIF}) from $j$ onwards. For the $j$-th term, we have that
	all $\alpha_{i-1}, h_{i-2}$ and $h_{i-2}^\ast$ are bounded and so we get
	%	\begin{align*}
	%		&\lim_{ x_i \rightarrow \infty}
	%		\left( - g_i(T_j)
	%		T_j
	%		+ g_{i-1}(T_{j-1})
	%		T_{j-1}
	%		\right)(G_i(T_j)-G_{i-1}(T_{j-1}))^{\beta-1}\\
	%		=& \lim_{ x_i \rightarrow \infty}
	%		\bigg[ 
	%		T_{j}\frac{\eta}{\alpha_{i}}\left(\frac{T_j}{\alpha_i}\right)^{\eta} \exp\left(-\left(\frac{T_j}{\alpha_i}\right)^{\eta}\right)
	%		- T_{j-1} \frac{\eta}{\alpha_{i-1}}\left(\frac{T_{j-1}}{\alpha_{i-1}}\right)^{\eta} \exp\left(-\left(\frac{T_{j-1}}{\alpha_{i-1}}\right)^{\eta}\right)
	%		\bigg]\\
	%		&
	%		\left( - \exp\left(-\left(\frac{T_j}{\alpha_{i}}\right)^{\eta}
	%		\right)  +   \exp\left(-\left(\frac{T_{j-1}}{\alpha_{i-1}}\right)^{\eta} \right) \right)^{\beta-1} \begin{cases}
	%			- \infty & \text{if } \beta = 0,\\
	%			< \infty & \text{if } \beta > 0.\\
	%		\end{cases}
	%	\end{align*}
\resizebox{0.5\textwidth}{!}{%
	\begin{minipage}{0.45\textwidth}
	\begin{align*}
		&(\boldsymbol{z}_j-\boldsymbol{z}_{j-1})\pi_j(\boldsymbol{\theta}_0)^{\beta-1} =
		\\ & 
		 \Bigg(  g_i(t_j + h_{i-1})\begin{pmatrix}
			-(t_j + h_{i-1})\\
			-(t_j + h_{i-1})x_i + h_{i-1}^\ast \\
			\log\left(\frac{t_j + h_{i-1}}{\alpha_i}\right)\frac{t_j + h_{i-1}}{\eta}
		\end{pmatrix}\\
	& -  g_{i-1}(t_{j-1} + h_{i-2})\begin{pmatrix}
			-(t_{j-1}+ h_{i-2})\\
			-(t_{j-1}+ h_{i-2})x_i + h_{i-2}^\ast\\
			\log\left(\frac{t_{j-1}+ h_{i-2}}{\alpha_{i-1}}\right)\frac{t_{j-1}+ h_{i-2}}{\eta}
		\end{pmatrix}
		\Bigg)\\
		& \times (G_i(t_j + h_{i-1})-G_{i-2}(t_{j-1}+ h_{i-2}))^{\beta-1}
	\end{align*}
	\end{minipage}
}

\resizebox{0.5\textwidth}{!}{%
	\begin{minipage}{0.45\textwidth}
		\begin{align*}
		&(\boldsymbol{z}_j-\boldsymbol{z}_{j-1})\pi_j(\boldsymbol{\theta}_0)^{\beta-1} \\ =
		& \Bigg[ 
		\frac{\eta}{\alpha_{i}}\left(\frac{t_j + h_{i-1}}{\alpha_i}\right)^{\eta} \exp\left(-\left(\frac{t_j + h_{i-1}}{\alpha_i}\right)^{\eta}\right)\\
		& \times
		\begin{pmatrix}
			-(t_j + h_{i-1})\\
			-(t_j + h_{i-1})x_i + h_{i-1}^\ast\\
			\log\left(\frac{t_j + h_{i-1}}{\alpha_i}\right)\frac{t_j + h_{i-1}}{\eta}
		\end{pmatrix}
		-  \frac{\eta}{\alpha_{i-1}}\left(\frac{t_{j-1}+ h_{i-2}}{\alpha_{i-1}}\right)^{\eta}\\
		&\times \exp\left(-\left(\frac{t_{j-1}+ h_{i-2}}{\alpha_{i-1}}\right)^{\eta}\right)
		\begin{pmatrix}
			-(t_{j-1}+ h_{i-2})\\
			-(t_{j-1}+ h_{i-2})x_i + h_{i-1}^\ast \\
			\log\left(\frac{t_{j-1}+ h_{i-2}}{\alpha_{i-1}}\right)\frac{t_{j-1}+ h_{i-2}}{\eta}
		\end{pmatrix}
		\Bigg]\\
		& \times \left( - \exp\left(-\left(\frac{t_j+h_i}{\alpha_i}\right)^{\eta}
		\right)  +   \exp\left(-\left(\frac{t_{j-1}+ h_{i-2}}{\alpha_{i-1}}\right)^{\eta} \right) \right)^{\beta-1}\\
		\xrightarrow[x_i \rightarrow \infty]{}& \begin{cases}
			+ \infty & \text{if } \beta = 0,\\
			< \infty & \text{if } \beta > 0.\\
		\end{cases} 
	\end{align*}
\end{minipage}
}
	Now, choosing $t_j$ any time of stress change such that $t_j \geq \tau_{i},$ we have 
	\resizebox{0.5\textwidth}{!}{%
	\begin{minipage}{0.45\textwidth}
	\begin{align*}
		&(\boldsymbol{z}_j-\boldsymbol{z}_{j-1})\pi_j(\boldsymbol{\theta}_0)^{\beta-1} 
		=\\
		&\Bigg[ 
		\frac{\eta}{\alpha_{i}}\left(\frac{t_j + h_{i-1}}{\alpha_i}\right)^{\eta} \exp\left(-\left(\frac{t_j + h_{i-1}}{\alpha_i}\right)^{\eta}\right)\\
		& \times 
		\begin{pmatrix}
			-(t_j + h_{i-1})\\
			-(t_j + h_{i-1})x_i + h_{i-1}^\ast\\
			\log\left(\frac{t_j + h_{i-1}}{\alpha_i}\right)\frac{t_j + h_{i-1}}{\eta}
		\end{pmatrix}-  \frac{\eta}{\alpha_{i-1}}\left(\frac{t_{j-1}+ h_{i-2}}{\alpha_{i-1}}\right)^{\eta}\\
	& \times \exp\left(-\left(\frac{t_{j-1}+ h_{i-2}}{\alpha_{i-1}}\right)^{\eta}\right)
		\begin{pmatrix}
			-(t_{j-1}+ h_{i-2})\\
			-(t_{j-1}+ h_{i-2})x_i + h_{i-1}^\ast \\
			\log\left(\frac{t_{j-1}+ h_{i-2}}{\alpha_{i-1}}\right)\frac{t_{j-1}+ h_{i-2}}{\eta}
		\end{pmatrix}
		\Bigg]\\
		& \times \left( - \exp\left(-\left(\frac{t_j+h_i}{\alpha_i}\right)^{\eta}
		\right)  +   \exp\left(-\left(\frac{t_{j-1}+ h_{i-2}}{\alpha_{i-1}}\right)^{\eta} \right) \right)^{\beta-1}\\
		\xrightarrow[x_i \rightarrow \infty]{}& \begin{cases}
			+ \infty & \text{if } \beta = 0,\\
			< \infty & \text{if } \beta > 0.\\
		\end{cases} 
	\end{align*}
\end{minipage}
}
	Finally, if $t_j$ is any inspection time different than any time of stress change, then the previous limit vanishes, and thus the IF of the proposed MDPDE is bounded only for $\beta > 0.$ That is, the proposed estimators are also robust for all type of outliers, whereas the MLE lacks robustness against these leverage points.
	%of any MDPDE for the first parameter $\theta_0$ with $\beta\geq0$ is bounded,
	%$ (\boldsymbol{z}_j-\boldsymbol{z}_{j-1})\pi_j(\boldsymbol{\theta})^{\beta-1}$ tends to zero for all whereas the IF 
\end{remark}

\subsection{Influence function of Wald-type test statistics}

%Besides considering the influence function of the test statistic, they have also proposed to study the behaviour of the level and power of the test as functions of an additional observation at any point x – it reflects the influence of the additional infinitesimal contamination on the level and power of the test. An essential result of this approach is the approximation of the asymptotic level and power under a contaminated distribution in a neighbourhood of the null hypothesis. 

The influence function of a test statistic can be similarly derived as the first derivative of the statistic viewed as a functional, and it therefore measures the approximate impact on an additional observation to the underlying data.
We first determine the functional associated with the proposed Wald-type test statistic, $W_{N}(\widehat{\boldsymbol{\theta}}^\beta),$ under the null hypothesis in (\ref{eq:null}): 
\begin{align*}
	W_{N}&(\boldsymbol{T}_\beta(G)) =
	N \boldsymbol{m}^T(\boldsymbol{T}_\beta(G) )  \bigg(\boldsymbol{M}^T(\boldsymbol{T}_\beta(G) )\boldsymbol{J}_\beta^{-1}(\boldsymbol{T}_\beta(G) )\\
	& \times \boldsymbol{K}_\beta(\boldsymbol{T}_\beta(G) )\boldsymbol{J}_\beta^{-1}(\boldsymbol{T}_\beta(G) )\boldsymbol{M}(\boldsymbol{T}_\beta(G))\bigg)^{-1}
	\boldsymbol{m}(\boldsymbol{T}_\beta(G) ).
\end{align*}
Therefore, the IF of the proposed Wald-type test statistic can be easily derived from the IF of the MDPDE, as
\begin{align*}
	 \text{IF}&\left(\boldsymbol{n}, W_{N}, G\right) = \frac{\partial W_{N}(\boldsymbol{T}_\beta(G_\varepsilon))}{\partial \varepsilon}\bigg|_{\varepsilon = 0}\\
	= & 2N \boldsymbol{m}^T(\boldsymbol{T}_\beta(G) ) \bigg(\boldsymbol{M}^T(\boldsymbol{T}_\beta(G) )\boldsymbol{J}_\beta^{-1}(\boldsymbol{T}_\beta(G) )\boldsymbol{K}_\beta(\boldsymbol{T}_\beta(G) )\\
	& \times \boldsymbol{J}_\beta^{-1}(\boldsymbol{T}_\beta(G) )\boldsymbol{M}(\boldsymbol{T}_\beta(G))\bigg)^{-1}
	\boldsymbol{M}^T(\boldsymbol{T}_\beta(G) ) \text{IF}\left(\boldsymbol{n}, T_\beta, G\right).
\end{align*}
Under the null hypothesis in (\ref{eq:null}), the above expression vanishes and so  the  second order influence function becomes necessary.
\begin{align*}
	\text{IF}&\left(\boldsymbol{n}, W_{N}, G\right) =  \frac{\partial^2 W_{N}(\boldsymbol{T}_\beta(G_\varepsilon))}{\partial \varepsilon^2}\bigg|_{\varepsilon = 0}\\
	= & 2N \text{IF}^T\left(\boldsymbol{n}, T_\beta, G\right) \boldsymbol{M}^T(\boldsymbol{T}_\beta(G) ) \bigg(\boldsymbol{M}^T(\boldsymbol{T}_\beta(G)) \\
	& \times\boldsymbol{J}_\beta^{-1}(\boldsymbol{T}_\beta(G) ) \boldsymbol{K}_\beta(\boldsymbol{T}_\beta(G)) \boldsymbol{J}_\beta^{-1}(\boldsymbol{T}_\beta(G) )\boldsymbol{M}(\boldsymbol{T}_\beta(G))\bigg)^{-1}\\
	& \times \boldsymbol{M}^T(\boldsymbol{T}_\beta(G) ) \text{IF}\left(\boldsymbol{n}, T_\beta, G\right).
\end{align*}

As the second-order IFs of the proposed Wald-type tests are quadratic functions of the corresponding IFs of the MDPDEs, the boundedness of the IF of the Wald-type test statistic at a contamination point $\boldsymbol{n}$ and true distribution $F_{\boldsymbol{\theta}_0}$ can be discussed by the boundedness of the IF of the corresponding MDPDE and thus, the robust estimators do produce robust test statistics.

\section{Optimal tuning parameter}\label{sec:choicebeta}

As discussed in the previous sections, the tuning parameter of the MDPDEs controls the trade-off between efficiency and robustness: large values are more robust, although less efficient. 
Therefore, there is not an overall optimal value of the tuning parameter, but it depends on the amount of contamination in data, which is generally difficult to assess in real data. 
Moderately large values of $\beta$, over 0.3, generally produce MDPDEs with a suitable  compromise between efficiency and robustness, although the selection of such tuning parameter could greatly influence the performance of the estimator or statistic.
Warwick and Jones (2005) introduced an useful data-based procedure for the choice of the tuning parameter for the MDPDE, depending on a pilot estimator, $\boldsymbol{\theta}_P$. They assumed that optimal values of $\beta$ would produce estimators with the smallest estimation error, ans so they propose to minimize the estimated mean squared error (MSE) in the estimation given by
\begin{equation} \label{eq:choicebeta}
	\begin{aligned}
		\widehat{\operatorname{MSE}}\left(\beta\right)  =& \left(\widehat{\boldsymbol{\theta}}^\beta- \boldsymbol{\theta}_P\right)^T\left( \widehat{\boldsymbol{\theta}}^\beta- \boldsymbol{\theta}_P\right)\\
		& + \frac{1}{N} \operatorname{Tr}\left\{ \boldsymbol{J}_\beta(\widehat{\boldsymbol{\theta}}^\beta)^{-1} \boldsymbol{K}_\beta(\widehat{\boldsymbol{\theta}}^\beta) \boldsymbol{J}_\beta(\widehat{\boldsymbol{\theta}}^\beta)^{-1}\right\}
	\end{aligned}
\end{equation}
The previous MSE strongly depends on the pilot estimator through the bias term, and so it could influence the election of the optimal value. 
Indeed, 
%optimal estimators tend to be close to the pilot and so results are biased by that initial choice.
 Balakrishnan et al. (2020) implemented the Warwick and Jones algorithm for choosing the optimal tuning parameter of the MDPDE, for constant-stress ALTs under Weibull lifetime distribution for one-shot devices. 
They found that the algorithm is strongly dependent on the pilot estimator, as the resulting optimal estimators were very close to the pilot. Thus, the optimal estimators were biased by that initial choice, 
%moving very little away from this initial value. Therefore, 
and an alternative algorithm alleviating this strong pilot-dependence is needed. 
The Warwick and Jones algorithm was improved in Basak et al. (2021) by removing the strong dependency of the initial estimator. 
They proposed an iterative algorithm that replaces at each step the pilot estimator with the optimal MDPDE computed in the previous step.  This way, the dependency of the initial estimator fades as the algorithm proceeds.
The algorithm iterates until choice of the tuning parameter (or equivalently, the pilot estimator) is stabilized.
% The process should be initialized with a suitable robust pilot estimator, but the final choice of $\beta$ is quite more pilot-independent.
Basak et al. (2021) empirically showed that, in many statistical models, when the pilot estimators are within the MDPDE class, all robust pilots lead to the same iterated optimal value, and moreover the performance of the algorithm improves even with pure data.
Since the pilot-dependency of the original algorithm is especially strong for the proposed model, this behaviour is transferred to Basak et al.'s algorithm, although the latter manages to eliminate the bias to a great extent.
The above description is summarized in the following algorithm.\\

\begin{algorithm}%[Choice of the tuning parameter]
	\caption{Choice of the tuning parameter}\label{alg:cap}
	\begin{enumerate}
		\item Fix the convergence rate $\varepsilon$ and choose a pilot estimator $\boldsymbol{\theta}_P;$
		\item Compute the optimal value of the tuning parameter, $\beta^\ast,$ minimizing the estimated MSE in (\ref{eq:choicebeta});
		\item When the optimal estimate $\widehat{\boldsymbol{\theta}}^{\beta^\ast}$ differs from the pilot estimator more than the convergence rate,
		\begin{enumerate}
			\item Set $\boldsymbol{\theta}_P = \widehat{\boldsymbol{\theta}}^{\beta^\ast};$
			\item Minimize $\widehat{\operatorname{MSE}}\left(\beta\right)$ in (\ref{eq:choicebeta}) and update the optimal value of the tuning parameter, $\beta^\ast.$
			%\item If the optimal estimate $\widehat{\boldsymbol{\theta}}^{\beta^\ast}$ differs from the pilot estimator less than the convergence rate: Stop\\ Else: return to step 2.
		\end{enumerate}
	\end{enumerate}
\end{algorithm}

Balakrishnan et al. (2023) compared optimal values obtained using different initial tuning parameters, $\beta_0,$ for the step-stress ALT with non-destructive one-shot devices under exponential lifetimes, and showed empirically the independency of that initial choice for the exponential model. 
For Weibull lifetimes, the pilot-dependency is slightly preserved, and so the initial estimator is crucial for a proper performance. We use as pilot estimator an average of MDPDEs obtained with different values of $\beta$.
It is of interest to note that large  values of  the optimal tuning parameter could imply high contamination in data, and so the proposed algorithm is also useful in assessing the amount of  contamination in the data.

\section{Simulation study}\label{sec:simstudy}

In this section, we examine the behaviour of the proposed robust MDPDEs and Wald-type tests
% MDPDE, $\widehat{\boldsymbol{\theta}}^\beta,$ and the Z-type and Rao-type test statistics, $Z_{N}(\widehat{\boldsymbol{\theta}}^\beta)$ and $R_{\beta, N}(\widetilde{\boldsymbol{\theta}}^\beta)$
for the SSALT model with Weibull lifetime  distribution in terms of efficiency and robustness.

To evaluate the robustness of the estimators and tests, we introduce contamination to the multinomial data by increasing (or decreasing)  the probability of failure in  (\ref{eq:th.prob}) for (at least) one interval (i.e., one cell); that is, we consider ``outlying cells'' instead of  ``outlying devices'' (see Balakrishnan et al. (2019a)). 
The  probability of failure at the contaminated cell is computed as
\begin{equation} \label{eq:contaminatedprob}
	\tilde{\pi}_j(\boldsymbol{\theta}) =  G_{\boldsymbol{\tilde{\theta}}}(\textit{IT}_j) - G_{\boldsymbol{\tilde{\theta}}}(\textit{IT}_{j-1})
\end{equation}
for some $j=2,...,L$, where $\boldsymbol{\tilde{\theta}} = (\widetilde{a}_0,\widetilde{a}_1, \widetilde{\eta})$ is a contaminated parameter with $\tilde{a}_0 \geq a_0,$ $\tilde{a}_1 \geq a_1$  and $\tilde{\eta} \geq \eta.$ 
After increasing (or decreasing) the probability of failure at one cell, the probability vector of the multinomial model, $\tilde{\boldsymbol{\pi}}(\boldsymbol{\theta}),$  must be normalised to add up to 1. 
%explain that model with great outliers tarnish the estimation

We will consider a 2-step stress ALT experiment with $L=13$ inspection times and a total of $N = 200$ devices under test. 
At the start of the experiment, all the devices are tested at a stress level $x_1=30$ until the first time of stress change $\tau_1 = 18.$ Then, the surviving units are subjected to an increased stress level, $x_2=40$ until the end of the experiment at $\tau_2=52.$ 
During the experiment, the numbers of failures are recorded at $\text{IT}=(6,10,14,18,20,24,28,32,36,40,44,48,52).$

The  value of the true parameter is set to be $\boldsymbol{\theta}_0=(5.3, -0.05, 1.5)$ and  the data are generated from the corresponding
multinomial model described in Section \ref{sec:modelformulation}, by assuming Weibull lifetime distributions. 
Moreover, we contaminate the data by increasing the probability of failure in the third interval using (\ref{eq:contaminatedprob}). We consider three scenarios of contamination corresponding to the increase of each of the three model parameters, $a_0, a_1$ and $\eta,$ respectively.

\subsection{Minimum density power divergence estimators} \label{sec:simstudyMDPPE}

We examine the accuracy of the proposed MDPDEs in different scenarios of contamination by means of a Monte Carlo simulation study. We calculate the root mean square error  (RMSE)  of the MDPDE for different values of $\beta \in \{0,0.2,0.4,0.6,0.8,1\},$ including the MLE for $\beta = 0,$ over $R=1000$ repetitions. 
Figure \ref{fig:estimatorscont} presents the results when contaminating each of the model parameters. The abscissa axis indicates the value of the corresponding contaminated parameter, while the remaining two model parameters keep their true value. The grid of contaminated parameters are chosen to increase the value of failure on the third cell of the model.
The empirical results show that the MLE is the most efficient estimator when there is no contamination. However, it is highly sensitive to outliers and so, when contamination is introduced in data by switching any of the model parameters, the MDPDEs with positive values of the DPD tuning parameter outperform the classical estimator. The greater the value of the tuning $\beta$ is, the more robust but less efficient the resulting estimator is. In the most extreme case considered here, the MDPDE with $\beta=1$ is clearly not efficient but is much less influenced by outliers compared to the other estimators. Therefore, MDPDEs with moderate values of $\beta$ offer a good trade-off between efficiency and robustness.

	Moreover, the optimal tuning parameter selected using Algorithm 1 is close to zero when there is no contamination and it increases when contamination is introduced. Therefore, the resulting optimal estimator performs similarly to the MLE in the absence of contamination but matches the performance of the MDPDEs with moderate values of $\beta$ (between 0.4 and 0.6) when contamination is present. That is, the MDPDE with optimal $\beta$ effectively provides an efficient estimator in the absence of contamination and a robust estimator under contamination scenarios.

\begin{figure}
	\centering
	\includegraphics[height=5.3cm, width=6.3cm]{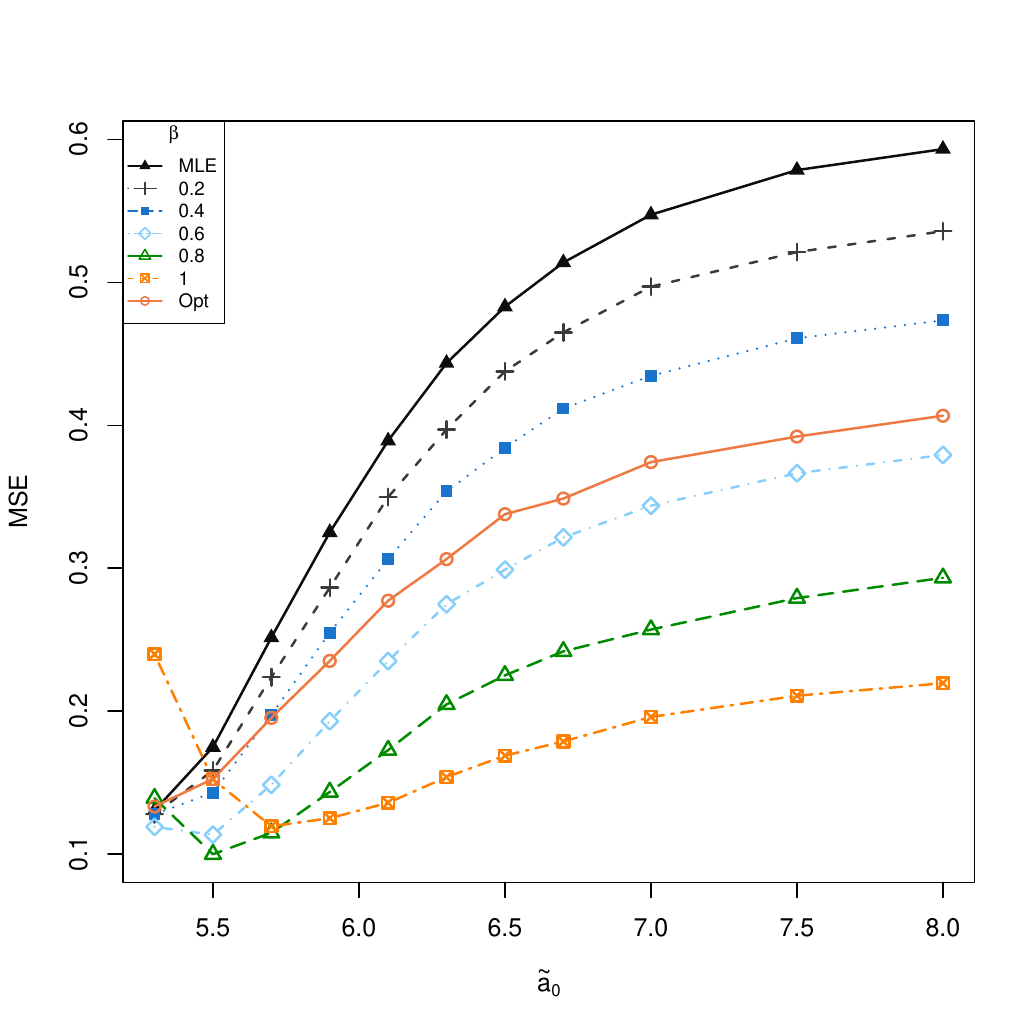}	
	\includegraphics[height=5.3cm, width=6.3cm]{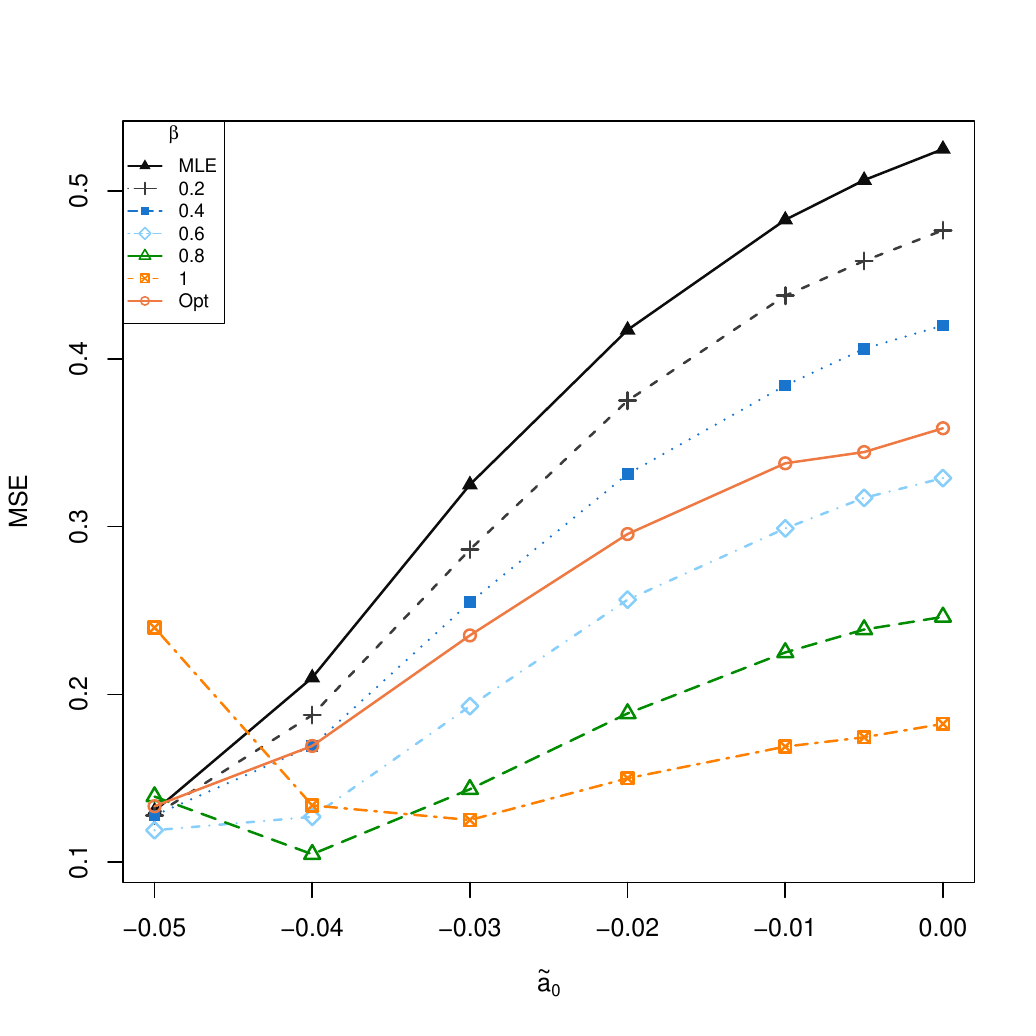}	
	\includegraphics[height=5.3cm, width=6.3cm]{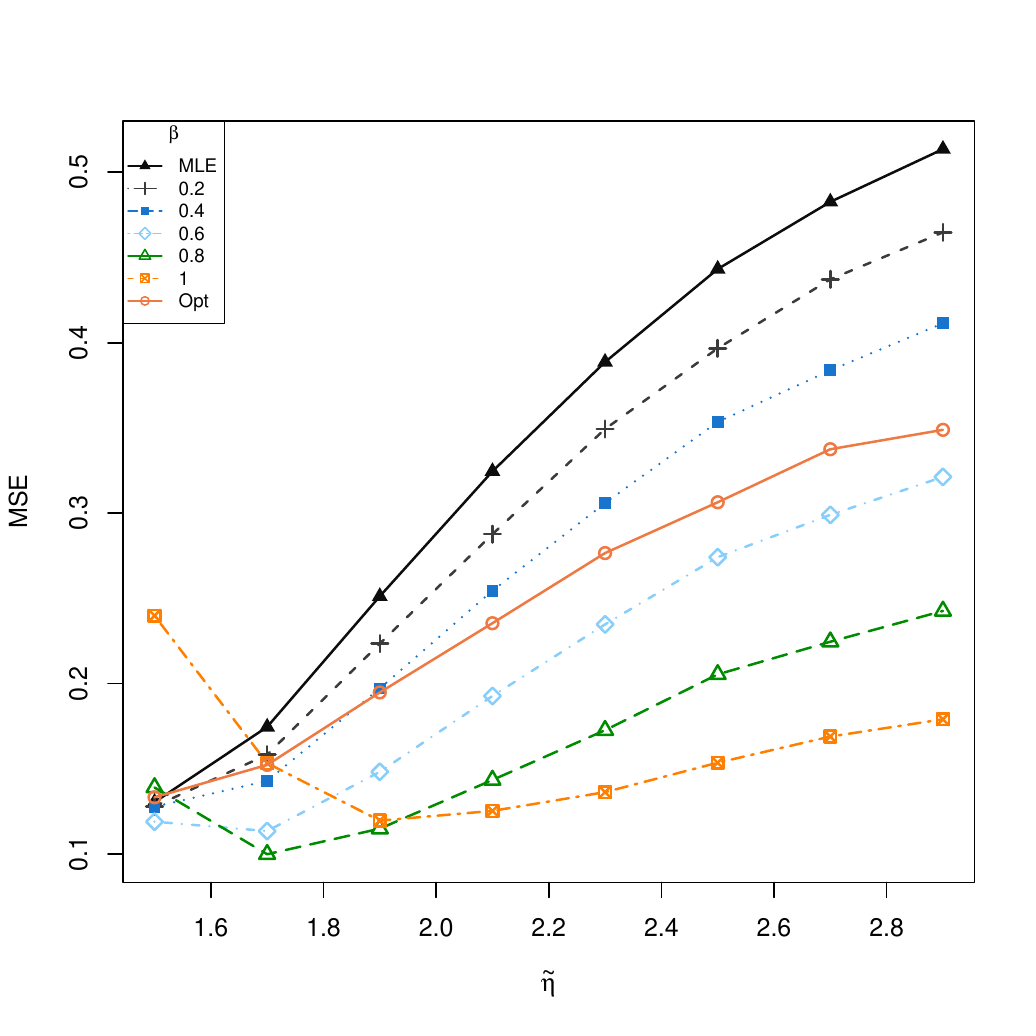}
	\caption{RMSE of the estimates under increasing contamination; $a_0$-contamination (top), $a_1$-contamination (middle) and $\eta$-contamination (bottom)}	
	\label{fig:estimatorscont}
\end{figure}

We additionally compute the mean square error (MSE) of the reliability estimates at a mission time $t_0=40$
%and mean lifetime of devices under a lower stress level $x_0=20$ 
in Figure \ref{fig:reliabilitycont}. Results for the mean time to failure are provided in Appendix \ref{app:MSEmean}. % and \ref{fig:meancont}, respectively.
% indicating the performance of the estimated model under normal operating conditions. 
The robustness properties of the estimators are inherited by the lifetime characteristic estimates. Consequently, the reliability estimates under contaminated scenarios are more accurate for positive values of $\beta$, demonstrating the robustness of the estimators, performing competitively with the MLE in the absence of contamination.
%As expected, the estimated reliability under contaminated scenarios are more accurate for positive values of $\beta$. 

\begin{figure}
	\centering
		\includegraphics[height=5.3cm, width=6.3cm]{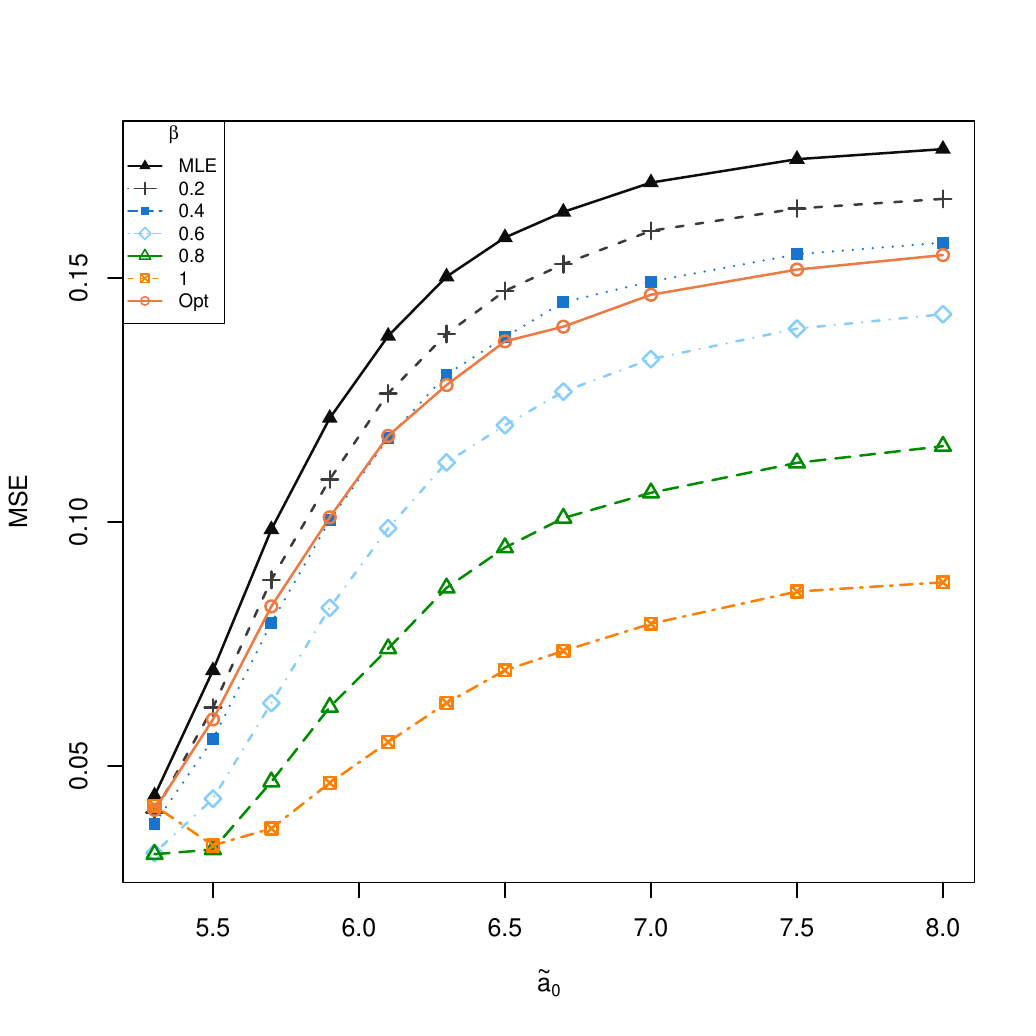}
		\includegraphics[height=5.3cm, width=6.3cm]{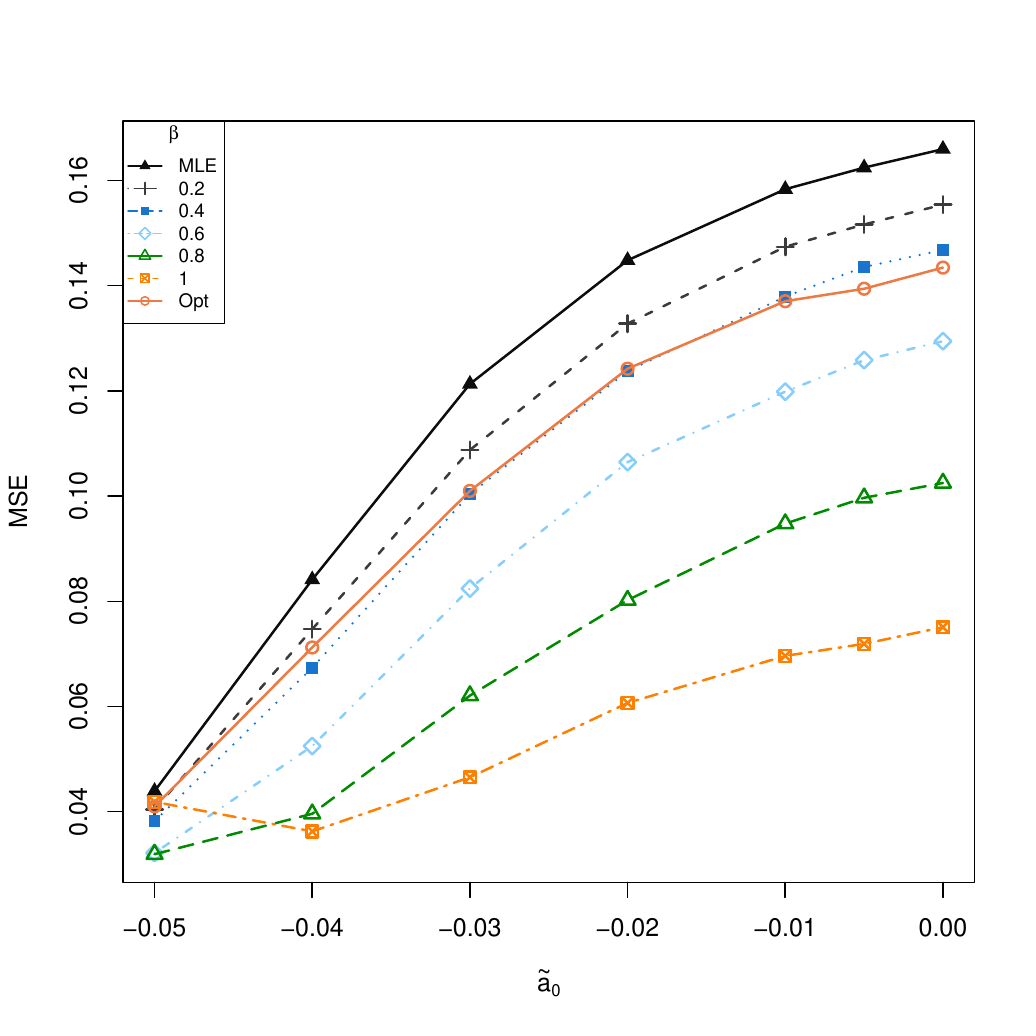}
		\includegraphics[height=5.3cm, width=6.3cm]{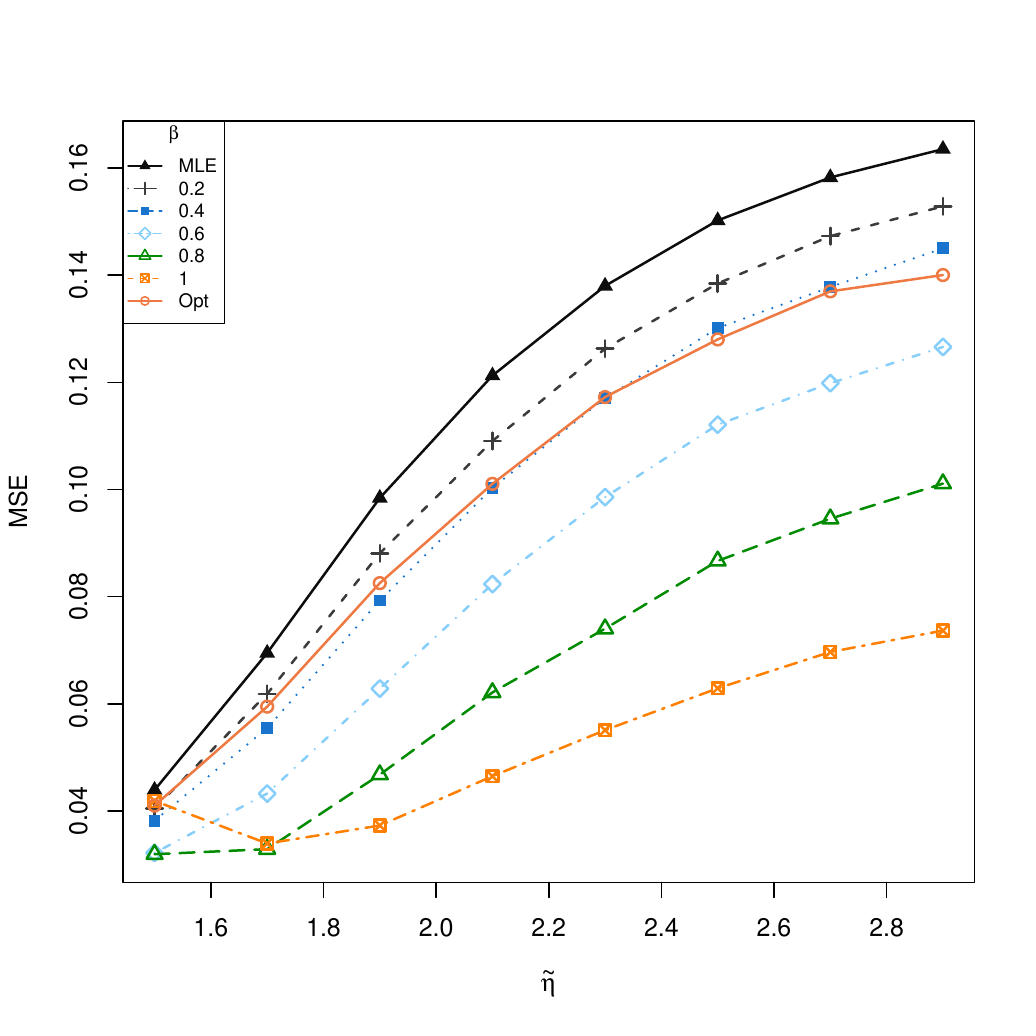}
	\caption{MSE of the reliability for $t=40$  under normal operating conditions $x_0=20$ when the contamination is introduced on the first (top) and second (bottom) model parameters.}	
	\label{fig:reliabilitycont}
\end{figure}

Further, we compare the direct and transformed asymptotic CIs of reliability at mission time $t_0=40$ and mean lifetime under a constant stress level $x_0=20$ in terms of coverage probability of the resulting intervals, in the three different scenarios of contamination.  The empirical coverage probabilities for the reliability and mean time to failure estimates are given in Tables \ref{table:CPreliabilitya0} and \ref{table:CPreliabilitya0}, respectively.
As shown, CIs based on MDPDEs with large values of the tuning parameter are more robust than CIs based on the classical MLE, while they perform competitively with the classical estimator under pure data scenarios. The nominal level of 95\% is not achieved by any approximate method, but transformed CIs have higher coverage probabilities, offering better performance. The difference in coverage between both methods is especially striking for highly contaminated scenarios. In these cases, direct CIs based on the MLE cover less than 15\%, while their corresponding transformed CIs include the true value of the characteristics around 20\% to 30\% times. In heavily contaminated scenarios, the advantage of using robust methods is clearly shown, as they maintain a coverage of around 75\% (for transformed CIs) in contrast to the 20-30\% achieved by the CI based on the MLE.

\begin{table}[ht]
	\centering
	\caption{Coverage of the direct and transformed CIs for reliability under contamination in the first parameter, for different values of the tuning parameter. }
	\label{table:CPreliabilitya0}
	\begin{tabular}{ccccccc}
		$a_0$ & 5.3 & 5.7 & 6 & 6.5 & 7 & 8 \\ 
		\hline
		$\beta$ & \multicolumn{6}{c}{Direct CI}\\
		\hline
		0 & 0.756 & 0.462 & 0.288 & 0.164 & 0.124 & 0.108 \\ 
		 0.2 & 0.776 & 0.520 & 0.338 & 0.234 & 0.180 & 0.156 \\ 
		 0.4 & 0.796 & 0.564 & 0.384 & 0.254 & 0.204 & 0.162 \\ 
	0.6 & 0.820 & 0.656 & 0.498 & 0.326 & 0.246 & 0.220 \\ 
		 0.8 & 0.778 & 0.750 & 0.634 & 0.478 & 0.382 & 0.344 \\ 
		 1 & 0.686 & 0.784 & 0.726 & 0.596 & 0.568 & 0.518 \\ 
		Optimal & 0.772 & 0.536 & 0.370 & 0.232 & 0.204 & 0.158 \\ 
		\hline
		$\beta$ & \multicolumn{6}{c}{Transformed CI}\\
		\hline
	0& 0.908 & 0.740 & 0.556 & 0.376 & 0.272 & 0.238  \\ 
	 0.2 & 0.900 & 0.758 & 0.598 & 0.424 & 0.348 & 0.296 \\ 
	0.4 & 0.892 & 0.796 & 0.652 & 0.488 & 0.398 & 0.346 \\ 
	 0.6 & 0.896 & 0.862 & 0.750 & 0.588 & 0.504 & 0.434 \\ 
		 0.8 & 0.854 & 0.910 & 0.832 & 0.704 & 0.642 & 0.578 \\ 
		 1 & 0.728 & 0.900 & 0.856 & 0.780 & 0.748 & 0.722 \\ 
		Optimal& 0.896 & 0.780 & 0.640 & 0.468 & 0.384 & 0.332 \\ 
		\hline
	\end{tabular}
\end{table}

\begin{table}[ht]
	\centering
	\caption{Coverage of the direct and transformed CIs for  reliability under contamination in the second parameter, for different values of the tuning parameter. }
	\label{table:CPreliabilitya1}
	\begin{tabular}{ccccccc}
		$a_1$ & -0.05 & -0.03 & -0.02 & -0.01 & -0.005 & 0 \\ 
		\hline
		$\beta$ & \multicolumn{6}{c}{Direct CI}\\
		\hline 
		0& 0.756 & 0.332 & 0.228 & 0.164 & 0.150 & 0.142 \\ 
		0.2 & 0.776 & 0.404 & 0.286 & 0.234 & 0.200 & 0.186 \\ 
	 0.4 & 0.796 & 0.462 & 0.320 & 0.254 & 0.222 & 0.202 \\ 
		 0.6 & 0.820 & 0.548 & 0.404 & 0.326 & 0.290 & 0.270 \\ 
	 0.8 & 0.778 & 0.654 & 0.570 & 0.478 & 0.438 & 0.398 \\ 
		 1 & 0.686 & 0.764 & 0.660 & 0.596 & 0.600 & 0.580 \\ 
		Optimal& 0.772 & 0.410 & 0.290 & 0.232 & 0.224 & 0.202 \\
		\hline
		$\beta$ & \multicolumn{6}{c}{Transformed CI}\\
		\hline
		0& 0.908 & 0.596 & 0.462 & 0.376 & 0.336 & 0.312 \\ 
	0.2 & 0.900 & 0.654 & 0.506 & 0.424 & 0.392 & 0.380 \\ 
		 0.4 & 0.892 & 0.698 & 0.562 & 0.488 & 0.448 & 0.432 \\ 
	 0.6 & 0.896 & 0.796 & 0.658 & 0.588 & 0.544 & 0.526 \\ 
	 0.8 & 0.854 & 0.842 & 0.778 & 0.704 & 0.678 & 0.670 \\ 
	 1 & 0.728 & 0.874 & 0.826 & 0.780 & 0.758 & 0.750 \\ 
		Optimal& 0.896 & 0.686 & 0.552 & 0.468 & 0.436 & 0.418 \\ 
		\hline
	\end{tabular}
\end{table}

\begin{table}[ht]
	\centering
	\caption{Coverage of the direct and transformed CIs for reliability under contamination in the third parameter, for different values of the tuning parameter. }
	\label{table:CPreliabilityeta}
	\begin{tabular}{ccccccc}
		\hline $\eta$ & 1.5 & 1.7 & 2 & 2.2&2.7 & 3 \\ 
		\hline
		$\beta$ & \multicolumn{6}{c}{Direct CI}\\
		\hline
		MLE& 0.756 & 0.620 & 0.400 & 0.290 & 0.164 & 0.144 \\ 
		 0.2 & 0.776 & 0.690 & 0.458 & 0.344 & 0.234 & 0.186 \\ 
		0.4 & 0.796 & 0.714 & 0.518 & 0.388 & 0.256 & 0.202 \\ 
		0.6 & 0.820 & 0.780 & 0.590 & 0.498 & 0.326 & 0.270 \\ 
	 0.8 & 0.778 & 0.824 & 0.718 & 0.634 & 0.478 & 0.398 \\ 
		1 & 0.686 & 0.782 & 0.770 & 0.728 & 0.596 & 0.582 \\ 
		Optimal& 0.772 & 0.678 & 0.478 & 0.366 & 0.232 & 0.204 \\ 
		\hline
	$\beta$ & \multicolumn{6}{c}{Transformed CI}\\
		\hline
	0 & 0.908 & 0.856 & 0.678 & 0.556 & 0.376 & 0.312 \\ 
	0.2 & 0.900 & 0.862 & 0.700 & 0.598 & 0.426 & 0.380 \\ 
	 0.4 & 0.892 & 0.888 & 0.752 & 0.652 & 0.492 & 0.432 \\ 
	0.6 & 0.896 & 0.916 & 0.818 & 0.750 & 0.588 & 0.526 \\ 
	 0.8 & 0.854 & 0.906 & 0.874 & 0.828 & 0.704 & 0.670 \\ 
	 1 & 0.728 & 0.858 & 0.884 & 0.860 & 0.780 & 0.750 \\ 
		Optimal& 0.896 & 0.862 & 0.738 & 0.638 & 0.470 & 0.418 \\ 
		\hline
	\end{tabular}
\end{table}

\begin{table}[ht]
	\centering
	\caption{Coverage of the direct and transformed CIs for mean lifetime under contamination in the first parameter, for different values of the tuning parameter. }
	\label{table:CPmeana0}
	\begin{tabular}{ccccccc}
		\hline
		$\beta $ $\backslash$ $a_0$ & 5.3 & 5.7 & 6 & 6.5 & 7 & 8 \\ 
		\hline
	$\beta$ & \multicolumn{6}{c}{Direct CI}\\
		\hline
	0 & 0.756 & 0.462 & 0.288 & 0.164 & 0.124 & 0.108 \\ 
	 0.2 & 0.776 & 0.520 & 0.338 & 0.234 & 0.180 & 0.156 \\ 
	 0.4 & 0.796 & 0.564 & 0.384 & 0.254 & 0.204 & 0.162 \\ 
	 0.6 & 0.820 & 0.656 & 0.498 & 0.326 & 0.246 & 0.220 \\ 
	 0.8 & 0.778 & 0.750 & 0.634 & 0.478 & 0.382 & 0.344 \\ 
	 1 & 0.686 & 0.784 & 0.726 & 0.596 & 0.568 & 0.518 \\ 
		Optimal& 0.772 & 0.536 & 0.370 & 0.232 & 0.204 & 0.158 \\ 
		\hline
	$\beta$ & \multicolumn{6}{c}{Transformed CI}\\
		\hline
	0& 0.932 & 0.726 & 0.546 & 0.366 & 0.276 & 0.234 \\ 
		0.2 & 0.938 & 0.754 & 0.598 & 0.416 & 0.348 & 0.304 \\ 
		 0.4 & 0.952 & 0.792 & 0.642 & 0.496 & 0.412 & 0.354 \\ 
	0.6 & 0.970 & 0.868 & 0.760 & 0.584 & 0.510 & 0.476 \\ 
		 0.8 & 0.954 & 0.926 & 0.832 & 0.716 & 0.656 & 0.596 \\ 
		 1 & 0.880 & 0.944 & 0.888 & 0.794 & 0.756 & 0.730 \\ 
		Optimal& 0.948 & 0.770 & 0.616 & 0.454 & 0.390 & 0.332 \\ 
		\hline
	\end{tabular}
\end{table}

\begin{table}[ht]
	\centering
	\caption{Coverage of the direct and transformed CIs for  mean lifetime under contamination in the second parameter, for different values of the tuning parameter. }
	\label{table:CPmeana1}
	\begin{tabular}{ccccccc}
		\hline
		$\beta $ $\backslash$ $a_1$ & -0.05 & -0.03 & -0.02 & -0.01 & -0.005 & 0 \\ 
		\hline
	$\beta$ & \multicolumn{6}{c}{Direct CI}\\
		\hline
	0 & 0.756 & 0.332 & 0.228 & 0.164 & 0.150 & 0.142 \\ 
		0.2 & 0.776 & 0.404 & 0.286 & 0.234 & 0.200 & 0.186 \\ 
		 0.4 & 0.796 & 0.462 & 0.320 & 0.254 & 0.222 & 0.202 \\ 
	 0.6 & 0.820 & 0.548 & 0.404 & 0.326 & 0.290 & 0.270 \\ 
	 0.8 & 0.778 & 0.654 & 0.570 & 0.478 & 0.438 & 0.398 \\ 
	 1 & 0.686 & 0.764 & 0.660 & 0.596 & 0.600 & 0.580 \\ 
		Optimal& 0.772 & 0.410 & 0.290 & 0.232 & 0.224 & 0.202 \\ 
		\hline
	$\beta$ & \multicolumn{6}{c}{Transformed CI}\\
		\hline
	0 & 0.932 & 0.574 & 0.452 & 0.366 & 0.340 & 0.318 \\ 
	 0.2 & 0.938 & 0.642 & 0.516 & 0.416 & 0.388 & 0.374 \\ 
	 0.4 & 0.952 & 0.702 & 0.572 & 0.496 & 0.454 & 0.440 \\ 
	 0.6 & 0.970 & 0.790 & 0.654 & 0.584 & 0.554 & 0.530 \\ 
	0.8 & 0.954 & 0.856 & 0.790 & 0.716 & 0.684 & 0.682 \\ 
	 1 & 0.880 & 0.896 & 0.856 & 0.794 & 0.788 & 0.782 \\ 
		Optimal& 0.948 & 0.666 & 0.556 & 0.454 & 0.436 & 0.418 \\ 
		\hline
	\end{tabular}
\end{table}

\begin{table}[ht]
	\centering
	\caption{Coverage of the direct and transformed CIs for mean lifetime under contamination in the third parameter, for different values of the tuning parameter. }
	\label{table:CPmeaneta}
	\begin{tabular}{ccccccc}
		\hline
		$\beta $ $\backslash$ $\eta$ & 1.5 & 1.7 & 2 & 2.2&2.7 & 3 \\ 
		\hline
	$\beta$ & \multicolumn{6}{c}{Direct CI}\\
		\hline
	0 & 0.756 & 0.620 & 0.400 & 0.290 & 0.164 & 0.144 \\ 
	 0.2 & 0.776 & 0.690 & 0.458 & 0.344 & 0.234 & 0.186 \\ 
	 0.4 & 0.796 & 0.714 & 0.518 & 0.388 & 0.256 & 0.202 \\ 
	 0.6 & 0.820 & 0.780 & 0.590 & 0.498 & 0.326 & 0.270 \\ 
	 0.8 & 0.778 & 0.824 & 0.718 & 0.634 & 0.478 & 0.398 \\ 
 1 & 0.686 & 0.782 & 0.770 & 0.728 & 0.596 & 0.582 \\ 
		Optimal& 0.772 & 0.678 & 0.478 & 0.366 & 0.232 & 0.204 \\ 
		\hline
	$\beta$ & \multicolumn{6}{c}{Transformed CI}\\
		\hline
	0 & 0.932 & 0.840 & 0.660 & 0.544 & 0.368 & 0.318 \\ 
	 0.2 & 0.938 & 0.860 & 0.706 & 0.598 & 0.418 & 0.374 \\ 
	 0.4 & 0.952 & 0.898 & 0.754 & 0.642 & 0.498 & 0.440 \\ 
	0.6 & 0.970 & 0.944 & 0.820 & 0.762 & 0.584 & 0.530 \\ 
	0.8 & 0.954 & 0.966 & 0.892 & 0.830 & 0.716 & 0.682 \\ 
	 1 & 0.880 & 0.940 & 0.924 & 0.892 & 0.794 & 0.782 \\ 
		Optimal & 0.948 & 0.878 & 0.748 & 0.614 & 0.456 & 0.418 \\ 
		\hline
	\end{tabular}
\end{table}

In addition, we test different scenarios of contamination. In the first scenario, we generate an outlying cell in the third interval by decreasing the value of the first parameter, $\theta_0,$ and at the second we perform similarly, but decreasing the second parameter $\theta_1$. In both cases, the lifetime rate $\lambda(\widetilde{\boldsymbol{\theta}})$ gets decreased; the smaller is the contaminated parameter, greater is the contamination. 

\subsection{Wald-type tests} \label{sec:wtestsimulation}

We empirically evaluate the performance of the Wald-type test statistics based on the MDPDEs for different values of the tuning parameter in terms of empirical level and power.
We set the simple step-stress model presented in Section \ref{sec:simstudy}, with $N=200$ devices and $L=13$ inspection times.
We consider the hypothesis testing problem
\begin{equation} \label{empiricaltest}
	\operatorname{H}_0: a_1 = -0.05 \hspace{0.5cm} \text{vs} \hspace{0.5cm} \operatorname{H}_1: a_1\neq -0.05.
\end{equation} 
and we fix the true value of the model parameter as $\boldsymbol{\theta}_0 = (5.3,-0.05,1.5)$  verifying the null hypothesis (when computing the empirical level) and 
$\boldsymbol{\theta}_0 = (5.3,-0.06,1.5)$ when computing the empirical power.
The Wald-type test statistic associated with the test in (\ref{empiricaltest}) based on the MDPDE, $\widehat{\boldsymbol{\theta}}^\beta,$ is defined using (\ref{eq:Wtypest}) with $\boldsymbol{m}(\boldsymbol{\theta}) = a_1-0.05$ and the critical region of the test is  as given in (\ref{eq:criticalregionWtest}).

Moreover, to examine the robustness of the test, we contaminate the third cell following the three contamination scenarios discussed in Section \ref{sec:simstudy}.
Figures \ref{fig:level} and \ref{fig:power}  show, respectively, the empirical  level  and power of the tests against cell contamination, when introducing contamination in the first parameter, $a_0,$ (top), the second parameter, $a_1$  (middle), are the third (bottom). 
The empirical level and power are computed as the proportion of rejected Wald-type test statistics over $R=500$ replications of the model under the null hypothesis for a significance level of $\alpha = 0.05$.

All Wald-type tests, based on MDPDEs with different values of the tuning parameter, perform similarly in the absence of contamination. However, Wald-type test statistics based on MDPDEs with large values of the tuning parameter are clearly more robust than the classical MLE in terms of empirical level, showing strong robustness in all contaminated scenarios. 
Conversely, the overall performance of all Wald-type test statistics, based on different values of the tuning parameter, is quite similar in terms of power.
When the contamination is introduced in the two parameters which are not under test, $a_0$ and $\eta,$ the power of the test is higher for low values of the tuning parameter $\beta$, including the MLE. However, when introducing contamination on the second parameter, $a_1,$ Wald-type test statistics based on MDPDE with large values of $\beta$ outperform the classical MLE in terms of robustness.
So, Wald-type test statistics based on the MDPDE with moderately large value of the tuning parameter offer an appealing alternative for the classical Wald-type test statistic based on the MLE, with a clear gain in robustness in terms of level while remaining competitive in terms of power.

\begin{figure}[ht]
	\centering	
		\includegraphics[height=5.3cm, width=6.3cm]{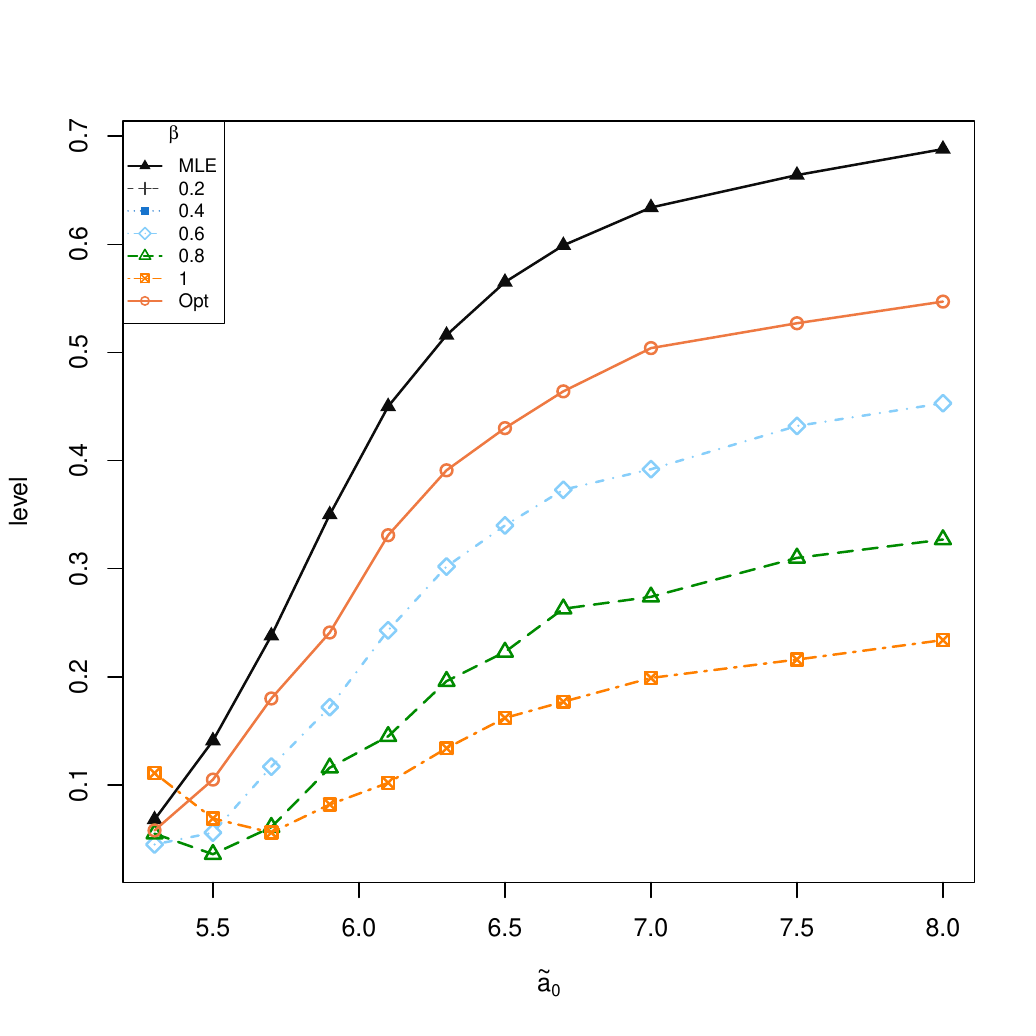}
		\includegraphics[height=5.3cm, width=6.3cm]{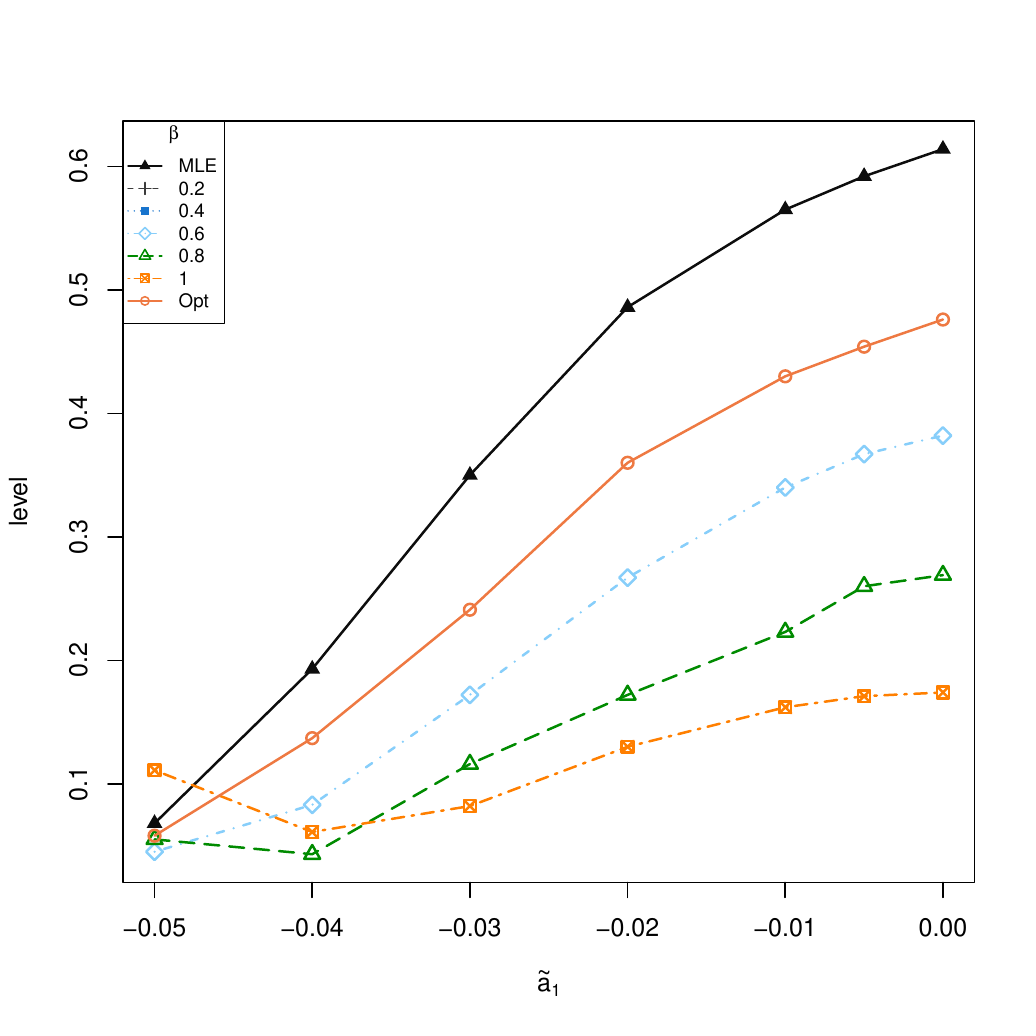}
		\includegraphics[height=5.3cm, width=6.3cm]{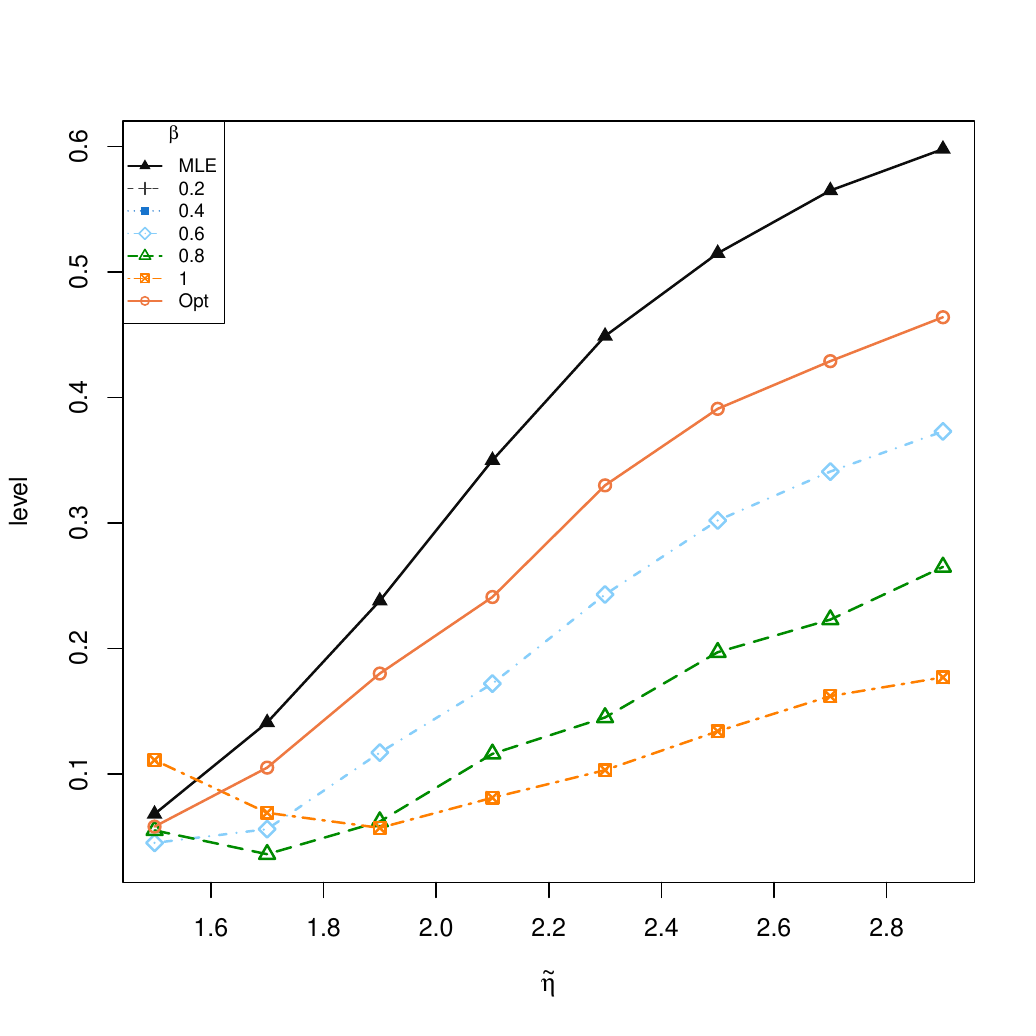}	
	\caption{Empirical level of the Wald-type test under increasing contamination on the first parameter $a_0$ (top), second $a_1$ (middle) and shape parameter $\nu$ (bottom).}	
	\label{fig:level}
\end{figure}

\begin{figure}[ht]
	\centering	
		\includegraphics[height=5.3cm, width=6.3cm]{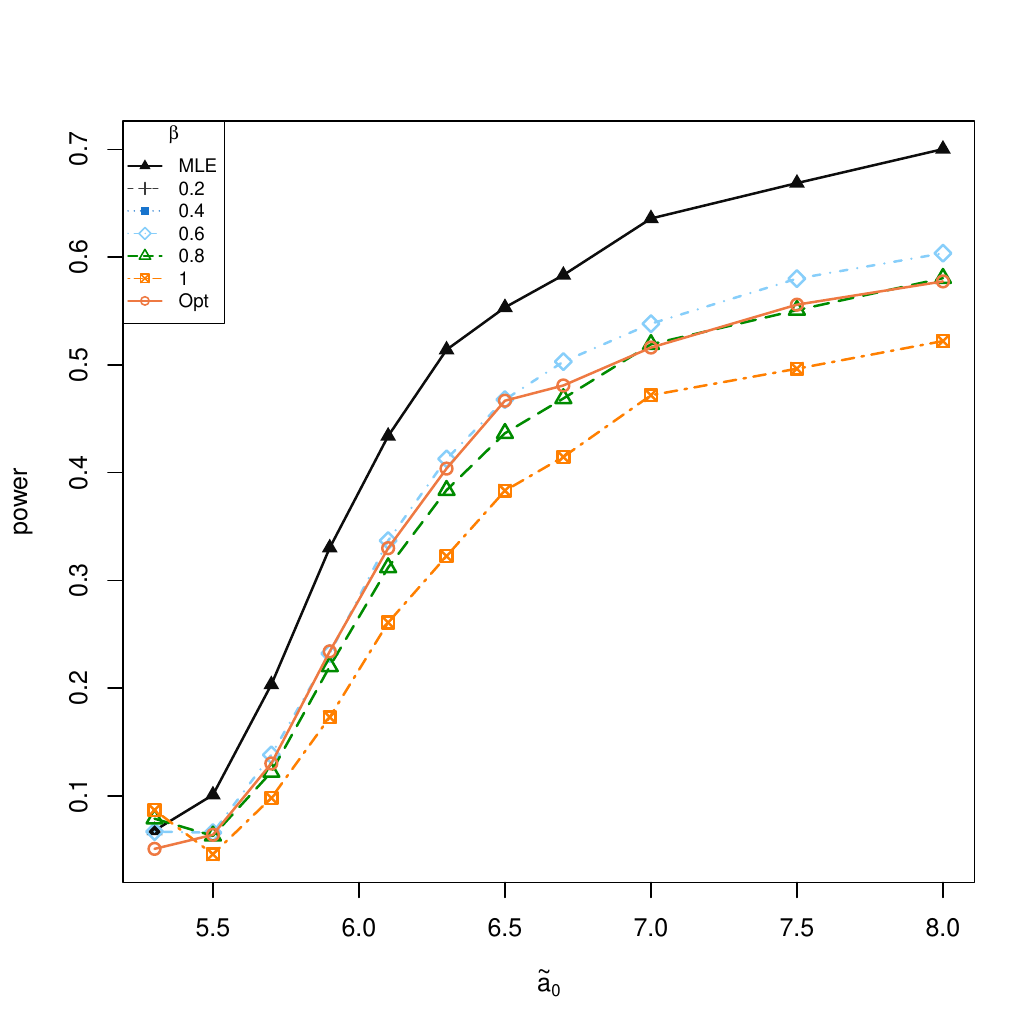}
		\includegraphics[height=5.3cm, width=6.3cm]{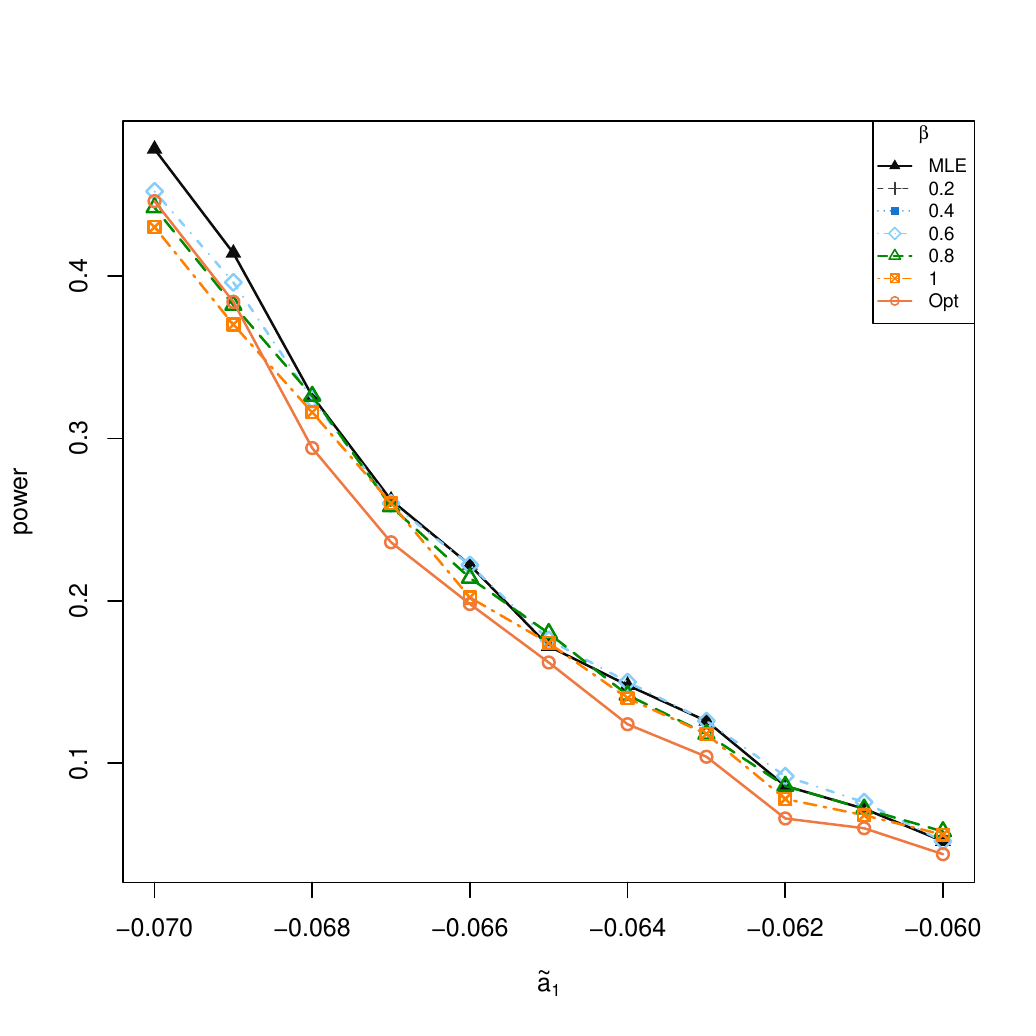}
		\includegraphics[height=5.3cm, width=6.3cm]{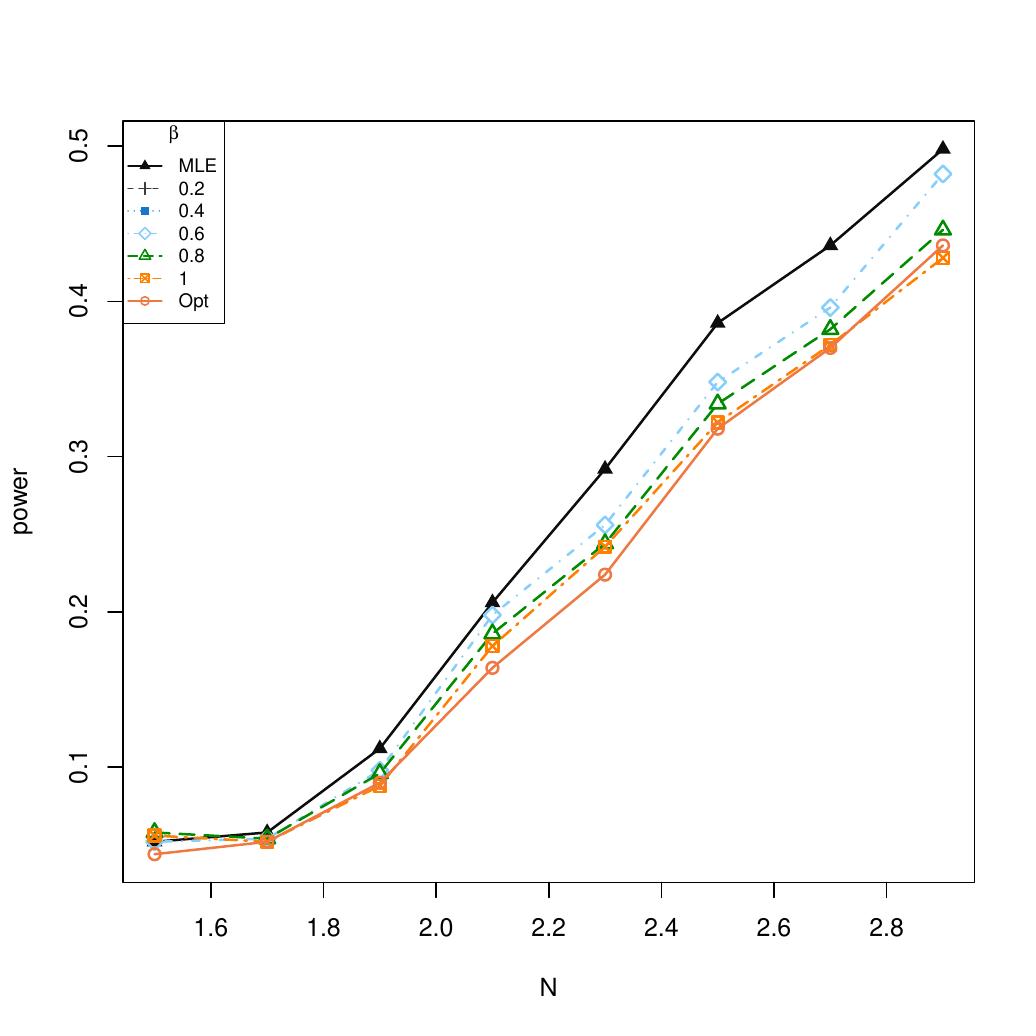}	
	\caption{Empirical power  of the Wald-type test under increasing contamination on the first parameter $a_0$ (top), second $a_1$ (middle) and shape parameter $\nu$ (bottom).}	
	\label{fig:power}
\end{figure}

\section{Real data analyses} \label{sec:realdata}

We now illustrate the performance of the proposed robust methods in analyzing the influence of temperature in solar lighting devices (dataset \ref{Dataset1}) 
%medium power silicon bipolar transistors (dataset \ref{Dataset2})
 and LED lights (dataset \ref{Dataset3}).

\subsection{Effect of temperature on solar lighting devices}

We fit the step-stress ALT model with Weibull lifetime distributions for the first dataset \ref{Dataset1} introduced in Section \ref{sec:Motivating}. 
Before fitting the model, the stress levels were normalized to $(0,1).$ 

Table \ref{table:solarlight} shows the estimated values of the model parameters with different values of the tuning parameter $\beta.$ 
Note that the MDPDE estimates are greater for the scale parameters under both stress levels $\lambda_i = \exp(a_0+a_1x_0), i =1,2$ and lower for the common shape parameter of the Weibull distibution.
%Note that robust estimators tend to estimate slightly higher scale and shape parameters than the classical MLE. 
In addition, approximate CIs for the model parameters related to the scale, $a_0$ and $a_1,$ are not too wide, indicating low variance of the estimators. 
The CI of the shape parameter $\eta$ includes the value of 1, and therefore we could not reject the null hypothesis of exponential lifetime distributions, at a confidence level of $95\%$. However, the estimated value of the shape parameter is away from 1 and so the Weibull lifetime would be more appropriate.

\begin{table}[ht]
	\centering
	\caption{Estimated values and approximate $95\%$ CIs for the model parameters with solar light data for different values of $\beta$}
	\label{table:solarlight}
	\begin{tabular}{ccccccc}
		\hline
		$\beta  $ & $\widehat{a}_0^\beta$ & IC($a_0^\beta$) & $\widehat{a}_1^\beta$ & IC($a_1^\beta$) & $\widehat{\eta}^\beta$ & IC($\eta^\beta$) \\ 
		\hline
		0 & 1.80 & [1.70, 1.92] & -2.39 & [-3.07, -1.71] & 1.54 & [0.88, 2.19] \\ 
		 0.2 & 1.81 & [1.71, 1.92] & -2.38 & [-3.06, -1.70] & 1.50 & [0.86, 2.14] \\ 
		 0.4 & 1.82 & [1.71, 1.93] & -2.38 & [-3.06, -1.70] & 1.47 & [0.84, 2.10] \\ 
		 0.6 & 1.83 & [1.71, 1.95] & -2.37 & [-3.06, -1.69] & 1.44 & [0.82, 2.07] \\ 
		0.8 & 1.83 & [1.71, 1.96] & -2.37 & [-3.06, -1.69] & 1.42 & [0.80, 2.04] \\ 
		 1 & 1.84 & [1.71, 1.98] & -2.37 & [-3.06, -1.68] & 1.40 & [0.78, 2.02] \\ 
		\hline
	\end{tabular}
\end{table}

Table \ref{table:solarlight2} presents the estimated mean lifetime, reliability at $t=400$ hours and $95\%$ distribution quantile for the solar light data with different values of the tuning parameter. Robust methods yield larger mean lifetimes (and consequently, smaller  quantiles) than the classical MLE. However, all methods agree on the reliability of  devices at 400 hours.

\begin{table}[ht]
	\centering
	\caption{Estimated mean lifetime (in hundred hours), reliability at $t=400$ hours and $95\%$ distributional quantile (in hundred hours) at a constant stress level $x_0=293K$ for solar light devices for different values of $\beta$}
	\label{table:solarlight2}
	\begin{tabular}{rrrr}
		\hline
	$\beta $	& $\operatorname{E}_T(\widehat{\boldsymbol{\theta}}^\beta)$ &  $R_4(\widehat{\boldsymbol{\theta}}^\beta)$   &  $Q_{0.95}(\widehat{\boldsymbol{\theta}}^\beta)$  \\ 
		\hline
	0 & 5.468 & 0.591 & 0.877 \\ 
	 0.2 & 5.531 & 0.590 & 0.842 \\ 
	 0.4 & 5.585 & 0.589 & 0.814 \\ 
	 0.6 & 5.633 & 0.588 & 0.791 \\ 
	 0.8 & 5.679 & 0.588 & 0.770 \\ 
	 1 & 5.717 & 0.587 & 0.752 \\ 
		\hline
	\end{tabular}
\end{table}

\subsection{Effect of temperature on Light Emitting Diodes (LEDs)}

%effect of temperature on Light Emitting Diodes (LEDs)
%A general accelerated life model for step-stress testing Wenbiao Zhao a & Elsayed A. Elsayed
Finally, we report the performance of the proposed estimators in the dataset \ref{Dataset3}.
The data are transformed to intermittent inspection set-up and the stress levels are normalized as usual. 

Table \ref{table:LED} presents the estimated coefficients for different values of the tuning parameter $\beta.$ 
Here, the shape parameter is moderately high, pointing out the appropriateness of using the Weibull lifetime distribution, instead of the exponential distribution. However, we must mention that approximate CIs are wide and that the value of 1 is included in the approximate CI for the shape parameter.
Table \ref{table:LED2} contains the estimated mean life, reliability at $t=2.88 \times 10^4$ hours, and $95\%$ quantile based on the MDPDE for different values of $\beta$ under normal operating temperature. 
In this case, the estimated mean time to failure based on MDPDE with positive values of $\beta$ are larger than those estimated by the MLE, implying higher reliability of the devices. The same trend is observed for the estimated reliability and quantiles. Note that the increase in these quantities gradually rises with 
$\beta$, aligning more closely with the results obtained by Zhao and Elsayed (2005).
However, it is important to note that these values reach their maximum at $\beta=0.6$ and start decreasing beyond this critical point. Based on our empirical recommendations for the DPD tuning parameter, results from MDPDE fitted with a tuning parameter ranging from 0.4 to 0.6 provide a suitable trade-off between robustness and efficiency. Beyond this range, the efficiency may be insufficient. Thus, our results obtained with values in the range of 0.4-0.6 are, from our point of view, more reliable and, moreover, align with other studies in the literature.

\begin{table}[ht]
	\centering
	\caption{Estimated values and approximate $95\%$ CIs for the model parameters with LEDs data for different values of $\beta$.}
	\label{table:LED}
	\resizebox{0.5\textwidth}{!}{
	\begin{tabular}{ccccccc}
		\hline
		$\beta $ & $\widehat{a}_0^\beta$ & IC($a_0^\beta$) & $\widehat{a}_1^\beta$ & IC($a_1^\beta$) & $\widehat{\eta}^\beta$ & IC($\eta^\beta$) \\ 
		\hline
	0 & 10.09 & [1.04, 98.18] & -4.89 & [-33.56, 23.77] & 1.88 & [0.00, 8.79] \\ 
	 0.2 & 10.09 & [1.03, 99.18] & -4.89 & [-33.68, 23.90] & 1.88 & [0.00, 8.80] \\ 
	 0.4 & 10.09 & [1.01, 100.77] & -4.89 & [-33.88, 24.10] & 1.88 & [0.00, 8.85] \\ 
 	0.6 & 10.40 & [0.91, 119.20] & -5.25 & [-36.57, 26.08] & 1.79 & [0.00, 8.53] \\ 
	 0.8 & 10.17 & [0.97, 106.46] & -4.97 & [-34.67, 24.73] & 1.90 & [0.00, 8.95] \\ 
	 1 & 10.15 & [0.97, 105.93] & -4.94 & [-34.55, 24.67] & 1.93 & [0.00, 9.10] \\ 
		\hline
	\end{tabular}
}
\end{table}

\begin{table}[ht]
	\centering
	\caption{Estimated mean lifetime (in hours $\times 10^4$), reliability at $t=2.88 \times 10^4$ hours and $95\%$ distributional quantile (in hours $\times 10^4$) at normal operating temperature $x_0=50\circ C$ for LEDs data with different values of $\beta$}
	\label{table:LED2}
	\begin{tabular}{rrrr}
		\hline
	$\beta  $	& $\operatorname{E}_T(\widehat{\boldsymbol{\theta}}^\beta)$ &  $R_{2.88}(\widehat{\boldsymbol{\theta}}^\beta)$   &  $Q_{0.95}(\widehat{\boldsymbol{\theta}}^\beta)$  \\ 
		\hline
	0 & 2.146 & 0.249 & 0.499 \\ 
		 0.2 & 2.138 & 0.247 & 0.495 \\ 
	 0.4 & 2.143 & 0.248 & 0.499 \\ 
		 0.6 & 2.908 & 0.451 & 0.623 \\ 
	0.8 & 2.308 & 0.297 & 0.544 \\ 
		 1  & 2.264 & 0.283 & 0.547 \\ 
		\hline
	\end{tabular}
\end{table}

\section{Concluding remarks \label{sec:conclusions}}

In this paper, we have developed robust estimators and Wald-type test statistics for testing general composite null hypothesis based on the popular density power divergence (DPD) approach. We have examined the robustness properties of the proposed estimators and the test statistics theoretically as well as empirically, showing the clear improvement in terms of robustness with a small loss of efficiency in the absence of contamination.
Further, point estimation and approximate confidence intervals (CIs) for the model parameters and certain lifetime characteristics of interest based on the MDPDEs are derived. Their robustness is empirically examined in terms of accuracy for the point estimation and coverage probability for the CIs. Both direct and transformed approximate CIs are compared, with the transformed CIs demonstrating superior performance, particularly in heavily contaminated scenarios.
Finally, three real datasets from the reliability engineering field have been analyzed to illustrate the use of the proposed robust estimators and test statistics in practical situations.

	The MDPDEs are presented as a robust generalizacion of the MLE, and so their performance is only compared with this classical estimator.
	However, there exist some other appealing proposal for inference.
	For instance, an approach that could be of potential use is 
	Fisher's fiducial inferential method, as discussed recently by Chen et al. (2016).  Though this method has some appealing features such as not requiring an exact specification of a prior distribution as the Bayesian method does, it faces the problem of non-uniqueness in the absence of a complete sufficient statistic.  However, it does lead to useful inference based on approximate pivot as described in the elaborate article by Hannig (2009) amply demonstrates the practical utility of this approach.  Not much work has gone into studying the robustness features and misspecification effects on this method of inference.  In our future work, we plan to investigate this specific issue connected with the usage of fiducial method to develop inference for the problem considered here and then compare its relative performance with that of the methodology developed here.

%	\end{tabular}
%\end{table}
\section{Acknowledgments}
\noindent This research is supported by the Spanish Grants PGC2018-095 194-B-100 (M.Jaenada and L.Pardo) and FPU 19/01824 (M.Jaenada). M.Jaenada and L.Pardo are members of the Interdisciplinary Mathematics Institute (IMI).

\newpage
\appendix[Proofs of the main results]
\section*{Proof of the Result \ref{thm:esitmatingequations}}

\begin{proof}
	As the MDPDE is a minimizer, it must satisfy the equation
	\begin{equation*}
		\begin{aligned}
		\frac{\partial d_{\beta}\left( \widehat{\boldsymbol{p}},\boldsymbol{\pi}\left(\boldsymbol{\theta}\right)\right)}{\partial \boldsymbol{\theta}}   &= (\beta+1)\sum_{j=1}^{L+1} \left(  \pi_j(\boldsymbol{\theta})^{\beta-1} \left(\pi_j(\boldsymbol{\theta}) -  \widehat{p}_j\right) \frac{\partial \pi_j(\boldsymbol{\theta})}{\partial \boldsymbol{\theta}}  \right)\\
		& = \boldsymbol{0}_3.
	\end{aligned}
	\end{equation*}
	where
	\begin{equation*}
		\frac{\partial\pi_j(\boldsymbol{\theta})}{\partial \boldsymbol{\theta}}
		= \frac{\partial G_T(t_j)}{\partial \boldsymbol{\theta}} -\frac{G_T(t_{j-1})}{\partial \boldsymbol{\theta}}.
	\end{equation*}
	%	Additionally, we define the quantity $a_{i-1}^\ast = \lambda_{i-1} \sum_{l=1}^{i-1}\left(\frac{z_l-z_{l-1}}{\lambda_l}\right)(x_{i-1}-x_l).$
	Now, setting as $x_i$ the stress level at which the units are tested before the $\tau_i-$th inspection time, with $\tau_{i-1} < t_j \leq \tau_i,$ we have
	%Next, the derivative of the probability of success $\pi_j(\boldsymbol{\theta})$
	%at each interval $(t_{j-1}, t_j],$ $j=1,...,L$ depends on the stress level at which the device is submitted. 
	\begin{align*}
		\boldsymbol{z}_j = \frac{\partial G_T(t_j)}{\partial \boldsymbol{\theta}} &= g_T(t_j)\begin{pmatrix}
			-(t_j+h_{i-1})\\
			-(t_j+h_{i-1})x_i + h_{i-1}^\ast \\
			\log\left(\frac{t_j+h_{i-1}}{\alpha_i}\right)\frac{t_j+h_{i-1}}{\eta}
		\end{pmatrix}
	\end{align*}
	%Note that we can assume that the devices are submitted to the same stress level throughout the entire interval, which we denoted by $x_i$ for some $i=1,...,k.$
	with
	$h_{i-1}^\ast$ as in (\ref{aast}).	Defining the matrix $\boldsymbol{W}$ with rows $\boldsymbol{w}_j =  \boldsymbol{z}_j-\boldsymbol{z}_{j-1},$ we obtain the desired result.
\end{proof}

\section*{Proof of the Result \ref{thm:asymptoticstat}}
\begin{proof}
	Let $\boldsymbol{\theta}_0$ be the true value of the model parameter satisfying the null hypothesis in (\ref{eq:null}).
	From Result \ref{thm:asymptoticestimator}, we know that
	$$ \sqrt{N}\left(\widehat{\boldsymbol{\theta}}^\beta - \boldsymbol{\theta}_0\right)\xrightarrow[N\rightarrow \infty]{L}\mathcal{N}\left(\boldsymbol{0}, \boldsymbol{J}_\beta^{-1}(\boldsymbol{\theta}_0)\boldsymbol{K}_\beta(\boldsymbol{\theta}_0)\boldsymbol{J}_\beta^{-1}(\boldsymbol{\theta}_0)\right)$$
	from which it follows that
	\begin{align*}
		\sqrt{N}\boldsymbol{m}(\widehat{\boldsymbol{\theta}}^\beta)\xrightarrow[N\rightarrow \infty]{L}\mathcal{N}\bigg(\boldsymbol{0}, \Sigma \bigg)
	\end{align*}
with
$$\Sigma = \boldsymbol{M}^T(\boldsymbol{\theta}_0)\boldsymbol{J}_\beta^{-1}(\boldsymbol{\theta}_0)\boldsymbol{K}_\beta(\boldsymbol{\theta}_0)\boldsymbol{J}_\beta^{-1}(\boldsymbol{\theta}_0)\boldsymbol{M}^T(\boldsymbol{\theta}_0)$$
	Now, the covariance matrix $\Sigma$  %$$\boldsymbol{M}^T(\boldsymbol{\theta}_0)\boldsymbol{J}_\beta^{-1}(\boldsymbol{\theta}_0)\boldsymbol{K}_\beta(\boldsymbol{\theta}_0)\boldsymbol{J}_\beta^{-1}(\boldsymbol{\theta}_0)\boldsymbol{M}^T(\boldsymbol{\theta}_0)$$
	 has maximum rank, and then  by transforming the variable $\sqrt{N}\boldsymbol{m}(\widehat{\boldsymbol{\theta}}^\beta),$ %$\sqrt{N}\left(\boldsymbol{m}^T\widehat{\boldsymbol{\theta}}^\beta - d\right)$
	%with the matrix $$\left(\boldsymbol{m}^T\boldsymbol{J}_\beta^{-1}(\widehat{\boldsymbol{\theta}}^\beta)\boldsymbol{K}_\beta(\widehat{\boldsymbol{\theta}}^\beta)\boldsymbol{J}_\beta^{-1}(\widehat{\boldsymbol{\theta}}^\beta)\boldsymbol{m}\right)^{-1/2}$$ 
	we obtain that
	$$ N\boldsymbol{m}^T\hspace{-0.05cm}(\widehat{\boldsymbol{\theta}}^\beta\hspace{-0.1cm})\left(\boldsymbol{M}^T\hspace{-0.05cm}(\boldsymbol{\theta}_0)\boldsymbol{J}_\beta^{-1}(\boldsymbol{\theta}_0)\boldsymbol{K}_\beta(\boldsymbol{\theta}_0)\boldsymbol{J}_\beta^{-1}(\boldsymbol{\theta}_0)\boldsymbol{M}(\boldsymbol{\theta}_0)\right)^{-1}\hspace{-0.3cm}\boldsymbol{m}(\widehat{\boldsymbol{\theta}}^\beta\hspace{-0.1cm})$$
	follows a chi-squared distribution with $r$ degrees of freedom.
	% \rightarrow \chi^2_r.$$
	As $\widehat{\boldsymbol{\theta}}^\beta$ is a consistent estimator of $\boldsymbol{\theta}_0,$ the stated result follows.
\end{proof}

\section*{Proof of the Result \ref{thmasymptoticpower}}
\begin{proof}
	A first-order Taylor expansion of $\ell(\boldsymbol{\theta}) = \ell^\ast(\boldsymbol{\theta},\boldsymbol{\theta}^\ast)$ at $\widehat{\boldsymbol{\theta}}^\beta$ around $\boldsymbol{\theta}^\ast$ gives
	$$\ell(\widehat{\boldsymbol{\theta}}^\beta) - \ell(\boldsymbol{\theta}^\ast) =  \frac{\partial \ell(\boldsymbol{\theta})}{\partial \boldsymbol{\theta}^T} \bigg|_{\boldsymbol{\theta} = \boldsymbol{\theta}^\ast} 
	\left(\widehat{\boldsymbol{\theta}}^\beta - \boldsymbol{\theta}^\ast\right) +  o_p\left(\parallel \widehat{\boldsymbol{\theta}}^\beta - \boldsymbol{\theta}^\ast \parallel \right) $$
	and therefore, from Result \ref{thm:asymptoticestimator}, it follows that
	$$\sqrt{N}\left( \ell(\widehat{\boldsymbol{\theta}}^\beta) - \ell(\boldsymbol{\theta}^\ast) \right) \rightarrow \mathcal{N}\left(0, \sigma_{W_N(\boldsymbol{\theta}^\ast)} \right).
	$$
	%	From the asymptotic distribution of the MDPDE stated in Result \ref{thm:asymptoticestimator}, we have that
	%	$$ N\boldsymbol{m}^T(\widehat{\boldsymbol{\theta}}^\beta)\left(\boldsymbol{M}^T(\boldsymbol{\theta}^\ast)\boldsymbol{J}_\beta^{-1}(\boldsymbol{\theta}^\ast)\boldsymbol{K}_\beta(\boldsymbol{\theta}^\ast)\boldsymbol{J}_\beta^{-1}(\boldsymbol{\theta}^\ast)\boldsymbol{M}(\boldsymbol{\theta}^\ast)\right)^{-1}\boldsymbol{m}(\widehat{\boldsymbol{\theta}}^\beta) \rightarrow \chi^2_r(\lambda).$$
	%	where $\chi^2_r(\lambda)$ is a non-central chi-square with $r$ degrees of freedom and  non-centrality parameter  $\lambda= \boldsymbol{m}^T(\boldsymbol{\theta}^\ast)\left(\boldsymbol{M}^T(\boldsymbol{\theta}^\ast)\boldsymbol{J}_\beta^{-1}(\boldsymbol{\theta}^\ast)\boldsymbol{K}_\beta(\boldsymbol{\theta}^\ast)\boldsymbol{J}_\beta^{-1}(\boldsymbol{\theta}^\ast)\boldsymbol{M}(\boldsymbol{\theta}^\ast)\right)^{-1}\boldsymbol{m}(\boldsymbol{\theta}^\ast).$ 
	Further, as $\widehat{\boldsymbol{\theta}}^\beta$ is a consistent estimator of the true parameter $\boldsymbol{\theta}^\ast$, we have 
	$$\sqrt{N}\left( W_N(\widehat{\boldsymbol{\theta}}^\beta)  - \ell(\boldsymbol{\theta}^\ast) \right) \xrightarrow[N\rightarrow \infty]{L} \mathcal{N}\left(0, \sigma_{W_N(\boldsymbol{\theta}^\ast)} \right).
	$$
	%	$$ W_N(\widehat{\boldsymbol{\theta}}^\beta) \rightarrow \chi^2_r(\lambda^\beta),$$
	%	with $\lambda^\beta$ defined in Result (\ref{thmasymptoticpower}).
	Now, the power function is the probability of rejection given the critical region (\ref{eq:criticalregionWtest}), and so we have 
	\begin{align*}
		\beta_N\left(\boldsymbol{\theta}^\ast\right) &= \mathbb{P}\left(W_{N}(\widehat{\boldsymbol{\theta}}^\beta) > \chi_{r,\alpha}^2 | \boldsymbol{\theta} = \boldsymbol{\theta}^\ast \right)\\
		& =1- \Phi\left(\frac{\sqrt{N}}{\sigma_{W_N(\boldsymbol{\theta}^\ast)} } \left(\frac{\chi^2_{r,\alpha}}{N}- \ell^\ast(\boldsymbol{\theta}^\ast, \boldsymbol{\theta}^\ast) \right)\right).
	\end{align*}
\end{proof}

\section*{Proof of the Result \ref{thm:IF}}

\begin{proof}
	Let us denote  $\boldsymbol{F}_{\boldsymbol{\theta},\varepsilon} = (1-\varepsilon)F_{\boldsymbol{\theta}_0} + \varepsilon \Delta_{\boldsymbol{n}}$ for the perturbed distribution function with mass function $\boldsymbol{\pi}_{\varepsilon}(\boldsymbol{\theta}),$ and  $\boldsymbol{\theta}_\varepsilon = \boldsymbol{T}_\beta(\boldsymbol{F}_{\boldsymbol{\theta},\varepsilon}).$ 
	The MDPDE satisfies the estimating equations
	%	\begin{equation}\label{eq:H}
		%		H_{\beta}\left( g,\boldsymbol{\pi}\left(\boldsymbol{\theta}\right)\right)   = \sum_{j=1}^{L+1} \left(\pi_j(\boldsymbol{\theta})^{1+\beta} -\left( 1+\frac{1}{\beta}\right) g_j\pi_j(\boldsymbol{\theta})^{\beta} + \frac{1}{\beta}g_j^{1+\beta} \right)
		%	\end{equation}
	\begin{equation} \label{eq:estimatingH}
		\sum_{j=1}^{L+1} \pi_j(\boldsymbol{\theta}_\varepsilon)^{\beta-1}\left[(\pi_j(\boldsymbol{\theta}_\varepsilon)-\boldsymbol{\pi}_{\varepsilon,j}(\boldsymbol{\theta}))\frac{\partial \pi_j(\boldsymbol{\theta}_\varepsilon)}{\partial \boldsymbol{\theta}}\right] = \boldsymbol{0}.
	\end{equation}
	Now, differentiating  (\ref{eq:estimatingH}), we get
	\resizebox{0.5\textwidth}{!}{
		\begin{minipage}{0.5\textwidth}
	\begin{align*}
		\sum_{j=1}^{L+1} &(\beta-1)\pi_j(\boldsymbol{\theta}_\varepsilon)^{\beta-2}\frac{\partial \pi_j(\boldsymbol{\theta}_\varepsilon)}{\partial \boldsymbol{\theta}}\frac{\partial \boldsymbol{\theta}_\varepsilon}{\partial \varepsilon}\left[(\pi_j(\boldsymbol{\theta}_\varepsilon)-\boldsymbol{\pi}_{\varepsilon,j}(\boldsymbol{\theta}))\frac{\partial \pi_j(\boldsymbol{\theta}_\varepsilon)}{\partial \boldsymbol{\theta}}\right]\\
		& + \pi_j(\boldsymbol{\theta}_\varepsilon)^{\beta-1}\bigg[\left(\frac{\partial \pi_j(\boldsymbol{\theta}_\varepsilon)}{\partial \boldsymbol{\theta}}\frac{\partial \boldsymbol{\theta}_\varepsilon}{\partial \varepsilon}-\frac{\boldsymbol{\pi}_{\varepsilon,j}(\boldsymbol{\theta})}{\partial \varepsilon} \right) \frac{\partial \pi_j(\boldsymbol{\theta}_\varepsilon)}{\partial \boldsymbol{\theta}}\\
		& + \left(\pi_j(\boldsymbol{\theta}_\varepsilon)-\boldsymbol{\pi}_{\varepsilon,j}(\boldsymbol{\theta})\right) \frac{\partial^2 \pi_j(\boldsymbol{\theta}_\varepsilon)}{\partial \boldsymbol{\theta}^2} \frac{\partial \boldsymbol{\theta}_\varepsilon}{\partial \varepsilon}\bigg] = \boldsymbol{0}.
	\end{align*}
	\end{minipage}
}
	Then, using the fact that $\boldsymbol{\pi}_{0}(\boldsymbol{\theta}) = \boldsymbol{\pi}(\boldsymbol{\theta}_0)$ and evaluating at $\varepsilon = 0$, we obtain
	\begin{equation*}
		\begin{aligned}
		\sum_{j=1}^{L+1} \pi_j(\boldsymbol{\theta}_0)^{\beta-1}&\bigg[\left(\frac{\partial \pi_j(\boldsymbol{\theta}_0)}{\partial \boldsymbol{\theta}}\right)^2\text{IF}\left(\boldsymbol{n}, \boldsymbol{T}_\beta, \boldsymbol{F}_{\boldsymbol{\theta}}\right)\\
		& -\left(\frac{\partial \pi_j(\boldsymbol{\theta}_0)}{\partial \boldsymbol{\theta}}\right)(-\pi_j(\boldsymbol{\theta}_0) + \Delta_{\boldsymbol{n}})   \bigg] = \boldsymbol{0}.
		\end{aligned}
	\end{equation*}
	Now, rewriting this equation in matrix form, we get
	$$ \boldsymbol{W}^T \boldsymbol{D}_{\boldsymbol{\pi}(\boldsymbol{\theta_0})}^{\beta-1} \boldsymbol{W} \cdot\text{IF}\left(\boldsymbol{n}, \boldsymbol{T}_\beta, \boldsymbol{F}_{\boldsymbol{\theta}} \right) - \boldsymbol{W}^T \boldsymbol{D}_{\boldsymbol{\pi}(\boldsymbol{\theta}_0)}^{\beta-1}\left(-\boldsymbol{\pi}(\boldsymbol{\theta}_0)+\Delta_{\boldsymbol{n}} \right)$$
	and finally solving for $\text{IF}\left(\boldsymbol{n}, \boldsymbol{T}_\beta, \boldsymbol{F}_{\boldsymbol{\theta}} \right),$ we obtain the stated expression.
\end{proof}

\appendix

\section{Additional Numerical Results}

\subsection{Performance of the MDPDEs under increasing sample size }

	Accelerated life-tests with highly reliable devices often deal with very small sample sizes, as is the case of our illustrative examples. However, the performance of the proposed estimators and especially their associated approximate confindence intervals will naturally depends on the sample size; larger sample sizes result in more accurate estimators and asymptotic confidence intervals.
	To illustrate such improvement, Figure \ref{fig:Npure} shows the MSE
	of the MDPDEs with different values of the tuning parameter $\beta$ for increasing sample size in the absence of contamination. As seen, all estimators improves their performance when increasing the sample size, and that difference gets reduced for larger sample sizes, where all estimators (including the optimal estimator obtained with Algorithm 1) present quite similar estimation errors.

	\begin{figure}
	\centering
	\includegraphics[scale=0.34]{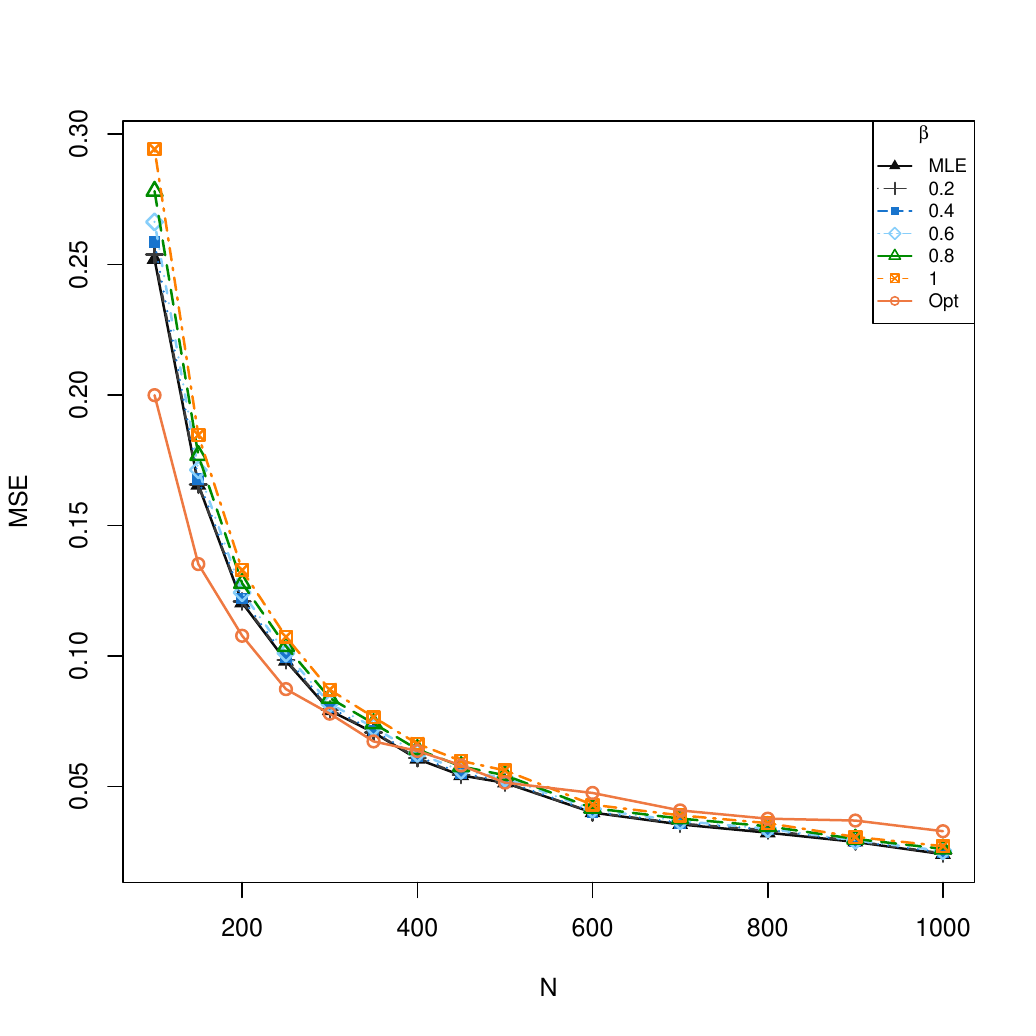}
	\caption{Mean Squared Error of the MDPDEs with different values of the tuning parameter $\beta$ for increasing sample size in the absence of contamination}
	\label{fig:Npure}
\end{figure}

	However, the robustness advantage of the MDDPE with positive values of $\beta$ compared to the MLE is maintained even for very large sample sizes. Figure \ref{fig:Ncont} illustrates the performance of the proposed robust estimators with increasing sample sizes in the presence of contamination in every of the three directions considered: contamination in the first parameter with $\widetilde{a}_0=6.3$ all the other two parameters remain constant  (top), second parameter with $\widetilde{a}_1=-0.03$ (middle) and third parameter with $\widetilde{\eta} = 2$ (bottom).
	As expected, the estimation error decreases with increasing sample size. However, the relative performance between MDPDEs with different values of the tuning parameter $\beta$ remains consistent across all sample sizes, demonstrating the enduring importance of robustness even for large samples. Furthermore, it is noteworthy that the optimal estimator obtained using Algorithm 1 outperforms the MLE and performs similarly to MDPDEs with moderate values of the tuning parameter, over 0.4.

	\begin{figure}
		\centering
		\includegraphics[scale=0.34]{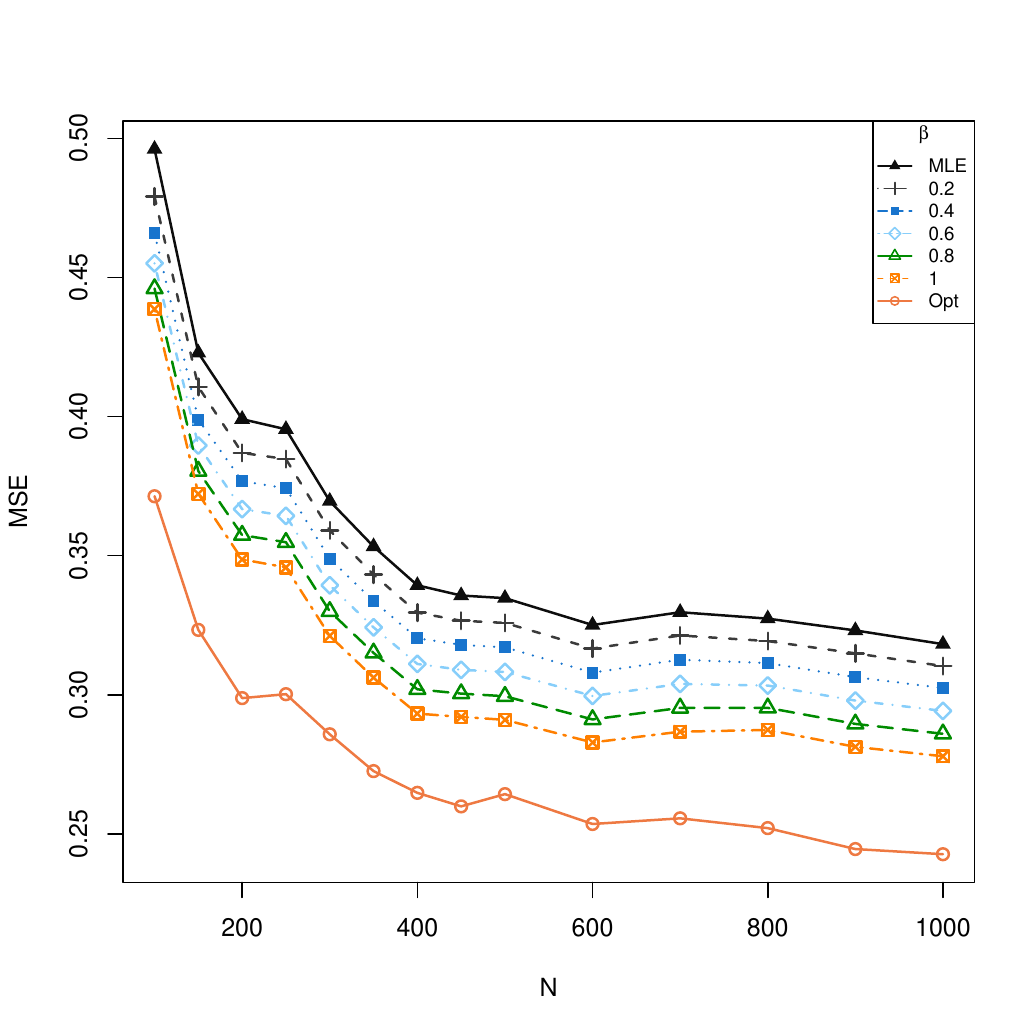}
		\includegraphics[scale=0.34]{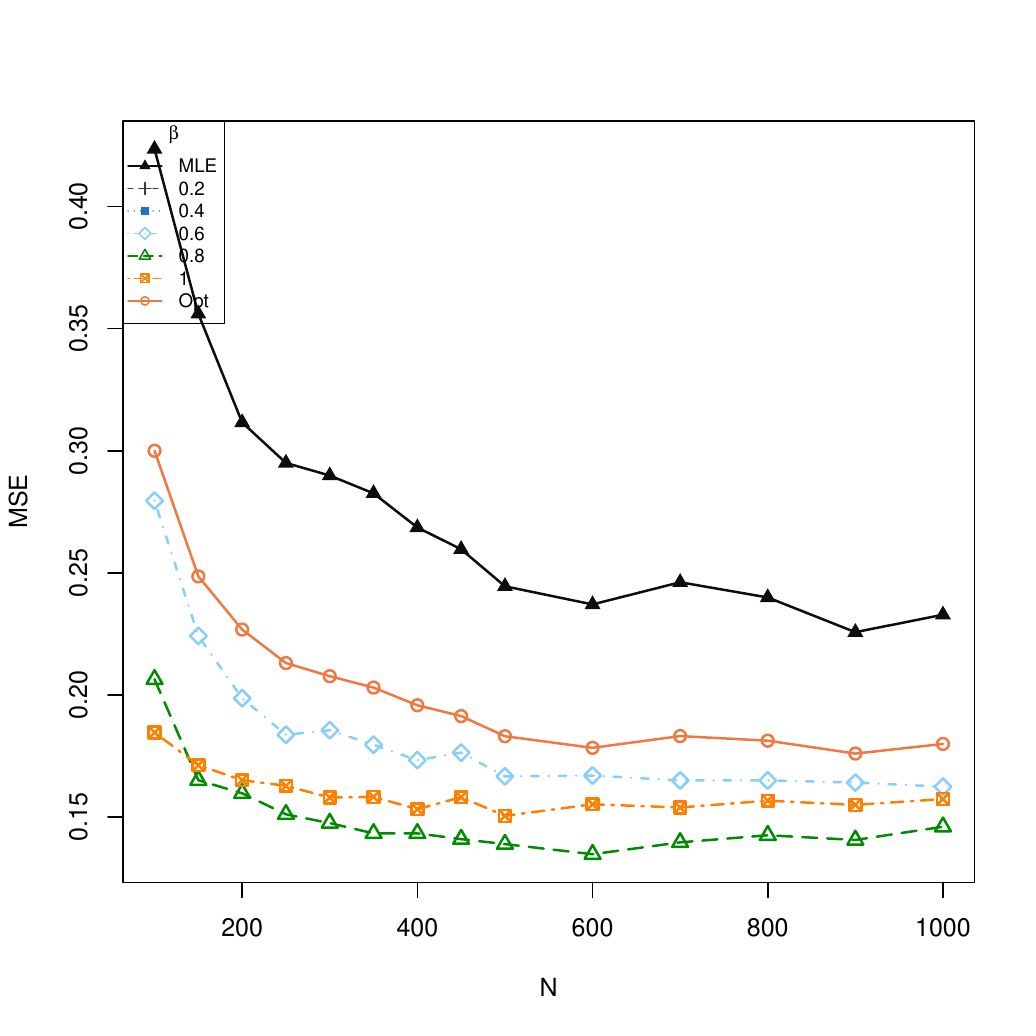}
		\includegraphics[scale=0.34]{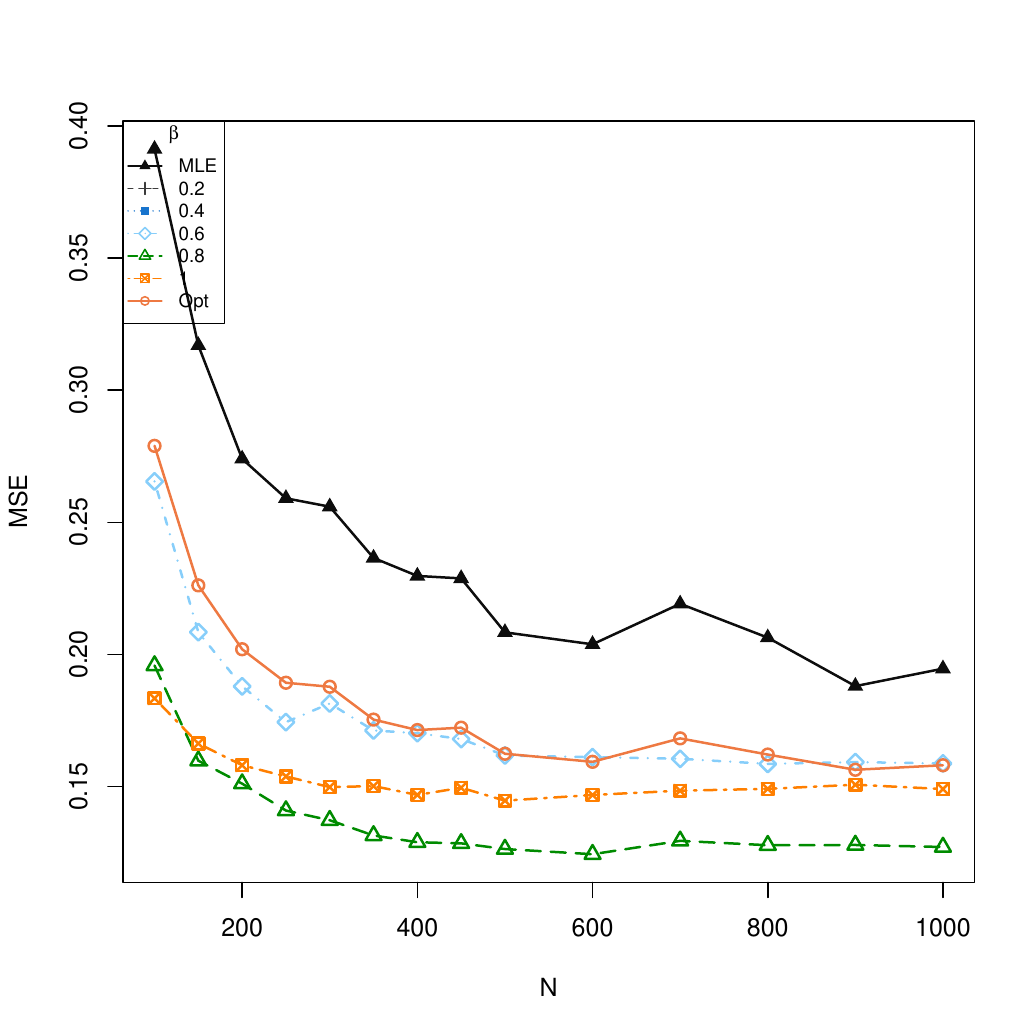}
		\caption{Mean Squared Error of the MDPDEs with different values of the tuning parameter $\beta$ for increasing sample size in the presence of contamination in the first parameter (top), second (middle) and third (bottom).}
		\label{fig:Ncont}
	\end{figure}

\subsection{Mean Squared Error of the mean lifetime to failure estimate under contamination \label{app:MSEmean}}

For completeness of the results, we include in this Section the empirical errors obtained when estimating the mean time to failure for the simulation study in Section \ref{sec:simstudy} under normal operating conditions $x_0=20$ when the contamination is introduced on the first (top), second (middle) and third (bottom) model parameters. 
The mean time to failure is significantly high, and the differences in the accuracies of the MDPDE are not evident in the characteristic estimation. Instead, all estimators appear to perform similarly, and their performance deteriorates with contamination.

\begin{figure}
	\centering
	%	\begin{tabular}{cc}
		%\includegraphics[height=5.7cm, width=6.5cm]{MSEreliability_a0cont}
		%\includegraphics[height=5.7cm, width=6.5cm]{MSEreliability_a1cont}
		\includegraphics[height=5.7cm, width=6.3cm]{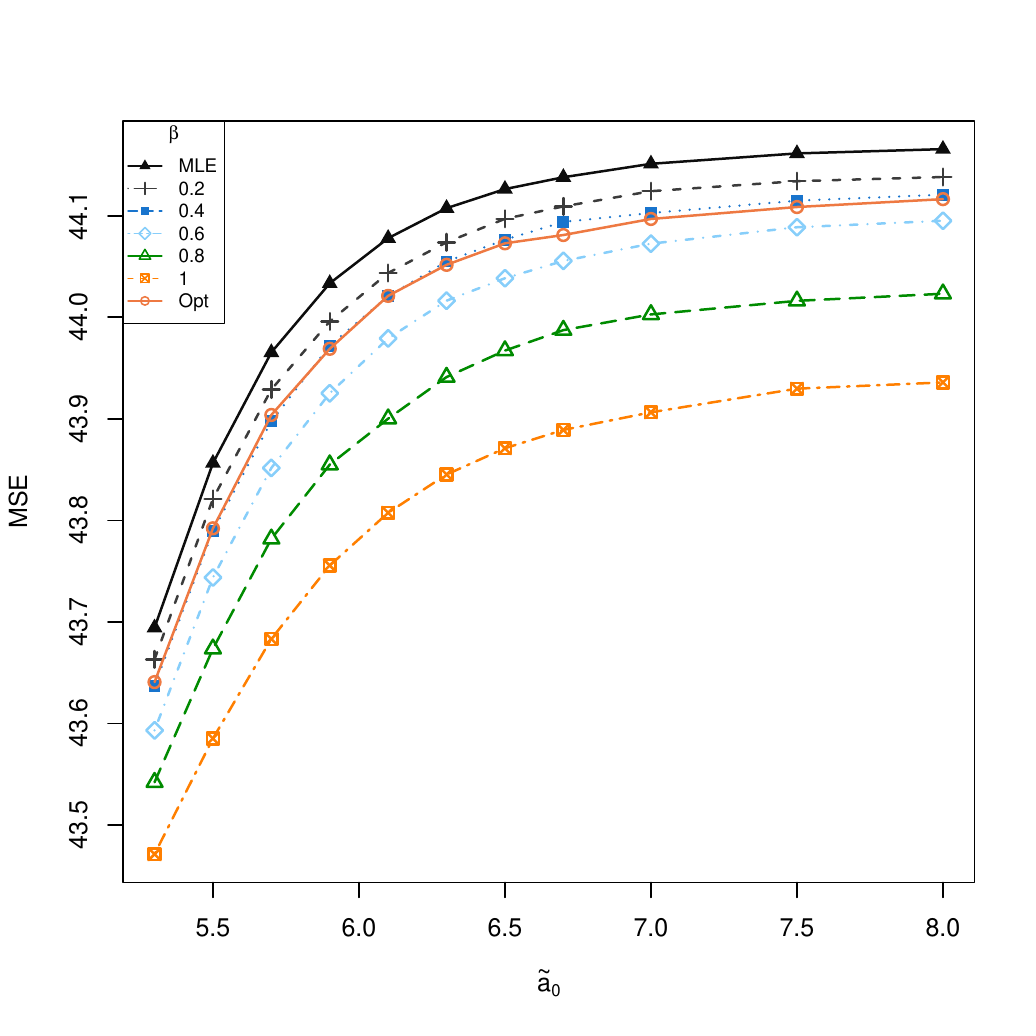}	
		\includegraphics[height=5.7cm, width=6.3cm]{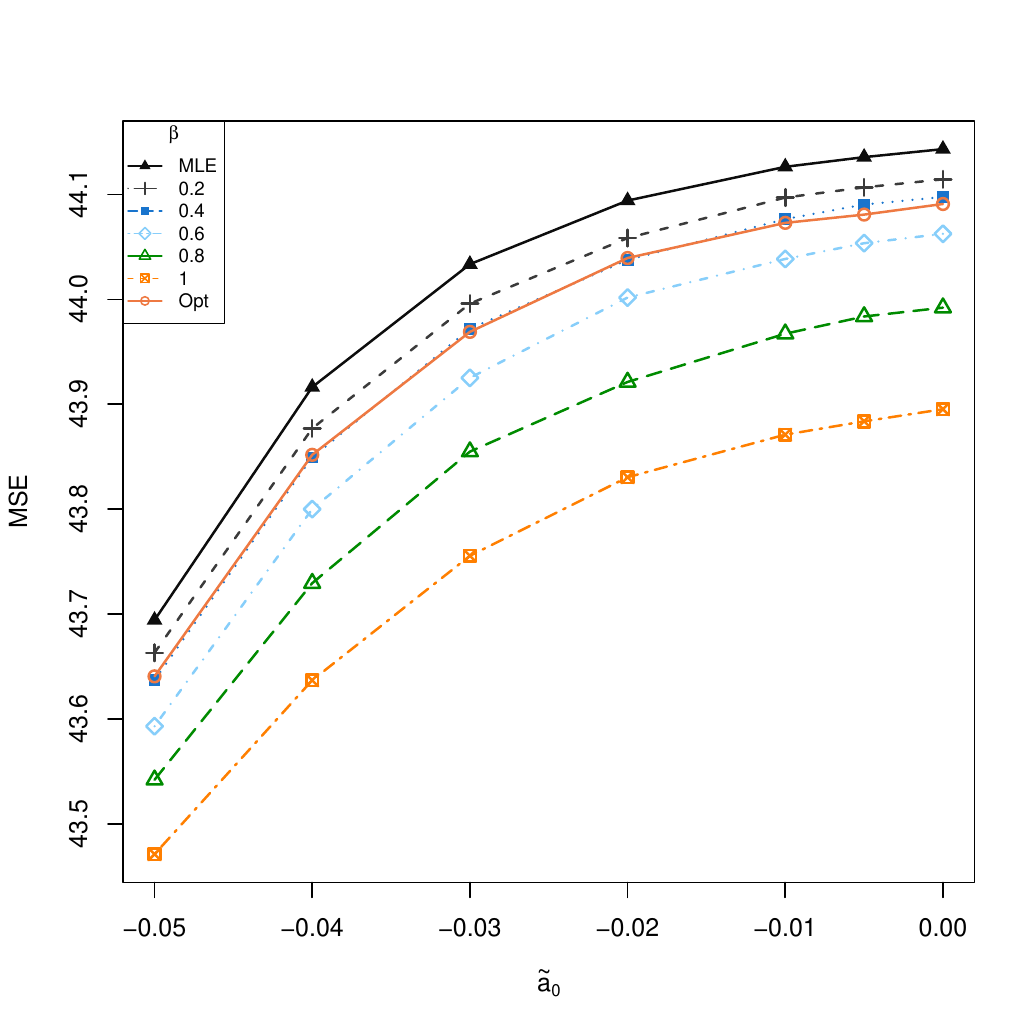}
		\includegraphics[height=5.7cm, width=6.3cm]{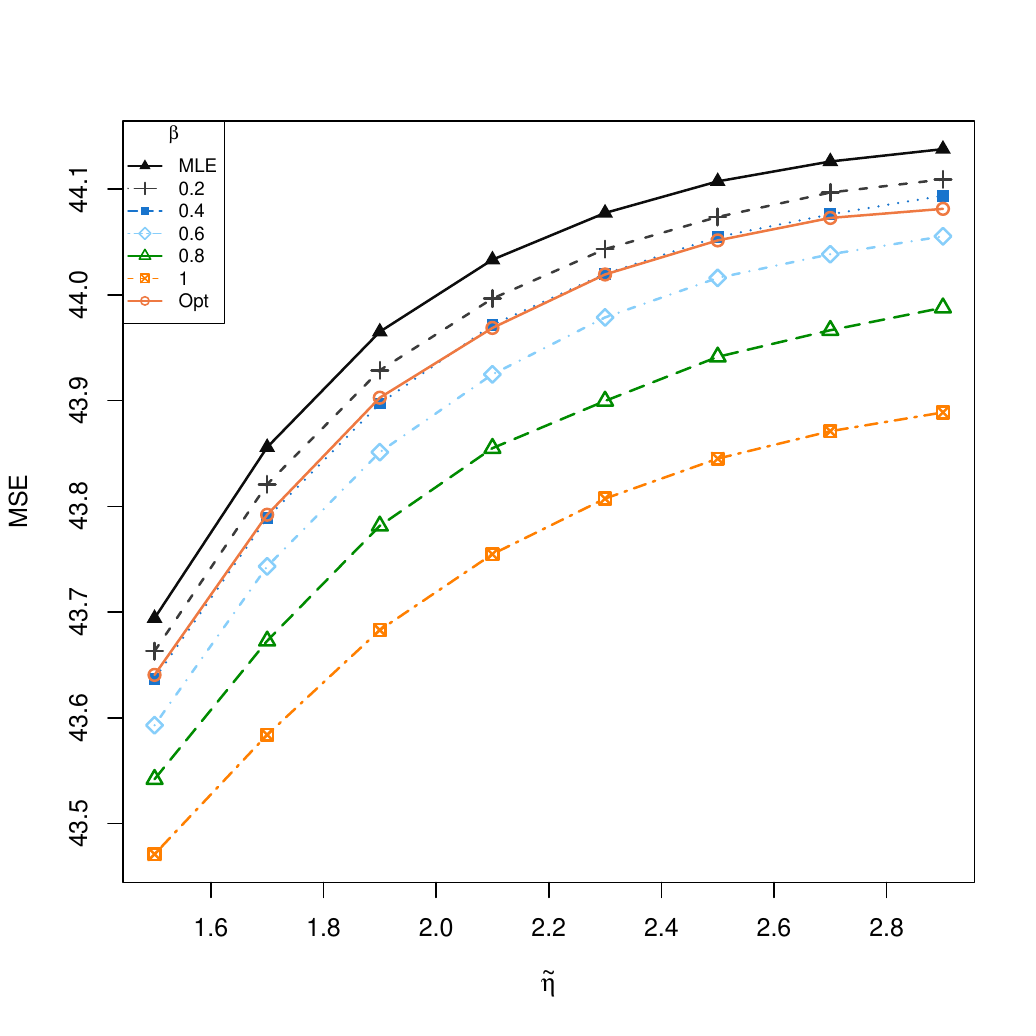}	
		%	\end{tabular}
	\caption{MSE of the mean lifetime ($\times 10^{2})$  under normal operating conditions $x_0=20$ when the contamination is introduced on the first (top) and second (bottom) model parameters.}	
	\label{fig:meancont}
\end{figure}

\subsection{Coverage probabilities of the asymptotic and direct CI for some lifetime characteristics}

We finally present in this section some additional numerical results on the coverage probabilities of the direct and transfomed confidence intervals for the mean time to failure and the reliability of the product at a certain mission time at a different nominal level, $\alpha = 0.1$. The results are quite similar to those for the nominal level $\alpha = 0.5$ presented in the main manuscript, which is why we have decided to present them as supplementary material in this appendix.

%\color{blue}

\begin{table}[ht]
	\centering
	\caption{Empirical coverage of direct and transformed asymptotic confidence intervals at 90\% for the mean time to failure and reliability and mission time $t_0=60$ under normal operating condition $x_0=20$ and $a_0$-contamination  }
	\label{table:coverage-a0cont}
	\resizebox{0.5\textwidth}{!}{
	\begin{tabular}{rrrrrrrrrr}
		\hline
		\multicolumn{10}{c}{Mean Time To Failure}\\
		\hline
		\hline
	$\widetilde{a}_0$ & 5.3 & 5.5 & 5.7 & 5.9 & 6.1 & 6.3 &6.5 & 6.7 & 7 \\ 
		\hline
		$\beta$ & \multicolumn{9}{c}{Direct}\\
		\hline
		0 & 0.794 & 0.692 & 0.514 & 0.380 & 0.302 & 0.236 & 0.206 & 0.178 & 0.156 \\ 
		0.2 & 0.788 & 0.696 & 0.520 & 0.398 & 0.310 & 0.238 & 0.206 & 0.180 & 0.160 \\ 
		0.4 & 0.802 & 0.704 & 0.522 & 0.402 & 0.322 & 0.246 & 0.210 & 0.188 & 0.164 \\ 
		0.6 & 0.790 & 0.694 & 0.526 & 0.428 & 0.332 & 0.262 & 0.218 & 0.190 & 0.162 \\ 
		0.8 & 0.778 & 0.696 & 0.534 & 0.434 & 0.342 & 0.268 & 0.222 & 0.196 & 0.172 \\ 
		1 & 0.774 & 0.696 & 0.540 & 0.442 & 0.348 & 0.270 & 0.226 & 0.202 & 0.172 \\ 
		opt & 0.796 & 0.698 & 0.520 & 0.396 & 0.312 & 0.234 & 0.204 & 0.172 & 0.154 \\ 
		\hline 
			$\beta$ & \multicolumn{9}{c}{Transformed}\\
		\hline
		0.1 & 0.946 & 0.864 & 0.768 & 0.638 & 0.550 & 0.478 & 0.410 & 0.378 & 0.330 \\ 
		0.2 & 0.946 & 0.866 & 0.764 & 0.644 & 0.558 & 0.486 & 0.424 & 0.386 & 0.340 \\ 
		0.4 & 0.948 & 0.878 & 0.768 & 0.658 & 0.566 & 0.500 & 0.438 & 0.398 & 0.346 \\ 
		0.6 & 0.948 & 0.882 & 0.780 & 0.666 & 0.582 & 0.504 & 0.446 & 0.422 & 0.356 \\ 
		0.8 & 0.950 & 0.892 & 0.786 & 0.686 & 0.590 & 0.516 & 0.468 & 0.434 & 0.376 \\ 
		1 & 0.956 & 0.896 & 0.788 & 0.698 & 0.600 & 0.532 & 0.482 & 0.446 & 0.394 \\ 
		opt & 0.950 & 0.872 & 0.764 & 0.648 & 0.568 & 0.490 & 0.422 & 0.390 & 0.342 \\ 
 		\hline
		 \multicolumn{10}{c}{Reliability at $t_0=60$}\\
		 \hline
		 	$\beta$ & \multicolumn{9}{c}{Direct}\\
		 \hline
		0 & 0.794 & 0.692 & 0.514 & 0.380 & 0.302 & 0.236 & 0.206 & 0.178 & 0.156 \\ 
		0.2 & 0.788 & 0.696 & 0.520 & 0.398 & 0.310 & 0.238 & 0.206 & 0.180 & 0.160 \\ 
		0.4 & 0.802 & 0.704 & 0.522 & 0.402 & 0.322 & 0.246 & 0.210 & 0.188 & 0.164 \\ 
		0.6 & 0.790 & 0.694 & 0.526 & 0.428 & 0.332 & 0.262 & 0.218 & 0.190 & 0.162 \\ 
		0.8 & 0.778 & 0.696 & 0.534 & 0.434 & 0.342 & 0.268 & 0.222 & 0.196 & 0.172 \\ 
		1 & 0.774 & 0.696 & 0.540 & 0.442 & 0.348 & 0.270 & 0.226 & 0.202 & 0.172 \\ 
		opt & 0.796 & 0.698 & 0.520 & 0.396 & 0.312 & 0.234 & 0.204 & 0.172 & 0.154 \\
			\hline 
			$\beta$ & \multicolumn{9}{c}{Transformed}\\
		\hline 
		0.1 & 0.912 & 0.874 & 0.764 & 0.636 & 0.558 & 0.482 & 0.420 & 0.380 & 0.324 \\ 
		0.2 & 0.906 & 0.876 & 0.770 & 0.656 & 0.562 & 0.492 & 0.428 & 0.390 & 0.336 \\ 
		0.4 & 0.908 & 0.874 & 0.774 & 0.670 & 0.582 & 0.498 & 0.450 & 0.406 & 0.344 \\ 
		0.6 & 0.908 & 0.880 & 0.782 & 0.684 & 0.598 & 0.506 & 0.458 & 0.414 & 0.360 \\ 
		0.8 & 0.908 & 0.898 & 0.800 & 0.700 & 0.606 & 0.534 & 0.482 & 0.440 & 0.374 \\ 
		1 & 0.904 & 0.902 & 0.804 & 0.726 & 0.624 & 0.564 & 0.496 & 0.458 & 0.392 \\ 
		opt & 0.912 & 0.876 & 0.774 & 0.672 & 0.564 & 0.490 & 0.438 & 0.398 & 0.340 \\ 
		\hline
	\end{tabular}
}
\end{table}

\begin{table}[ht]
	\centering
	\caption{Empirical coverage of direct and transformed asymptotic confidence intervals at 90\% for the mean time to failure and reliability and mission time $t_0=60$ under normal operating condition $x_0=20$ and $a_1$-contamination  }
	\label{table:coverage-a1cont}
	\begin{tabular}{rrrrrrrr}
		\hline
		\multicolumn{8}{c}{Mean Time To Failure}\\
		\hline
		\hline
		$\widetilde{a}_1 $ &-0.05 & -0.04 & -0.03 & -0.02 & -0.01 & -0.005 &  0.0  \\ 
		\hline
		$\beta$ & \multicolumn{7}{c}{Direct}\\
		\hline
		0 & 0.794 & 0.598 & 0.380 & 0.260 & 0.206 & 0.186 & 0.164 \\ 
		0.2 & 0.788 & 0.602 & 0.398 & 0.258 & 0.206 & 0.186 & 0.164 \\ 
		0.4 & 0.802 & 0.606 & 0.402 & 0.266 & 0.210 & 0.192 & 0.172 \\ 
		0.6 & 0.790 & 0.616 & 0.428 & 0.282 & 0.218 & 0.192 & 0.174 \\ 
		0.8 & 0.778 & 0.624 & 0.434 & 0.296 & 0.222 & 0.198 & 0.184 \\ 
		1 & 0.774 & 0.632 & 0.442 & 0.298 & 0.226 & 0.202 & 0.192 \\ 
		opt & 0.796 & 0.610 & 0.396 & 0.258 & 0.204 & 0.180 & 0.158 \\ 
		\hline 
		$\beta$ & \multicolumn{7}{c}{Transformed}\\
		\hline
		0.1 & 0.946 & 0.820 & 0.638 & 0.508 & 0.410 & 0.382 & 0.360 \\ 
		0.2 & 0.946 & 0.816 & 0.644 & 0.520 & 0.424 & 0.390 & 0.370 \\ 
		0.4 & 0.948 & 0.820 & 0.658 & 0.534 & 0.438 & 0.402 & 0.376 \\ 
		0.6 & 0.948 & 0.830 & 0.666 & 0.550 & 0.446 & 0.426 & 0.404 \\ 
		0.8 & 0.950 & 0.826 & 0.686 & 0.552 & 0.468 & 0.440 & 0.414 \\ 
		1 & 0.956 & 0.832 & 0.698 & 0.552 & 0.482 & 0.450 & 0.426 \\ 
		opt & 0.950 & 0.818 & 0.648 & 0.524 & 0.422 & 0.394 & 0.372 \\ 
		\hline
		 \multicolumn{8}{c}{Reliability at $t_0=60$}\\
		\hline 
		$\beta$ & \multicolumn{7}{c}{Direct}\\
		\hline
		0 & 0.794 & 0.598 & 0.380 & 0.260 & 0.206 & 0.186 & 0.164 \\ 
		0.2 & 0.788 & 0.602 & 0.398 & 0.258 & 0.206 & 0.186 & 0.164 \\ 
		0.4 & 0.802 & 0.606 & 0.402 & 0.266 & 0.210 & 0.192 & 0.172 \\ 
		0.6 & 0.790 & 0.616 & 0.428 & 0.282 & 0.218 & 0.192 & 0.174 \\ 
		0.8 & 0.778 & 0.624 & 0.434 & 0.296 & 0.222 & 0.198 & 0.184 \\ 
		1 & 0.774 & 0.632 & 0.442 & 0.298 & 0.226 & 0.202 & 0.192 \\ 
		opt & 0.796 & 0.610 & 0.396 & 0.258 & 0.204 & 0.180 & 0.158 \\ 
		\hline 
		$\beta$ & \multicolumn{7}{c}{Transformed}\\
		\hline
		0.1 & 0.912 & 0.820 & 0.636 & 0.506 & 0.420 & 0.380 & 0.362 \\ 
		0.2 & 0.906 & 0.824 & 0.656 & 0.518 & 0.428 & 0.392 & 0.376 \\ 
		0.4 & 0.908 & 0.838 & 0.670 & 0.532 & 0.450 & 0.414 & 0.388 \\ 
		0.6 & 0.908 & 0.846 & 0.684 & 0.542 & 0.458 & 0.420 & 0.402 \\ 
		0.8 & 0.908 & 0.850 & 0.700 & 0.562 & 0.482 & 0.442 & 0.418 \\ 
		1 & 0.904 & 0.852 & 0.726 & 0.584 & 0.496 & 0.460 & 0.442 \\ 
		opt & 0.912 & 0.836 & 0.672 & 0.524 & 0.438 & 0.404 & 0.386 \\ 
		\hline
		\end{tabular}
	\end{table}

\begin{table}[ht]
	\centering
	\caption{Empirical coverage of direct and transformed asymptotic confidence intervals at 90\% for the mean time to failure and reliability and mission time $t_0=60$ under normal operating condition $x_0=20$ and $\eta$-contamination }
	\label{table:coverage-etacont}
	\begin{tabular}{rrrrrrrrr}
			\hline
			\multicolumn{9}{c}{Mean Time To Failure}\\
			\hline
			\hline
			$\widehat{\eta}$	& 1.5	&1.7	&1.9	&2.1	&2.3	&2.5	&2.7	&2.9 \\ 
			\hline
			$\beta$ & \multicolumn{8}{c}{Direct}\\
			\hline
			0 & 0.794 & 0.692 & 0.514 & 0.380 & 0.302 & 0.238 & 0.206 & 0.178 \\ 
			0.2 & 0.788 & 0.696 & 0.520 & 0.398 & 0.310 & 0.240 & 0.206 & 0.180 \\ 
			0.4 & 0.802 & 0.704 & 0.522 & 0.402 & 0.322 & 0.246 & 0.212 & 0.188 \\ 
			0.6 & 0.790 & 0.694 & 0.526 & 0.428 & 0.332 & 0.262 & 0.220 & 0.190 \\ 
			0.8 & 0.778 & 0.696 & 0.534 & 0.434 & 0.342 & 0.268 & 0.224 & 0.196 \\ 
			1 & 0.774 & 0.696 & 0.540 & 0.442 & 0.348 & 0.270 & 0.228 & 0.202 \\ 
			opt & 0.796 & 0.698 & 0.520 & 0.396 & 0.312 & 0.236 & 0.204 & 0.172 \\ 
			\hline 
			$\beta$ & \multicolumn{8}{c}{Transformed}\\
			\hline
			0.1 & 0.946 & 0.864 & 0.768 & 0.638 & 0.552 & 0.478 & 0.410 & 0.378 \\ 
			0.2 & 0.946 & 0.866 & 0.764 & 0.644 & 0.560 & 0.486 & 0.426 & 0.386 \\ 
			0.4 & 0.948 & 0.878 & 0.768 & 0.658 & 0.568 & 0.504 & 0.440 & 0.398 \\ 
			0.6 & 0.948 & 0.882 & 0.780 & 0.666 & 0.584 & 0.508 & 0.450 & 0.422 \\ 
			0.8 & 0.950 & 0.892 & 0.786 & 0.686 & 0.592 & 0.520 & 0.472 & 0.436 \\ 
			1 & 0.956 & 0.896 & 0.788 & 0.698 & 0.602 & 0.534 & 0.486 & 0.446 \\ 
			opt & 0.950 & 0.872 & 0.764 & 0.648 & 0.570 & 0.490 & 0.424 & 0.390 \\ 
			\hline
			 \multicolumn{9}{c}{Reliability at $t_0=60$}\\
			 \hline
			$\beta$ & \multicolumn{8}{c}{Direct}\\
			\hline
			0 & 0.794 & 0.692 & 0.514 & 0.380 & 0.302 & 0.238 & 0.206 & 0.178 \\ 
			0.2 & 0.788 & 0.696 & 0.520 & 0.398 & 0.310 & 0.240 & 0.206 & 0.180 \\ 
			0.4 & 0.802 & 0.704 & 0.522 & 0.402 & 0.322 & 0.246 & 0.212 & 0.188 \\ 
			0.6 & 0.790 & 0.694 & 0.526 & 0.428 & 0.332 & 0.262 & 0.220 & 0.190 \\ 
			0.8 & 0.778 & 0.696 & 0.534 & 0.434 & 0.342 & 0.268 & 0.224 & 0.196 \\ 
			1 & 0.774 & 0.696 & 0.540 & 0.442 & 0.348 & 0.270 & 0.228 & 0.202 \\ 
			opt & 0.796 & 0.698 & 0.520 & 0.396 & 0.312 & 0.236 & 0.204 & 0.172 \\
			\hline 
			$\beta$ & \multicolumn{8}{c}{Transformed}\\
			\hline 
			0.1 & 0.912 & 0.874 & 0.764 & 0.636 & 0.560 & 0.482 & 0.424 & 0.380 \\ 
			0.2 & 0.906 & 0.876 & 0.770 & 0.656 & 0.564 & 0.494 & 0.430 & 0.390 \\ 
			0.4 & 0.908 & 0.874 & 0.774 & 0.670 & 0.584 & 0.502 & 0.452 & 0.408 \\ 
			0.6 & 0.908 & 0.880 & 0.782 & 0.684 & 0.600 & 0.510 & 0.460 & 0.416 \\ 
			0.8 & 0.908 & 0.898 & 0.800 & 0.700 & 0.608 & 0.538 & 0.482 & 0.440 \\ 
			1 & 0.904 & 0.902 & 0.804 & 0.726 & 0.626 & 0.566 & 0.496 & 0.458 \\ 
			opt & 0.912 & 0.876 & 0.774 & 0.672 & 0.566 & 0.494 & 0.440 & 0.398 \\ 
			\hline
			
		\end{tabular}
\end{table}

\end{document}